\documentclass[11pt, reqno]{amsart}

\usepackage[margin = 1in]{geometry}

\usepackage{amsmath, amssymb, amsthm}
\usepackage[normalem]{ulem}
\usepackage{mathtools}
\usepackage{microtype}
\usepackage{graphicx}
\usepackage{bbm}
\usepackage{stmaryrd}

\usepackage{tikz}
\usetikzlibrary{arrows.meta,backgrounds,calc,fit}

\allowdisplaybreaks
\usepackage[dvipsnames]{xcolor}
\usepackage[colorlinks = true, linkcolor = MidnightBlue, citecolor = MidnightBlue, urlcolor = MidnightBlue]{hyperref}

\usepackage[backend = biber, style = alphabetic, maxnames = 5, url = false, doi = false, isbn = false]{biblatex}
\numberwithin{equation}{section}

\newcommand{\rank}{\mathrm{rank}}

\newcommand{\eps}{\varepsilon}
\newcommand{\Tr}{\mathrm{Tr}}
\newcommand{\E}{\mathbb{E}}

\newcommand{\Z}{\mathbb{Z}}

\newcommand{\R}{\mathbb{R}}
\newcommand{\C}{\mathbb{C}}
\newcommand{\cc}{{\rm C}}

\newcommand{\h}{{\mathcal H}}
\newcommand{\s}{{\rm S}}
\renewcommand{\r}{{\rm R}}
\renewcommand{\b}{{\mathcal B}}
\newcommand{\e}{{\rm E}}
\newcommand{\one}{{\mathbf 1}}
\renewcommand{\L}{{\mathcal L}}
\renewcommand{\P}{\mathcal{P}}
\renewcommand{\S}{\mathcal S}
\newcommand{\q}{\mathcal Q}
\newcommand{\hh}{{\mathfrak h}}
\newcommand{\btheta}{\boldsymbol{\theta}}
\newcommand{\bm}{\boldsymbol{m}}
\newcommand{\br}{\boldsymbol{r}}
\newcommand{\bs}{\boldsymbol{s}}
\newcommand{\bomega}{\boldsymbol{\omega}}
\newcommand{\ve}{{\varepsilon}}
\newcommand{\cw}{{C_{\rm w}}}

\definecolor{Collective}{RGB}{104, 78, 153} % Purple
\definecolor{Coarse}{RGB}{198, 116, 36}     % Amber
\definecolor{Local}{RGB}{26, 132, 126}      % Teal
\definecolor{BlueGrey}{RGB}{78, 108, 132}   % Blue-grey
\definecolor{Edge}{RGB}{145,149,153}
\definecolor{Ink}{RGB}{38, 42, 46}
\definecolor{Bad}{RGB}{184,72,72}

\tikzset{
    >=Latex,
    every node/.style={text=Ink},
    vertex/.style={circle,draw=Ink,line width=.68pt,fill=white,minimum size=4.7mm,inner sep=0pt},
    sint/.style={draw=Edge,line width=.67pt},
    cedge/.style={draw=Collective,line width=.82pt,opacity=.88},
    medge/.style={draw=Coarse,line width=.82pt,opacity=.88},
    ledge/.style={draw=Local,line width=.82pt,opacity=.90},
    paneltitle/.style={font=\bfseries\large,align=center},
    panelsub/.style={font=\footnotesize,align=center},
    regionlabel/.style={font=\bfseries\scriptsize},
    rlabel/.style={font=\scriptsize,text=Ink!100},
    row/.style={font=\small,align=center,inner sep=0pt,outer sep=0pt},
    staterow/.style={font=\footnotesize,align=center,inner sep=0pt,outer sep=0pt},
    refarrow/.style={{Latex[length=2.2mm,width=1.55mm]}-{Latex[length=2.2mm,width=1.55mm]},draw=Ink!45,line width=.75pt}
}

\tikzset{
    sedge/.style={draw=Local,line width=.82pt,opacity=.90},
    cNedge/.style={draw=Collective,line width=.82pt,opacity=.88},
    redge/.style={draw=Coarse,line width=.82pt,opacity=.88},
    flowarrow/.style={-{Latex[length=2.2mm,width=1.55mm]},draw=Ink!45,line width=.75pt},
     formula/.style={font=\small,align=center,inner sep=0pt,outer sep=0pt},
    tinylabel/.style={font=\scriptsize},
    sub/.style={font=\footnotesize}
}

\newcommand{\Lrow}[1]{%
  \node[row,anchor=center] at (0,2.72) {#1};%
}
\newcommand{\StateMain}[1]{%
  \node[font=\scriptsize,align=center,inner sep=0pt,outer sep=0pt,anchor=center]
    at (0,1.72) {#1};%
}

\newcommand{\systemblob}{%
  \path[fill=BlueGrey,fill opacity=.16,draw=BlueGrey,draw opacity=.32,
        line width=.55pt]
    (-2.12,-.08)
    .. controls (-2.03,.77) and (-1.38,1.02) .. (-.64,.79)
    .. controls (-.03,1.10) and (.82,.96) .. (1.33,.72)
    .. controls (1.92,.83) and (2.22,.27) .. (2.02,-.34)
    .. controls (1.83,-.91) and (.97,-.91) .. (.29,-.72)
    .. controls (-.44,-.94) and (-1.39,-.87) .. (-1.88,-.54)
    .. controls (-2.08,-.41) and (-2.19,-.24) .. cycle;
}

\newcommand{\systemgraph}[1]{%
  \node[vertex] (#1s1) at (-1.50,.08) {};
  \node[vertex] (#1s2) at (-.73,-.36) {};
  \node[vertex] (#1s3) at (.05,.22) {};
  \node[vertex] (#1s4) at (.88,-.27) {};
  \node[vertex] (#1s5) at (1.51,.35) {};
  \draw[sint] (#1s1)--(#1s2);
  \draw[sint] (#1s1)--(#1s3);
  \draw[sint] (#1s2)--(#1s3);
  \draw[sint] (#1s2)--(#1s4);
  \draw[sint] (#1s3)--(#1s4);
  \draw[sint] (#1s3)--(#1s5);
  \draw[sint] (#1s4)--(#1s5);
}

\newcommand{\Sgraph}[1]{%
    \begin{scope}[on background layer]
        \path[fill=BlueGrey,fill opacity=.16,draw=BlueGrey,draw opacity=.32,line width=.50pt]
        (-.82,-.45)..controls(-.90,.08)and(-.55,.52)..(0,.55)
        ..controls(.57,.57)and(.88,.18)..(.80,-.39)
        ..controls(.38,-.67)and(-.37,-.67)..cycle;
    \end{scope}
 
    \node[vertex](#1a)at(-.43,.08){};\node[vertex](#1b)at(.02,-.18){};\node[vertex](#1c)at(.43,.18){};
    \draw[draw=Ink!34,line width=.55pt](#1a)--(#1b)--(#1c)--(#1a);
}

\newcommand{\snap}[6]{%
    \begin{scope}[shift={(#1,#2)}]
    \Sgraph{#6}
        \foreach \i/\px/\py in {1/-1.10/.90,2/-.72/1.32,3/-.28/1.62,4/-1.35/1.52,5/-.85/1.83}{%
            \ifnum\i>#3\relax\else
                \node[vertex,draw=Collective] (#6p\i) at (\px,\py) {};
                \draw[draw=Collective,line width=.72pt] (#6p\i)--(#6a);
        \fi}
 
        \foreach \i/\qx/\qy in {1/1.10/.90,2/.72/1.32,3/.28/1.62,4/1.35/1.52,5/.85/1.83}{%
            \ifnum\i>#4\relax\else
                \node[vertex,draw=Coarse] (#6q\i) at (\qx,\qy) {};
                \draw[draw=Coarse,line width=.72pt] (#6q\i)--(#6c);
        \fi}
 
        \node[tinylabel,anchor=north] at (0,-.82) {$(\nu_1,\nu_2,\varepsilon_K)=(#3,#4,#5)$};
    \end{scope}
}

\usepackage[capitalize, noabbrev, nameinlink]{cleveref}
\usepackage{thmtools}
\crefname{proposition}{Proposition}{Propositions}
\crefname{definition}{Definition}{Definitions}
\crefname{corollary}{Corollary}{Corollaries}
\crefname{theorem}{Theorem}{Theorems}
\crefname{example}{Example}{Examples}
\crefname{remark}{Remark}{Remarks}
\crefname{lemma}{Lemma}{Lemmas}
\crefname{equation}{eq.}{eqs.}
\crefname{discussion}{Discussion}{Discussions}

\theoremstyle{plain}
\newtheorem{theorem}{Theorem}[section]
\newtheorem{proposition}[theorem]{Proposition}
\newtheorem{corollary}[theorem]{Corollary}
\newtheorem{lemma}[theorem]{Lemma}

\theoremstyle{definition}
\newtheorem{definition}[theorem]{Definition}
\newtheorem{example}[theorem]{Example}
\newtheorem{remark}[theorem]{Remark}
\newtheorem{discussion}[theorem]{Discussion}

\newcommand{\sss}{{\zeta}}

\title{De Finetti structures \\
in networks of quantum systems}

\author{Gavin Hull}
\address{Department of Mathematics and Statistics, Memorial University of Newfoundland, St.~John's, NL, Canada}
\curraddr{David R. Cheriton School of Computer Science, University of Waterloo, Waterloo, Ontario, Canada}
\email{g2hull@uwaterloo.ca}

\author{Marco Merkli}
\address{Department of Mathematics and Statistics, Memorial University of Newfoundland, St.~John's, NL, Canada}
\email{merkli@mun.ca}
\thanks{Both authors were supported by a Discovery Grant from NSERC, the Natural Sciences and Engineering Research Council of Canada.}

\date{}

\begin{document}

\begin{abstract}
Consider $N$ identical quantum systems at the vertices of a graph, interacting via pairwise edge couplings. We designate $K$ vertices as a system of interest $\s$ and the remaining $N-K$ as a reservoir $\r$. Under generic connectivity assumptions between $\s$ and $\r$, we show that for $N\to\infty$, the reduced state of $\s$ converges to a convex combination of symmetric product states of $K$ single-vertex density matrices |  called a de Finetti representation. Unlike in standard derivations for symmetric (exchangeable) systems, we establish de Finetti representations for quantum networks, which are inherently non-exchangeable. We construct the representations explicitly and quantify the convergence rate.  We show that  Erd\H os--R\'enyi random graphs exhibit the de Finetti structure:  as the reservoir grows, the reduced state of any fixed finite subgraph converges almost surely to an explicit deterministic product state. As applications, we establish a general area law for the entanglement entropy of $\s$, analyze mean-field and fluctuation scalings of the $\s$-$\r$ interaction, and illustrate our results on qubit networks. We also derive quantitative convergence estimates for information-theoretic quantities, including von Neumann entropy, mutual information, and entanglement within $\s$ and between $\s$ and $\r$.
\end{abstract}

\maketitle
%\tableofcontents

\section{Introduction}
\medskip

The Hilbert space of a complex consisting of $N$ identical quantum systems, each described by a Hilbert space $\hh$, is given by 
\begin{equation}
\label{m1'}
\h_N = \bigotimes _{n=1}^N \hh.
\end{equation} 
A state of the complex is a density matrix $\rho_N$, that is, a non-negative trace-class operator on $\h_N$ normalized as $\Tr(\rho_N)=1$. Classical mixtures of product states of the form
\begin{align}
\label{in3}
\int_{\Omega} \sigma_1(\alpha)\otimes\cdots\otimes\sigma_N(\alpha) d\mu(\alpha),
\end{align}
where $\mu$ is a probability measure on a space $\Omega$ and each $\sigma_j(\alpha)$ is a state on $\hh$, are called separable states. States which cannot be written as such superpositions are called entangled. If $\mu$ is a Dirac point mass then the corresponding state \eqref{in3} is a product $\sigma_1\otimes\cdots\otimes\sigma_N$, representing a physical situation where the subsystems are entirely uncorrelated. More general probability measures $\mu$ may introduce classical correlations in the state. From a dynamical viewpoint, even if a product initial state is prepared, interactions between the subsystems generate entanglement among them.

The reduced state, or marginal, of a group of $K$ among the $N$ subsystems is described by a  density matrix acting on $K$ of the factors of $\h_N$; after relabelling, these may be taken as the first $K$ factors in \eqref{m1'}. The reduced state is obtained by taking the partial trace over the remaining factors,
\begin{align}
\label{in4}
\rho_N^{(K)}  =\Tr_{K+1,\ldots,N} \big(\rho_N\big).  
\end{align}
Suppose now that the size $N$ of the total complex increases,
while we observe a fixed subsystem of size $K$. For each
$N> K$, let $\rho_N$ be a density matrix on $\h_N$ and
let $\rho_N^{(K)}=\Tr_{K+1,\ldots,N}(\rho_N)$ denote its
$K$-site marginal. We consider what happens to $\rho_N^{(K)}$ as $N\rightarrow\infty$. In other words: {\it what is the structure of a state of a fixed subsystem of a very large total complex?}
\smallskip

We answer this question for a network of $N$ identical quantum systems labeled by vertices of a graph, starting out in a state $\rho^{\otimes N}=\rho\otimes\cdots\otimes\rho$, where each subsystem is in the same single-vertex state $\rho$. The subsystems subsequently communicate with each other via pairwise interaction in accordance with the edges of the graph (see Figure \ref{fig:setup}). Every edge between two vertices $i,j$ describes an interaction implemented by an operator $A_i A_j = \mathbf 1\otimes \cdots \otimes\mathbf 1\otimes A\otimes \mathbf 1\otimes\cdots\otimes\mathbf 1\otimes A\otimes \mathbf 1\otimes \cdots\otimes \mathbf 1$, where the fixed single-vertex operator $A$ sits on the factors $i$ and $j$ in the product. This interaction gives rise to a Hamiltonian $H_N$, after which the state becomes 
\begin{align}
\label{in9}
\rho_N =e^{-i TH_N} (\rho^{\otimes N})e^{iTH_N},\qquad \textstyle H_N=\sum_{{\rm edges\, }(i,j) }A_i A_j,
\end{align}
where $T$ is the interaction time. Our focus is  the structure of the reduced state 
\begin{align}
\rho_N^{(K)}=\Tr_{K+1,\ldots, N}\big(e^{-i TH_N} (\rho^{\otimes N})e^{iTH_N}\big),
\end{align}
for large $N$ and fixed (but arbitrary) $K$, $T$. The model can be thought of as an open quantum system $\s$, where the $N-K$ sites represent a reservoir $\r$ to which the $K$ sites of $\s$ are coupled, and where intra-system and intra-reservoir couplings are permitted (see Figure \ref{fig:setup}). 

Put simply, our main finding is that under suitable connectivity conditions on the graph, we have
\begin{align}
\label{in5}
\lim_{N\rightarrow \infty}\rho_N^{(K)} = \int_{\mathcal S_1} \sigma^{\otimes K} d\mu(\sigma), 
\end{align}
where $\mu$ is a probability measure on $\mathcal S_1$, the set of states (density matrices) of a single-vertex system $\hh$. We construct the measure $\mu$  explicitly and give
bounds on the rates of convergence.

Expansions of reduced states in terms of a convex combination of symmetric product states, as given in \eqref{in5}, are a trademark of the Quantum de Finetti Theorem. This theorem deals with a sequence of compatible density matrices $\rho_N$, which are in addition {\it symmetric}, or {\it exchangeable}: for any permutation $\pi$ of $\{1,\ldots,N\}$ one has
\begin{align}
\label{01.3}
W_\pi \rho_N W_\pi^* = \rho_N,    
\end{align}
where $W_\pi$ is the unitary on $\h_N$ determined by $W_\pi(\psi_1\otimes \cdots \otimes \psi_N) = \psi_{\pi(1)}\otimes\cdots \otimes\psi_{\pi(N)}$, $\psi_j\in\hh$. The Quantum de Finetti Theorem asserts that in the limit $N\rightarrow\infty$, the state of the fixed subsystem is a classical superposition of product states with identical factors, expressed as \eqref{in5}. Related finite size Quantum de Finetti Theorems approximate
marginals of symmetric finite systems by such mixtures,
with errors controlled in terms of the ambient size. The precise meaning of \eqref{in5} is analyzed in the literature with various degrees of generality. It is treated in a $C^*$-algebraic setting in \cite{stormer1969symmetric, hudson1976locally}. For finite dimensional systems ($\dim\hh<\infty$) it is proved in \cite{Caves_2002, Lewin, Chiri}, including bounds on the convergence speed. A discussion in connection with quantum information theory is provided in \cite{Brandao,Chri, Chri2}.

While symmetric states play a central role in the quantum theory of bosonic many-body systems, where permutation invariance is intrinsic, they do not arise naturally in the context of quantum networks. Indeed, the state $\rho_N$, \eqref{in9} of the entire network after interaction, is not exchangeable, except for very special graphs and interactions (see also Example \ref{example 2.1}). Showing that de Finetti representations emerge in systems lacking exchangeability is one novelty of the present work. A further novel contribution is that we construct the de Finetti representations {\it explicitly}. This contrasts with existing proofs of the Quantum de Finetti Theorem, which are predominantly based on abstract arguments. The generality of our method and results is well illustrated by Erd\H os--R\'enyi random graphs, in which edges are placed between pairs of vertices randomly and independently, with arbitrary edge probability $0<p<1$. We show that the reduced state of any subsystem of $K$ vertices converges almost surely, as $N \rightarrow \infty$, to the $K$-fold tensor product of a deterministic, single-vertex density matrix $\tau$.

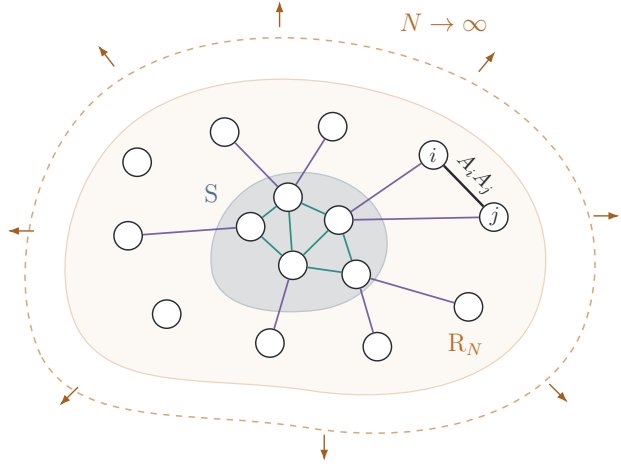
\begin{figure}[t]
    \centering

    \resizebox{0.5\textwidth}{!}{%
        \begin{tikzpicture}[x=1cm,y=1cm]
            %\node[paneltitle] at (0,4.78) {System--Reservoir network};
            %\node[panelsub,text=Ink!60] at (0,4.40) {fixed system, growing reservoir};

            \begin{scope}[on background layer]
                \path[fill=Coarse,fill opacity=.050,draw=Coarse,draw opacity=.31,line width=.58pt]
                (-4.05,-1.55)
                .. controls (-4.58,-.56) and (-4.16,1.25) .. (-2.98,2.02)
                .. controls (-1.55,2.95) and (.67,2.67) .. (1.94,2.03)
                .. controls (3.35,1.31) and (3.96,.14) .. (3.58,-1.10)
                .. controls (3.17,-2.43) and (1.37,-2.79) .. (-.20,-2.55)
                .. controls (-1.88,-2.30) and (-3.45,-2.58) .. (-4.05,-1.55) -- cycle;

                \path[draw=Coarse!64,dashed,line width=.72pt,fill=none]
                (-4.68,-1.96)
                .. controls (-5.28,-.64) and (-4.83,1.68) .. (-3.34,2.62)
                .. controls (-1.63,3.70) and (.92,3.37) .. (2.48,2.55)
                .. controls (4.11,1.68) and (4.82,.23) .. (4.38,-1.31)
                .. controls (3.90,-2.97) and (1.71,-3.48) .. (-.22,-3.17)
                .. controls (-2.19,-2.86) and (-4.03,-3.19) .. (-4.68,-1.96) -- cycle;
                
                \path[fill=BlueGrey,fill opacity=.18,draw=BlueGrey,draw opacity=.36,line width=.62pt]
                 (-1.83,-.62)
                 .. controls (-1.98,.10) and (-1.48,.88) .. (-.74,1.03)
                 .. controls (.06,1.20) and (.83,.80) .. (1.02,.18)
                 .. controls (1.22,-.51) and (.69,-1.10) .. (-.08,-1.21)
                 .. controls (-.87,-1.31) and (-1.68,-1.18) .. (-1.83,-.62) -- cycle;
            \end{scope}

            \node[vertex](s1)at(-1.20,.17){};
            \node[vertex](s2)at(-.58,.66){};
            \node[vertex](s3)at(-.50,-.47){};
            \node[vertex](s4)at(.26,.28){};
            \node[vertex](s5)at(.55,-.62){};
            \draw[sedge](s1)--(s2); 
            \draw[sedge](s1)--(s3); 
            \draw[sedge](s2)--(s3);
            \draw[sedge](s2)--(s4); 
            \draw[sedge](s3)--(s4); 
            \draw[sedge](s3)--(s5); 
            \draw[sedge](s4)--(s5);

            \node[vertex](r1)at(-3.08,1.24){};
            \node[vertex](r2)at(-1.63,1.74){};
            \node[vertex](r3)at(.16,1.83){};
            \node[vertex](r4)at(1.83,1.36){};
            \node[vertex](r5)at(2.83,.33){};
            \node[vertex](r6)at(2.41,-1.16){};
            \node[vertex](r7)at(.90,-1.82){};
            \node[vertex](r8)at(-.88,-1.76){};
            \node[vertex](r9)at(-2.59,-1.28){};
            \node[vertex](r10)at(-3.23,.02){};
            
            % \foreach \a/\b in {r1/r2,r2/r3,r3/r4,r4/r5,r5/r6,r6/r7,r7/r8,r8/r9,r9/r10,r10/r1,r1/r3,r3/r5,r5/r7,r7/r9,r9/r1}{\draw[redge](\a)--(\b);}
            % Representative pair interaction
            \draw[draw=Ink,line width=1.15pt] (r4)--(r5);
            
            \node[font=\small,anchor=south east]
                at ($(r4)+(.19,-.25)$) {$i$};
            
            \node[font=\small,anchor=south west]
                at ($(r5)+(-.2,-.27)$) {$j$};
            
            \node[
                font=\small,
                text=Ink,
                rotate=-45.9,
                inner sep=1pt
            ]
                at ($(r4)!0.7!(r5)+(0,.4)$)
                {$A_iA_j$};
            \foreach \a/\b in {s2/r2,s2/r3,s4/r4,s4/r5,s5/r6,s5/r7,s3/r8,s1/r10}{\draw[cedge](\a)--(\b);}

            \node[regionlabel,text=BlueGrey,anchor=east] at (-1.60,.72) {\large $\s$};
            \node[regionlabel,text=Coarse!92!black] at (2.36,-1.78) {\large $\r_N$};

            \draw[-{Latex[length=1.9mm,width=1.3mm]},draw=Coarse!82!black,line width=.65pt] (-3.46,2.83)--++(-.27,.35);
            \draw[-{Latex[length=1.9mm,width=1.3mm]},draw=Coarse!82!black,line width=.65pt] (-.72,3.48)--++(0,.44);
            \draw[-{Latex[length=1.9mm,width=1.3mm]},draw=Coarse!82!black,line width=.65pt] (2.58,2.74)--++(.28,.34);
            \draw[-{Latex[length=1.9mm,width=1.3mm]},draw=Coarse!82!black,line width=.65pt] (4.48,.35)--++(.44,0);
            \draw[-{Latex[length=1.9mm,width=1.3mm]},draw=Coarse!82!black,line width=.65pt] (3.75,-2.44)--++(.31,-.31);
            \draw[-{Latex[length=1.9mm,width=1.3mm]},draw=Coarse!82!black,line width=.65pt] (.02,-3.28)--++(0,-.44);
            \draw[-{Latex[length=1.9mm,width=1.3mm]},draw=Coarse!82!black,line width=.65pt] (-4.06,-2.49)--++(-.31,-.31);
            \draw[-{Latex[length=1.9mm,width=1.3mm]},draw=Coarse!82!black,line width=.65pt] (-4.79,.10)--++(-.44,0);
            \node[panelsub,text=Coarse!90!black] at (2,3.56) {\large $N\to\infty$};

         %   \node[formula] at (0,-4.8) {\large $E(N)={\color{Local}E_{\s}}\sqcup{\color{Collective}E_{{\rm C},N}}\sqcup{\color{Coarse}E_{\r,N}}$};
         %   \node[panelsub,text=Local] at (-2.72,-4.56) {within $\mathsf S$};
         %   \node[panelsub,text=Collective] at (0,-4.56) {\large between $\s$ and $\r_N$};
         %   \node[panelsub,text=Coarse] at (2.72,-4.56) {within $\r_N$};
        \end{tikzpicture}
    }

    \caption{{\it A network of quantum systems.} Each of the $N$ vertices of a graph hosts an identical quantum system (white discs). Every edge in the graph represents the interaction between the quantum systems located at the corresponding vertices. The interaction between vertices $i,j$ is of the form $A_i A_j$, where $A_i$ is the single-vertex operator $A$ translated to act on the vertex $i$. $K$ among the $N$ vertices are singled out to form an open system $\s$ which is connected to a reservoir $\r_N$ consisting of the remaining $N-K$ vertices. %The edge set splits into a disjoint union of edges within $\s$, connecting $\s$ and $\r$, and within $\r$. 
    As the graph size $N$ is increased, existing edges remain unchanged, while new ones may be added. We analyze the state of the fixed system $\s$ as $N\rightarrow\infty$.}
    \label{fig:setup}
\end{figure}

\subsection{Overview of the main results.} We present the basic results on the convergence of the reduced density matrix  in Section \ref{ch:conv}. In Theorem  \ref{thm2} we construct a completely positive trace preserving (CPTP) map $\P$ acting on bounded operators on $\h_K$, and we estimate the trace distance between $\rho_N^{(K)}$ and  the state $ \mathcal P \big[\rho^{\otimes K}\big]$, where we recall that $\rho$ is the single vertex state before interaction. We construct an explicit Kraus form of the map $\P$, which encodes the connectivity of the system $\s$ with the reservoir $\r$. The action of $\P$ can be interpreted as a non-selective quantum measurement, or equivalently the decoherence, associated to a commuting set of operators $\mathcal L$, called the boundary (or link) operators. $\L$ is determined by the edges linking vertices of $\s$ and $\r$ ($\s$-$\r$ interaction). If this interaction is sufficiently effective (a notion we make precise in \eqref{wellconn}), then our Corollary \ref{cor:deco} asserts that 
\begin{align}
\lim_{N\rightarrow\infty}  \rho_N^{(K)} = e^{-iTH_\s}\mathcal P\big[\rho^{\otimes K}\big]e^{iTH_\s},
\end{align}
where $H_\s$ is the system Hamiltonian, generated by all interactions between vertices within $\s$. The structure of $\mathcal P$ links to  the de Finetti representation. For certain types of $\s$-$\r$ connections, which we call `local', the action of $\P$ preserves the product structure and we obtain (Theorem \ref{thm3}),
\begin{align}
\label{in7}
\lim_{N\rightarrow\infty}\rho_N^{(K)} = \tau^{\otimes K},
\end{align}
where $\tau$ is the explicit single-vertex density matrix obtained by deleting in $\rho$ the coherences associated to the eigenspaces of the interaction operator $A$, which stands for the decoherence, or block diagonalization of $\rho$ in the  eigenbasis of $A$.  For other types of $\s$-$\r$ connections, called `collective', we show that (Theorem \ref{thm:definetti})
\begin{align}
\label{in8}
\lim_{N\rightarrow\infty}  \rho_N^{(K)}= \frac{1}{(2\pi)^{\sss}} \int_{[0,2\pi]^\sss}   \rho(\btheta)^{\otimes K}\, d\btheta,
\end{align}
where the integral is over a torus of dimension $\sss$, the cardinality of the spectrum of the single-vertex operator $A$ governing the interaction between sites, and where $\rho(\btheta)$ is a family of  explicit single-vertex density matrices. Both \eqref{in7} and \eqref{in8} are explicit de Finetti representations. Between the local and the collective types of $\s$-$\r$ interactions above, the interaction can be of `coarse' type leading to an associated interpolating de Finetti representation (Theorem \ref{thm:definetti}). 

In many-body quantum theory the interactions are often scaled with the size of the environment ($N$) to describe different physical regimes. We consider the limit of $\rho_N^{(K)}$ for general such scalings in Section \ref{ch:regimes}, Theorem \ref{thm3.4}. In particular, we find the limit for a scaling factor $1/N$ in front of the $\s$-$\r$ interaction; this makes the $\s$-$\r$ interaction energy of the order one in $N$ (the same as the system energy) and is known as the mean-field scaling. We show that $\rho_N^{(K)}$ converges to a product state of $K$ single-vertex density matrices, as in \eqref{in7} but with a different (still explicit) $\tau$, in which now some coherences are retained. In the so-called fluctuation scaling, the $\s$-$\r$ interaction is scaled by a prefactor $1/\sqrt N$. We prove that, modulo oscillations (denoted $\doteq$), the reduced system state converges to a Gaussian de Finetti superposition,
\begin{align}
\lim_{N\rightarrow\infty} \rho_N^{(K)}   \doteq \frac{1}{\sqrt\pi} \int_{\mathbb R} \sigma(x)^{\otimes K}e^{-x^2}\,dx,
\end{align}
where $\sigma(x)$ is again an explicit single-vertex density matrix. We thus independently recover the fact that Gaussian structures emerge from the fluctuation scaling \cite{FaMe}.

In Section \ref{sec:cano} we present an analysis of several concrete and physically relevant examples. We analyze networks in which every vertex is coupled to every other, as well as networks of spins (qubits), for which we give a detailed, explicit analysis. We then use Erd\H os--R\'enyi random graphs to illustrate how general the graph-theoretic conditions underlying our de Finetti results are. We prove that, for every fixed edge probability $0 < p < 1$,
\[
    \rho_N^{(K)} \longrightarrow \tau^{\otimes K}\qquad \text{almost surely,}
\]
for an explicit, deterministic $\tau$. For $p = 1/2$, every labeled graph on $N$ vertices is equally likely | in this sense our results show that the proportion of all graph geometries for which the de Finetti structure may not emerge, vanishes exponentially quickly as $N\rightarrow\infty$ (see Theorem \ref{thm:erdoes-renyi}).

Section \ref{sec:convinfo} is devoted to the study of the area law for entanglement entropy in networks of quantum systems. We show in Theorem \ref{thm:arealaw} that as $N\rightarrow\infty$, the von Neumann entropy generated by the interaction of $\s$ with $\r$ grows proportionally to the size of the boundary of $\s$, which is the number of vertices in $\s$ which are connected to any vertex in $\r$. When the full network
state is pure, this yields an entanglement entropy area law. This result supports the universality of the area law, which has been verified independently for models in quantum information science, quantum many-body physics and
black hole physics.

In Section \ref{sec:infotheo} we find explicit upper bounds on the convergence speed as $N\rightarrow\infty$, of several information-theoretic quantities associated with $\rho_N^{(K)}$: the conditional von Neumann entropy, mutual information, $\s$-$\r$ entanglement (generalized concurrence), as well as intra-system Wootters concurrence in qubit models. 
\bigskip
\medskip

\section{The model}
\label{sec:model}
\medskip

We consider a network of an increasing number $N$ of components, each a $d$-dimensional quantum system. For fixed $N$ the network is described by the total Hilbert space of pure states, 
\begin{equation}
\label{m1}
\h_N = \bigotimes _{n=1}^N \C^d.
\end{equation} 
The network can be viewed as a {\it graph} with $N$ vertices labeled from $1$ to $N$ in accordance with the factorization \eqref{m1}. Edges in the graph represent pairwise interactions between two vertices (components). We consider {\it simple} graphs, meaning there are no loops and no repeated edges. If two vertices $i$ and $j$ interact, they do so via the product interaction operator
$$
A_i A_j
$$
where $A\in\b(\C^d)$ is a fixed self-adjoint operator (Hermitian matrix) acting on the single-vertex Hilbert space $\C^d$ and  
\begin{align}
\label{ai}
A_i = \one\otimes\cdots\otimes A\otimes\cdots \otimes\one
\end{align}
with $A$ acting on the $i$-th factor of $\h_N$. We denote the graph as $\mathcal{G}_N = (\llbracket 1, N \rrbracket, E(N))$, where we use the notation $\llbracket m,n\rrbracket=\{m,m+1,\ldots,n\}$ for two natural numbers $m\le n$, and where the edge set consists of unordered vertex pairs, represented uniquely by tuples $(i, j)$ with $i < j$:
\begin{equation}
\label{edgeset}
    E(N) \subseteq \left\{(i, j): 1 \le i < j \le N\right\}.
\end{equation}
The sequence of graphs is monotone in the sense that adding a new vertex $N+1$ does not alter any existing edge,
\begin{align}
\label{noch}
    E(N + 1) \cap \big\{(i, j): 1 \le i < j \le N\big\} = E(N).
\end{align}
When the vertex $N + 1$ is added, it may be connected to any subset of the vertices $1, \ldots, N$, while all previous edges remain unchanged. The interaction Hamiltonian is given by
\begin{equation}
\label{eq:hamilt}
H_N = \sum_{(i,j)\in E(N)} A_i A_j.
\end{equation}
Let $\rho$ be a single-vertex density matrix | a positive semi-definite operator on $\C^d$ with trace equal to one | and consider the symmetric product state of the whole network given by the density matrix on $\h_N$, 
$$
\rho^{\otimes N} = \rho\otimes\cdots\otimes\rho.
$$
According to the principles of quantum theory, after an interaction of duration $T > 0$ the network state becomes
\begin{equation}
\label{m2}
\rho_N = U(N) \big( \rho^{\otimes N}\big) U(N)^*,
\end{equation}
where the unitary propagator is given by
\begin{equation}
\label{m7}
U(N)= e^{-iTH_N}.
\end{equation}
The interaction time $T > 0$ is fixed and we do not denote it explicitly in our notation | our main object of interest is not the dynamics {\it per se}, but rather the structure of entanglement created within a subsystem due to the interaction $H_N$ at fixed (generic) $T>0$. 

\begin{example}[$\rho_N$ is generically not exchangeable]
\label{example 2.1}
Let $W_\pi$ be the unitary defined after \eqref{01.3}. From \eqref{m2}, \eqref{m7}, we have 
\begin{align}
W_\pi \rho_N W_\pi^* = e^{-iT \sum_{(i,j)\in E(N)} A_{\pi(i)} A_{\pi(j)}} \big(\rho^{\otimes N}\big) e^{iT \sum_{(i,j)\in E(N)} A_{\pi(i)} A_{\pi(j)}}.
\end{align}
The right-hand side can be interpreted as the state obtained by applying the interactions between vertices of a different graph with permuted edges. Typically, it does not equal $\rho_N$, so this state is not exchangeable.
\end{example}

We consider the first $K$ vertices as a `system of interest' $\s$ while the remaining $N-K$ sites are considered a `reservoir' $\r$; a terminology borrowed from the theory of open quantum systems. The labelling of vertices in our graph is arbitrary, and thus the selection of the first $K$ vertices is done without loss of generality. According to this bipartition we have 
\begin{align}
\label{hs}
\h_N=\h_\s\otimes\h_\r,\qquad \h_\s=\bigotimes_{k=1}^K\C^d,\quad \h_\r=\bigotimes_{n=K+1}^N\C^d.
\end{align}
The edge set \eqref{edgeset} is the disjoint union of edges purely inside $\s$, edges purely inside $\r$ and edges linking $\s$ and $\r$,
\begin{equation}
\label{decompo}
E(N) = E_\s \sqcup E_{\r,N}\sqcup E_{\cc,N},
\end{equation}
where
\begin{equation}
\label{decompo1}
E_\s = \{(i,j)\in E(N) : i,j\in\llbracket 1,K\rrbracket\},\quad  E_{\r,N} = \{(i,j)\in E(N) : i,j\in\llbracket K+1,N\rrbracket\}
\end{equation}
and
$E_{\cc,N}=E(N)\setminus (E_\s\cup E_{\r,N})$ is the set of connecting edges between $\s$ and $\r$. The set $E_\s$ is independent of $N$ because increasing $N$ does not alter the edges within $\s$, \eqref{noch}. The decomposition \eqref{decompo} results in a splitting of the Hamiltonian \eqref{eq:hamilt}, 
\begin{equation}
\label{m10}
H_N = H_\s + H_{\r,N} + H_{\cc,N}
\end{equation}
where
\begin{equation}
\label{eq:ch2-H}
H_\s = \sum_{(i,j)\in E_\s}A_iA_j, \quad H_{\r,N} = \sum_{(i,j)\in E_{\r,N}} A_iA_j, \quad H_{\cc,N} = \sum_{(i,j)\in E_{\cc,N}} A_iA_j.
\end{equation}
Our main object of interest is the reduced state of $\s$, given by the density matrix $\rho^{(K)}_N$ on $\h_\s$ obtained as the `marginal' of the full state $\rho_N$ \eqref{m2},
\begin{equation}
\label{m5}
\rho_N^{(K)} = \Tr_\r(\rho_N) \equiv \Tr_{K+1,\ldots,N}(\rho_N),
\end{equation}
where  $\Tr_\r\equiv \Tr_{K+1,\ldots,N}$ 
denotes the partial trace over the Hilbert space $\h_\r$.

\subsection{Stabilization of the boundary of $\s$ and of its boundary operators.} 
\label{sec:stab}
The {\it boundary} of $\s$ is the collection of the vertices in $\s$ which are connected to any vertex in $\r$. As $\s$ has a fixed number $K$ of vertices there is a smallest natural number $N_*$ such that for all $N\ge N_*$ the boundary is independent of $N$. We denote it as
\begin{equation}
\label{deltaS}
\partial \s = \big\{i\in \llbracket 1,K\rrbracket \, :\, (i,j)\in E(N) \text{\ for some $j\in \llbracket K+1,N\rrbracket$}\big\}.
\end{equation}
We define the {\it boundary operator} $L(j)$ for $j\in \llbracket K+1,N\rrbracket$ by
\begin{equation}
\label{eq:bndryop}
L(j) = \sum_{\substack{1 \le i \le K:\\ (i,j)\in E(N)}} A_i.\qquad 
\end{equation}
This operator is independent of $N$ due to \eqref{noch}. If some reservoir vertex $j\in\llbracket K+1,N\rrbracket$ is not connected to any system vertex $i\in\llbracket 1,K\rrbracket$  then the associated boundary operator $L(j)$ is set equal to zero.
The collection of all boundary operators,
\begin{equation}
\label{LN}
\L_N=\{ L(j) : j\in \llbracket K + 1, N \rrbracket \text{\ with $(i,j)\in E_{\cc,N}$ for some $i \in \llbracket 1, K \rrbracket$} \}
\end{equation}
is a set of commuting operators acting nontrivially on $\h_\s$ only (the first $K$ factors of $\h_N$). The sequence is increasing, $\L_N\subseteq\L_{N+1}$, and every element in $\L_N$ is of the form $\sum_{i=1}^K\varepsilon_i A_i$ for $\varepsilon_i\in\{0,1\}$, so we have $|\L_N|\le 2^K$ for any $N$. Therefore, like the boundary of $\s$, the sequence $\L_N$ stabilizes: there exists a smallest natural number  $N'_*\in\mathbb N$ such that 
\begin{equation}
\label{N*}
\L_N=\L_{N'_*}\quad \text{for all $N\ge N'_*$.}
\end{equation}
Henceforth, we assume that $N\ge \max\{N_*, N_*'\}$. The stabilized boundary set is denoted by $\L$ and its elements are enumerated as 
\begin{equation}
\label{Lstar}
\L = \{ B_\ell : \ell =1,\ldots,|\L|\}. 
\end{equation}
The case $\L=\varnothing$ corresponds to the trivial situation where $\s$ and $\r$ do not interact (no edges linking $\s$ and $\r$). Each element of $\L$ has the form
\begin{equation}
\label{m18}
B_\ell = \sum_{i\in J_\ell} A_i
\end{equation}
for some subset $J_\ell\subset\llbracket  1, K\rrbracket$. As all elements of $\L$ commute, the can be simultaneously diagonalized. There is a unique family of minimal orthogonal projections $P_\lambda$ on $\h_\s$, labelled by $\lambda$ belonging to an index set $\Lambda$, such that 
\begin{equation}
\label{Beval}
B_\ell P_\lambda = b_\ell(\lambda)P_\lambda,
\end{equation}
where the $b_\ell(\lambda)$ are the eigenvalues of $B_\ell$ and the projections satisfy\footnote{The $P_\lambda$ can be constructed as follows. For each $\ell$ take the spectral decomposition $B_\ell=\sum_{\beta\in{\rm spec}(B_\ell)} \beta Q_\ell(\beta)$. For every $\beta_1\in{\rm spec}(B_1),\ldots,\beta_{|\L|}\in{\rm spec}(B_{|\L|})$ set $P(\beta_1,\ldots,\beta_{|\L|}) = Q_1(\beta_1)\cdots Q_{|\L|}(\beta_{|\L|})$ (commuting factors). $P(\beta_1,\ldots,\beta_{|\L|})$ is the orthogonal projection onto the intersection of the corresponding eigenspaces.  Among the list of orthogonal projections $P(\beta_1,\ldots,\beta_{|\L|})$ select all nonzero ones, corresponding to projections onto eigenspaces common to all $B_\ell$, and denote them by $P_\lambda$, $\lambda\in\Lambda$. }
\begin{equation}
\label{resid}
P_\lambda P_\mu=\delta_{\lambda,\mu}P_\lambda,\qquad \sum_{\lambda\in\Lambda}P_\lambda =\one.
\end{equation}
Moreover, the minimal projections separate the operators in the sense that  
\begin{equation}
\label{sep}
\big\{\lambda,\mu\in\Lambda \text{\ and $\lambda\neq\mu$}\big\} \ \Longrightarrow \ \big\{\exists \ell=1,\ldots,|\L|\ \text{\  s.t.\ \,} b_\ell(\lambda)\neq b_\ell(\mu) \big\}.
\end{equation}
We define the set of {\it Bohr energies} of $\L$ as
\begin{equation}
\label{Bohr}
\mathfrak B = \big\{b_\ell(\lambda)-b_\ell(\mu) : \lambda,\mu\in\Lambda\,, \ell=1,\ldots,|\L|\big\}.
\end{equation}
If $|\Lambda|=1$ then we have $P_1= \mathbf 1$, so by \eqref{Beval} $B_1\propto \mathbf 1$ and by \eqref{m18} $A\propto \mathbf 1$, which is a trivial case (the subsystems are not coupled). We therefore consider only  $|\Lambda|\ge 2$, which implies that  $\mathfrak B\neq \{0\}$.
\bigskip
\medskip

\section{ Convergence of the system density matrix}
\label{ch:conv}
\medskip

We analyze the density matrix $\rho_N^{(K)}$, \eqref{m5} for large values of $N$. 
We will show that it has a limit provided $\s$ and $\r$ are ``well coupled''. In what follows, the convergence $\lim_{N\to\infty}\rho_N^{(K)}=\rho$ is understood in the norm of operators, with the choice of norm being immaterial since they are all equivalent in finite dimensions. However, estimates of differences $\rho_N^{(K)}-\rho$ depend on the norm utilized, and it is natural to use the trace norm, 
\begin{align}
\label{trn}
\|O\|_1=\Tr \sqrt{O^* O}.
\end{align}

\begin{definition} \ \\[-3ex]
\label{def:connection number}
\begin{itemize}
\item[$\circ$] The {\it incidence number} of $B_\ell\in\mathcal L$ is defined as
\begin{equation}
\label{incidence number}
\nu(\ell,N) = \big|\{ j=K+1,\ldots,N : L(j) = B_\ell\}\big|.
\end{equation}
It counts the number of vertices in the reservoir $\r$ with boundary operator $B_\ell$. Trivially, $0\le \nu(\ell,N)\le N-K$ for all $\ell$. \\[-.4ex]

\item[$\circ$] The {\it effective  connection number} of the network is defined as
\begin{equation}
\label{connection number}
\ve_K(N)  = \min_{\lambda\neq \mu} \sum_{\ell: \, b_\ell(\lambda)\neq b_\ell(\mu)}\nu(\ell,N).
\end{equation}

\item[$\circ$] The system $\s$ (the sites or vertices  $k=1,\ldots K$) is called {\it well-connected} if 
\begin{equation}
\label{wellconn}
\lim_{N\rightarrow\infty} \ve_K(N) = \infty.
\end{equation}
\end{itemize}
\end{definition}

\medskip

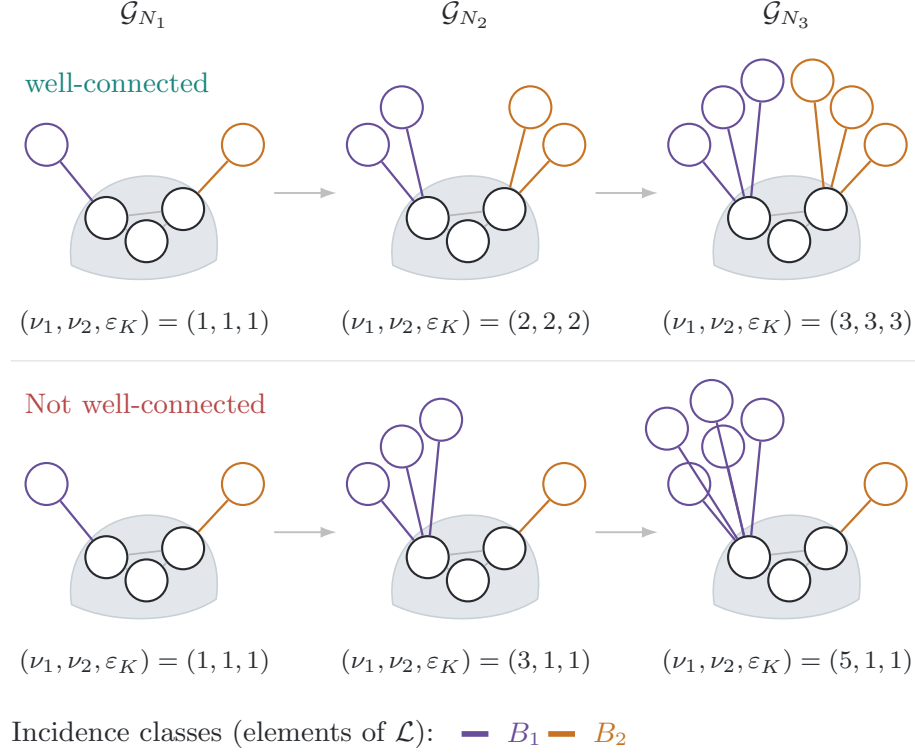
\begin{figure}[t]
    \centering

    \resizebox{0.75\textwidth}{!}{%
        \begin{tikzpicture}[x=1cm,y=1cm]
           % \node[paneltitle] at (0,4.05) {Effective connectivity along a graph sequence};
          %  \node[sub,text=Ink!60] at (0,3.68) {the system is fixed while reservoir incidence classes accumulate};

            \node[sub,anchor=east] at (-.22,-5.3) {Incidence classes (elements of $\L$):};
            \draw[Collective,line width=1.8pt] (-.05,-5.3)--(.25,-5.3);\node[sub,anchor=west,text=Collective] at (.32,-5.3) {$B_1$};
            \draw[Coarse,line width=1.8pt] (.92,-5.3)--(1.22,-5.3);\node[sub,anchor=west,text=Coarse] at (1.29,-5.3) {$B_2$};
        
            \foreach \x/\lab in {-3.6/{\mathcal G_{N_1}},0/{\mathcal G_{N_2}},3.6/{\mathcal G_{N_3}}}{\node[sub,font=\bfseries\footnotesize] at (\x,2.76) {$\lab$};}

            % ----- Well-connected sequence -----
            \node[sub,font=\footnotesize,text=Local,anchor=west] at (-5.08,2.0) {well-connected};
            % \draw[Local!28,line width=.75pt] (-5.08,1.61)--(-4.00,1.61);
            \snap{-3.6}{.40}{1}{1}{1}{g1}
            \snap{0}{.40}{2}{2}{2}{g2}
            \snap{3.6}{.40}{3}{3}{3}{g3}
            \draw[-{Latex[length=1.8mm,width=1.25mm]},draw=Ink!30,line width=.60pt] (-2.15,.76)--(-1.45,.76);
            \draw[-{Latex[length=1.8mm,width=1.25mm]},draw=Ink!30,line width=.60pt] (1.45,.76)--(2.15,.76);
        
            \draw[Ink!11,line width=.55pt] (-5.1,-1.12)--(5.1,-1.12);

            % ----- Not well-connected sequence -----
            \begin{scope}[yshift=-3mm]
                \node[sub,font=\footnotesize,text=Bad,anchor=west] at (-5.08,-1.3) {Not well-connected};
                % \draw[Bad!28,line width=.75pt] (-5.08,-1.80)--(-3.72,-1.80);
                \snap{-3.6}{-3.10}{1}{1}{1}{b1}
                \snap{0}{-3.10}{3}{1}{1}{b2}
                \snap{3.6}{-3.10}{5}{1}{1}{b3}
                \draw[-{Latex[length=1.8mm,width=1.25mm]},draw=Ink!30,line width=.60pt] (-2.15,-2.74)--(-1.45,-2.74);
                \draw[-{Latex[length=1.8mm,width=1.25mm]},draw=Ink!30,line width=.60pt] (1.45,-2.74)--(2.15,-2.74);
            \end{scope}            
        \end{tikzpicture}
    }

    \caption{{\it Effective connectivity along a graph sequence.} $\s$ consisting of $K=3$ vertices connected to an $\r$ of a growing number $N-K=2,4,6$ of vertices. $\L$, \eqref{Lstar} consists of two elements $B_1, B_2$ (each one of them is one of the $A_i$). The system is fixed while reservoir incidence classes accumulate. As the graph grows, the incidence numbers $\nu_1,\nu_2$ associated to $B_1,B_2$, \eqref{incidence number} are growing in the first row, while one of them remains bounded (equal to one) in the bottom row. Accordingly, the effective connection number $\ve_K(N)$,  \eqref{connection number} grows in $N$ or stays bounded. Growing $\ve_K(N)\rightarrow\infty$ gives a well-connected system $\s$.}
    \label{fig:incidence}
\end{figure}

If $\s$ is well-connected, then by \eqref{connection number} the sum grows indefinitely, so $\forall\lambda\neq\mu$, $\exists\ell$ such that $b_\ell(\lambda)\neq b_\ell(\mu)$ and $\nu(\ell,N)\rightarrow\infty$. Conversely, suppose that $\forall\lambda\neq\mu$, $\exists\ell$ such that $b_\ell(\lambda)\neq b_\ell(\mu)$ and $\nu(\ell,N)\rightarrow\infty$. Then $\forall\lambda\neq\mu$, the sum in \eqref{connection number} grows to $\infty$ and so $\s$ is well-connected. This shows the equivalence 
\begin{equation}
\label{equivalence}
\text{$\s$ is well-connected} \quad \Longleftrightarrow \quad \text{$\forall\lambda,\mu\in \Lambda,\ \lambda\neq\mu$, $\exists\ell$ {s.t.}~$b_\ell(\lambda)\neq b_\ell(\mu)$ and $\lim_{N\rightarrow\infty}\nu(\ell,N)=\infty$.}
\end{equation}

\begin{example}  \ \\[-3ex]
\begin{itemize}
\item[$\circ$]  If none of the $\nu(\ell,N)$ grow to $\infty$ with $N\rightarrow\infty$ then $\s$ is not well-connected (see also Figure \ref{fig:incidence}). A necessary condition for well-connectedness is that the number of connections (edges) between the system $\s$ and the reservoir $\r$ satisfies $|E_{\cc,N}|\rightarrow\infty$ as $N\rightarrow\infty$. \\[-.6ex]

\item[$\circ$] If each $B_\ell\in\L$ has an incidence number growing to infinity with $N\rightarrow\infty$, then the system is well-connected (see also Figure \ref{fig:incidence}).
\end{itemize}
\end{example}

Denote the spectral decomposition of the single vertex interaction operator $A$ (see \eqref{ai}) as 
\begin{equation}
\label{m4}
A=\sum_{s=1}^\sss a_s\Pi_s,
\end{equation}
where $\sss$ is the number of distinct eigenvalues $a_s\in\R$  and $\Pi_s$ are the corresponding orthogonal eigenprojections with $\rank\, \Pi_s\ge 1$. The projections are a complete, disjoint family,
\begin{equation}
\label{eq:specres}
\Pi_s\Pi_{s'}=\delta_{s,s'}\Pi_s,\qquad \sum_{s=1}^\sss \Pi_s=\one.
\end{equation}

\begin{example} 
If $\L=\{A_1,\ldots,A_k\}$ for some $1\le k\le K$ (see \eqref{ai}) then $\Lambda=\{1,\ldots,\sss\}^k$ and for $\lambda=(s_1,\ldots,s_k)\in\Lambda$ we have $P_\lambda = \Pi_{s_1}\otimes \cdots \otimes \Pi_{s_k}\otimes \one^{\otimes (K - k)}$. 
\end{example}

The probability that a quantum (von Neumann projective) measurement of the observable $A$ in the single-vertex state $\rho$ (see \eqref{m2}) yields the output $a_s$ is given by 
\begin{equation}
\label{eq:ps}
p_s=\Tr(\Pi_s\rho).
\end{equation}
We have $0\le p_s\le 1$ and $p_1+\cdots+p_\sss=1$. Let $\xi$ be a random variable taking the values $a_1,\ldots,a_\sss$ with probabilities $p_1,\ldots,p_\sss$: $\Pr(\xi=a_s)=p_s$ and denote its characteristic function by 
\begin{equation}
\label{charfunct}
\psi(\Delta) = \E\big[e^{-iT\xi\Delta}\big] = \sum_{s=1}^\sss p_s e^{-iT a_s\Delta},\qquad \Delta\in\R.
\end{equation}
We set ({\it cf.}~\eqref{Bohr})
\begin{equation}
\label{def q}
q= \max_{\Delta\in\mathfrak B\backslash\{0\}} |\psi(\Delta)|.
\end{equation}
A map $\mathcal P$ acting on operators $\mathcal B(\h_\s)$ is called CPTP (completely positive, trace preserving) if it preserves the trace and if for all $n\in\mathbb N$ the map $\mathcal P\otimes{\rm Id}_n$ is positivity preserving  on $\mathcal B(\h_\s)\otimes \mathcal B(\mathbb C^n)$ \cite{Stein,Kraus,MerkliNeira2026}. Recall the definition \eqref{trn} of the trace norm.

\begin{theorem}
\label{thm2}
We have 
\begin{equation}
\label{m23}
\big\| \rho_N^{(K)} - e^{-iTH_\s} \mathcal P \big[\rho^{\otimes K}\big] e^{iTH_\s}\big\|_1 \le |\Lambda| q^{\ve_K(N)},
\end{equation}
where the CPTP map $\mathcal P$ is defined by
\begin{equation}
\label{defP}
\mathcal P[X] =  \sum_{\lambda\in\Lambda} P_\lambda XP_\lambda,\qquad X\in\mathcal B(\h_\s).
\end{equation}
\end{theorem}
The map $\mathcal P$ implements a partial diagonalization with respect to the direct sum (block) decomposition $\h_\s=\bigoplus_{\lambda\in\Lambda} P_\lambda \h_\s$. Its action on a density matrix $X$ represents physically the {\it decoherence}  or equivalently, the output of a {\it non-selective quantum measurement} relative to the family $\{P_\lambda\}$. $\mathcal P$ is also called a {\it twirl} in representation theory and in quantum information theory.

We have $|\Lambda|\le \sss^K$ 
where we recall that $\sss$ is the number of distinct eigenvalues of $A$. The upper bound of \eqref{m23} converges to zero as $N\rightarrow\infty$ if the graph is well-connected, \eqref{wellconn} and if $q<1$. For this inequality we note that the characteristic function \eqref{charfunct} satisfies $|\psi(\Delta)|\le 1$ and
$|\psi(\Delta)| = 1 \Longleftrightarrow |\sum_{s=1}^\sss p_s e^{-iTa_s\Delta}|=1$. The sum is a convex combination of points on the complex unit circle, so its value lies on the unit circle if and only if all the points with nonzero weight are the same,
\begin{align}
|\psi(\Delta)| = 1 
& \Longleftrightarrow (a_s-a_{s'}) \Delta \in \frac{2\pi}{T}\Z, \ \forall s,s'\text{\ s.t. $p_s,p_{s'}\neq 0$}.
\label{psi1}
\end{align}
The case $|\psi(\Delta)|=1$ for $\Delta\neq 0$ thus imposes a delicate lattice on the eigenvalues $a_s$ of $A$, which is generically not satisfied | and if it happens to be satisfied for one $T$ it won't be satisfied for nearby values of $T$. This motivates the following. 

\begin{definition}
The {\it generic regime} is characterized by
\begin{equation}
\label{gencond}
q= \max_{\Delta\in\mathfrak B\backslash\{0\}} |\psi(\Delta)| <1. 
\end{equation}
\end{definition}
The next result follows directly from the above discussion and Theorem \ref{thm2}.

\begin{corollary}[Decoherence]
\label{cor:deco}
Suppose the genericity condition \eqref{gencond} holds and that the system is well-connected, \eqref{wellconn}. Then 
\begin{equation}
\label{m25}
\lim_{N\rightarrow\infty}  \rho_N^{(K)} =  e^{-iTH_\s} \mathcal P\big[\rho^{\otimes K}\big]  e^{iTH_\s}.
\end{equation}
\end{corollary}

\begin{example} In a network where every site is connected to every other site we have $E_{\cc,N} = \{(i,j): 1 \le i \le K < j \le N\}$ and $\L=\{ B=\sum_{k=1}^KA_k\}$ is a singleton. The incidence number of $B$ is $\nu(1,N)=N-K$ which makes the system  well-connected.  The Bohr energies are the eigenvalue differences of the operator $B$. As an example, take each vertex to host a qubit ($d=2$ in \eqref{m1}) and let $A=\sigma_x$ be the Pauli $x$-matrix with eigenvalues $\pm 1$. The spectrum of $B$ is ${\rm spec}(B) = \{k_+-k_-\ :\ k_\pm \in \Z_{\ge 0},\, k_++k_-=K\}$, so that $\mathfrak B\subset\mathbb Z$. The genericity assumption \eqref{gencond} holds if $p_+p_- > 0$ and, for all $s,s'$ with $s\neq s'$, we have  $a_s-a_{s'}\not\in \frac{2\pi}{T}\mathbb Q$.
\end{example}

Theorem \ref{thm2} and Corollary \ref{cor:deco} do not require any symmetry properties on the distribution of the edges in the graph. As a result, the state $\rho_N^{(K)}$, and its limit, need not be symmetric under permutations of subsystems (vertices) within $\s$. In order to obtain a representation of these states as superpositions of symmetric states, we should impose some uniformity in how each vertex of $\s$ is coupled to $\r$. For this we introduce the following notions.

\begin{definition}[Local, collective, coarse resolutions of $\L$]
\label{def:loccolres}
Recall the definition of the boundary operators $\L$, \eqref{LN}.
\begin{itemize}
\item[{\bf (i)}] We say $\L$ is {\it locally resolved} if the minimal projections $P_\lambda$ are given by
\begin{equation}
\label{locres}
P_\lambda=\Pi_{s_1}\otimes\cdots \otimes\Pi_{s_K},\quad \text{with \quad $\lambda=(s_1,\ldots,s_K)\in\Lambda= \{1,\ldots, \sss\}^K$}.
\end{equation} 

\item[{\bf (ii)}] We say $\L$ is {\it collectively resolved} if $\L$ is the singleton $\L=\{\sum_{k=1}^KA_k\}$. \medskip

\item[{\bf (iii)}] We say $\L$ is {\it coarsely resolved} if $\llbracket 1,K\rrbracket$ is partitioned into $r+1 \ge 2$ nonempty consecutive blocks, according to 
\begin{align}
\label{2.16}
0 = K_0 < K_1<\cdots<K_r < K_{r + 1} = K
\end{align}
and if \[
    \L=\{ B_1,\ldots,B_{r+1}\}, \qquad B_\ell = \sum_{K_{\ell-1}+1\le i\le K_\ell} A_i.
\] 
\end{itemize}
\end{definition}

\medskip

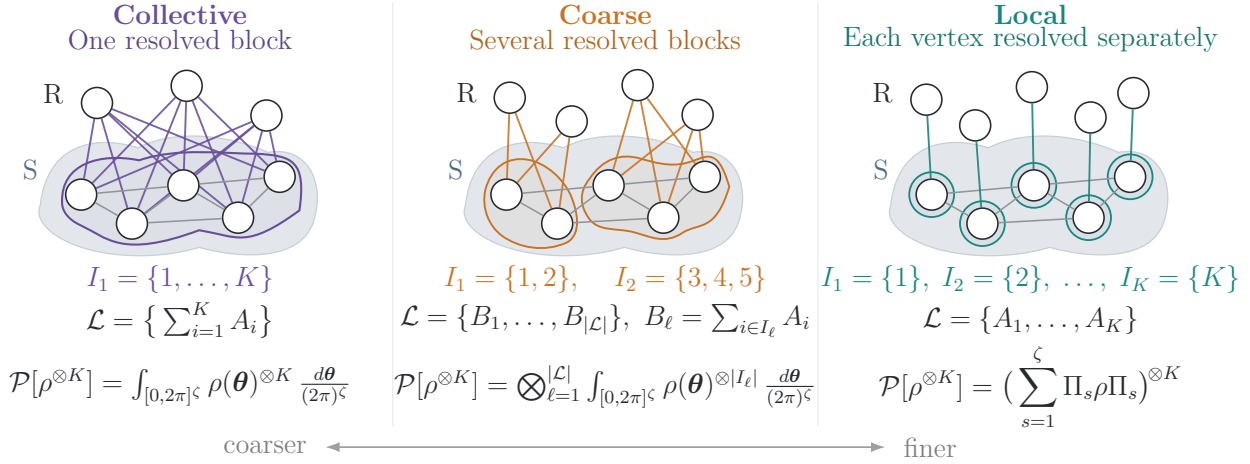
\begin{figure}[t]
    \centering

    \resizebox{\textwidth}{!}{%
        \begin{tikzpicture}[x=1cm,y=1cm]
            \node[font=\bfseries\large\rm] at (6.45,8.47) {\large };
            \node[panelsub,text=Ink!60,anchor=east] at (2.06,.8) {\large coarser};
            \draw[refarrow] (2.20,.8)--(10.70,.8);
            \node[panelsub,text=Ink!60,anchor=west] at (10.84,.8) {\large finer};
        
            \def\XC{0.00}
            \def\XM{6.45}
            \def\XL{12.90}
        
            % ----- Collective -----
            \begin{scope}[shift={(\XC,0)}]
                \node[paneltitle,text=Collective] at (0,7.36) {Collective};
                \node[panelsub,text=Collective!88!black] at (0,6.99) {\large One resolved block};
        
                \node[vertex] (Cr1) at (-1.26,6.02) {};
                \node[vertex] (Cr2) at (.10,6.28) {};
                \node[vertex] (Cr3) at (1.31,5.83) {};
                \node[rlabel,anchor=east] at (-1.64,6.17) {\large $\r$};
        
                \begin{scope}[shift={(0,4.55)}]
                    \begin{scope}[on background layer]
                        \systemblob
                        
                        \path[draw=Collective,line width=.90pt,fill=Collective,fill opacity=.045]
                        (-1.79,-.02)
                        .. controls (-1.8,.55) and (-1.21,.72) .. (-.62,.60)
                        .. controls (-.05,.80) and (.68,.70) .. (1.16,.70)
                        .. controls (1.7,.67) and (1.92,.60) .. (1.81,-.16)
                        .. controls (1.0,-.75) and (.88,-.56) .. (.29,-.46)
                        .. controls (-0.9,-0.9) and (-1.11,-.57) .. (-1.52,-.38)
                        .. controls (-1.69,-.30) and (-1.80,-.17) .. cycle;
                    \end{scope}
                    
                    \systemgraph{C}
                    
                    \node[regionlabel,text=BlueGrey,anchor=east] at (-2.02,.46) {\large $\s$};
                    \node[regionlabel,text=Collective,anchor=north] at (0,-.86) {\large $I_1=\{1,\ldots,K\}$};
                \end{scope}
            
                \begin{scope}[on background layer]
                    \foreach \r in {1,2,3}{\foreach \s in {1,2,3,4,5}{\draw[cedge] (Cr\r)--(Cs\s);}}
                \end{scope}
            
                \Lrow{\large $\mathcal L=\big\{\sum_{i=1}^{K}A_i\big\}$}
                \StateMain{\large $\mathcal P[\rho^{\otimes K}] =\int_{[0,2\pi]^\sss} \rho(\boldsymbol\theta)^{\otimes K}\,\frac{d\boldsymbol\theta}{(2\pi)^\sss}$}
            \end{scope}
        
            % ----- Coarse -----
            \begin{scope}[shift={(\XM,0)}]
                \node[paneltitle,text=Coarse] at (0,7.36) {Coarse};
                \node[panelsub,text=Coarse!88!black] at (0,6.99) {\large Several resolved blocks};
            
                \node[vertex] (Mr1) at (-1.45,6.10) {};
                \node[vertex] (Mr2) at (-.51,5.74) {};
                \node[vertex] (Mr3) at (.49,6.29) {};
                \node[vertex] (Mr4) at (1.40,5.84) {};
                \node[rlabel,anchor=east] at (-1.82,6.18) {\large $\r$};
            
                \begin{scope}[shift={(0,4.55)}]
                    \begin{scope}[on background layer]
                        \systemblob
        
                        \path[fill=Coarse,fill opacity=.055,draw=Coarse,line width=.82pt]
                        (-1.82,.15)
                        .. controls (-1.73,.50) and (-1.48,.66) .. (-1.18,.61)
                        .. controls (-.89,.56) and (-.63,.31) .. (-.52,.06)
                        .. controls (-.42,-.19) and (-.36,-.50) .. (-0.57,-.64)
                        .. controls (-.93,-.81) and (-1.32,-.72) .. (-1.56,-.55)
                        .. controls (-1.77,-.40) and (-1.88,-.12) .. cycle;
                  
                        \path[fill=Coarse,fill opacity=.055,draw=Coarse,line width=.82pt]
                        (-.24,.32)
                        .. controls (-.04,.62) and (.23,.71) .. (.52,.63)
                        .. controls (.76,.56) and (.98,.50) .. (1.23,.56)
                        .. controls (1.51,.63) and (1.69,.86) .. (1.88,.20)
                        .. controls (1.67,-.50) and (1.42,-.34) .. (1.17,-.5)
                        .. controls (.90,-.68) and (.69,-.56) .. (.43,-.49)
                        .. controls (.16,-.42) and (-.10,-.49) .. (-.25,-.28)
                        .. controls (-.39,-.08) and (-.40,.12) .. cycle;
                    \end{scope}
                
                    \systemgraph{M}
                    \node[regionlabel,text=BlueGrey,anchor=east] at (-2.02,.46) {\large $\s$};
                    \node[regionlabel,text=Coarse,anchor=north] at (0,-.86) {\large $I_1=\{1,2\},\quad I_2=\{3,4,5\}$};
                \end{scope}
            
                \begin{scope}[on background layer]
                    \foreach \r in {1,2}{\foreach \s in {1,2}{\draw[medge] (Mr\r)--(Ms\s);}}
                    \foreach \r in {3,4}{\foreach \s in {3,4,5}{\draw[medge] (Mr\r)--(Ms\s);}}
                \end{scope}
            
                \Lrow{\large $\mathcal L=\{B_1,\ldots,B_{|\L|}\},\  B_\ell=\sum_{i\in I_\ell}A_i$}
                \StateMain{\large $\mathcal P[\rho^{\otimes K}] =\bigotimes_{\ell=1}^{|\L|}\int_{[0,2\pi]^\sss} \rho(\boldsymbol\theta)^{\otimes |I_\ell|}\,\frac{d\boldsymbol\theta}{(2\pi)^\sss}$}
            \end{scope}
        
            % ----- Local -----
            \begin{scope}[shift={(\XL,0)}]
                \node[paneltitle,text=Local] at (0,7.36) {Local};
                \node[panelsub,text=Local!88!black] at (0,6.99) {\large Each vertex resolved separately};
            
                \node[vertex] (Lr1) at (-1.58,6.07) {};
                \node[vertex] (Lr2) at (-.84,5.68) {};
                \node[vertex] (Lr3) at (.02,6.26) {};
                \node[vertex] (Lr4) at (.91,5.80) {};
                \node[vertex] (Lr5) at (1.56,6.17) {};
                \node[rlabel,anchor=east] at (-1.96,6.18) {\large $\r$};
            
                \begin{scope}[shift={(0,4.55)}]
                    \begin{scope}[on background layer]
                        \systemblob
                        
                        \foreach \x/\y in {-1.50/.08,-.73/-.36,.05/.22,.88/-.27,1.51/.35}{
                            \filldraw[fill=Local,fill opacity=.055,draw=Local,line width=.78pt] (\x,\y) circle[radius=.34cm];
                        }
                    \end{scope}
                
                    \systemgraph{L}
                    
                    \node[regionlabel,text=BlueGrey,anchor=east] at (-2.02,.46) {\large $\s$};
                    \node[regionlabel,text=Local,anchor=north] at (0,-.86) {\large $I_1=\{1\},\;I_2=\{2\},\;\ldots,\;I_K=\{K\}$};
                \end{scope}
            
                \foreach \i in {1,2,3,4,5}{\draw[ledge] (Lr\i)--(Ls\i);}
            
                \Lrow{\large $\displaystyle \mathcal L=\{A_1,\ldots,A_K\}$}
                \StateMain{\large $\displaystyle \mathcal P[\rho^{\otimes K}] =\big(\sum_{s=1}^{\sss}\Pi_s\rho\Pi_s\big)^{\!\otimes K}$}
            \end{scope}
        
            \draw[Ink!10,line width=.50pt] (3.225,1.12)--(3.225,7.55);
            \draw[Ink!10,line width=.50pt] (9.675,1.12)--(9.675,7.55);
        \end{tikzpicture}
    }
\caption{{\it Different resolutions of the system by the reservoir.} Collective resolution: each vertex in $\r$ connected to $\s$ is linked to every vertex in $\s$; the set of boundary operators $\L$, \eqref{Lstar} is a singleton; the CPTP map $\P$, \eqref{defP} is an integral over a $\sss$-dimensional torus with factorized, symmetric integrand. Coarse resolution: $\s$ is subdivided into several subsystems (two in the example depicted); each vertex in $\r$ which connects to $\s$ is linked to all the vertices of exactly one subsystem; each subsystem contributes to $\L$ with one operator; $\P$ is the product of actions of the collective type on each subsystem. Local resolution: each vertex in $\r$ which is connected to $\s$ connects to exactly one vertex in $\s$, and all vertices in $\s$ are connected to $\r$. $\L$ consists of all individual $A_i$ and the action of $\P$ gives a symmetric product state in which each factor is decohered. In all cases: Vertices in $\r$ may be connected among each other and there may be vertices in $\r$ not linked to $\s$ (not shown in picture).}
    \label{fig:resolution}
\end{figure}

\begin{example} 
$\L$ is locally resolved  if $A_k\in\L$ for all $k=1,\ldots,K$, which in turn is the case if each vertex in the system $\s$ is linked to at least one vertex in the reservoir $\r$, while no vertex in the reservoir is linked to more than one vertex in the system. $k$. If on the other hand there exists a vertex in the reservoir which is connected to all vertices in the system, and all the vertices in the reservoir which are connected to $\s$ are connected in this way, then $\L$ is collectively resolved.  If a system $\s$ is subdivided into $r+1$ disjoint subsystems as per \eqref{2.16} and each subsystem is collectively resolved in its own right, then $\s$ is coarsely resolved.
\end{example}

For a coarsely resolved $\L$  the index set $\Lambda$ has a product structure $\Lambda = \Lambda_1\times \cdots \times \Lambda_{r+1}$ and the projections $P_\lambda$ are given by 
\begin{align} 
\label{coares}
P_\lambda= \bigotimes_{\ell=1}^{r+1} P^{[\ell]}_{\lambda_\ell},\qquad \lambda=(\lambda_1,\ldots,\lambda_{r+1})\in\Lambda_1\times \cdots \times \Lambda_{r+1},
\end{align}
where $P^{[\ell]}_j$ acts nontrivially only on the tensor factors of the subsystem with indices in $\llbracket K_{\ell-1}+1,K_\ell\rrbracket$.

\begin{example}
Suppose $r=1$, then $\s$ has two subsystems with vertices in $\llbracket 1,K_1\rrbracket$ and $\llbracket K_1+1,K\rrbracket$. The $P^{[1]}_{\lambda_1}$ are the eigenprojections of $B_1=A_1+\cdots+A_{K_1}$, $\lambda_1\in \Lambda_1$, acting nontrivially on the first $K_1$ tensor factors of $\h_\s$ and similarly for $P^{[2]}_{\lambda_2}$. 
\end{example}

\begin{lemma} \label{lem:P} Recall the definition \eqref{defP} of the CPTP map $\P$ acting on $\b(\h_\s)$.
\begin{itemize}
\item[{\rm \bf (1)}] If $\L$ is locally resolved then $\forall X_1,\ldots,X_K\in\b(\C^d)$,
\begin{equation}
\label{locresreas}
e^{-iTH_\s}\mathcal P\big[ X_1\otimes\cdots \otimes X_K\big]  e^{iTH_\s} =\Big(\sum_{s=1}^\sss \Pi_s X_1\Pi_s\Big)\otimes\cdots\otimes \Big(\sum_{s=1}^\sss \Pi_s X_K\Pi_s\Big).
\end{equation}

\item[{\rm \bf (2)}] If $\L$ is collectively resolved then $\forall X\in\b(\h_\s)$,
\begin{align}
\label{2.24}
 \P[X]  =\lim_{R\rightarrow\infty} \frac 1R \int_0^R e^{-it\sum_{k=1}^KA_k} X e^{it\sum_{k=1}^KA_k}dt.
\end{align}

\item[{\rm \bf (3)}] Let $\L$ be coarsely resolved according to the partition \eqref{2.16} and define the map $\P_R^{(\ell)}$ on $\b(\h_\s)$ by
\begin{align}
\P^{(\ell)}_R[X] = \frac 1R\int_0^R e^{-itB_\ell} X e^{itB_\ell} dt,\qquad R>0,\ \ell=1,\ldots,r+1.
\end{align}
Then we have for all $X\in\b(\h_\s)$, 
\begin{align}
\label{2.24.1}
 \P[X]  = \lim_{R\rightarrow\infty} \P_R^{(1)}\circ\cdots \circ \P_R^{(r+1)}[X].
\end{align}
\end{itemize}
\end{lemma}
The case \eqref{2.24} is the special case \(r=0\) of \eqref{2.24.1}, corresponding to a single collectively resolved subsystem. At the opposite extreme, when \(r=K-1\) and each subsystem consists of a single vertex, \eqref{2.24.1} reduces to the locally resolved case \eqref{locresreas}. Thus the coarsely resolved setting may be viewed as an interpolation between collective and local resolution, obtained by successively refining the partition of the system into independently resolved subsystems.
\medskip

{\bf Proof of Lemma \ref{lem:P}.} For locally resolved $\L$ the $P_\lambda$ are eigenprojections of $H_\s$, \eqref{eq:ch2-H}. Then  \eqref{locresreas} follows immediately. For a collectively resolved $\L$ the $P_\lambda$, $\lambda\in\Lambda$, are the eigenprojections of $B=\sum_{k=1}^KA_k$ with distinct eigenvalues $b(\lambda)$. By \eqref{eq:specres} we have $\forall X\in\b(\h_\s)$ and $R>0$,
\begin{align}
\label{-m1}
\frac 1R \int_0^R e^{-it\sum_{k=1}^KA_k} X e^{it\sum_{k=1}^KA_k}dt &= \sum_{\lambda\in\Lambda} P_\lambda XP_\lambda + \sum_{\substack{\lambda, \mu \in \Lambda\\\lambda \ne \mu}} P_\lambda XP_\mu \ \frac1R\int_0^R e^{-it(b(\lambda)-b(\mu))}dt.
\end{align}
The first summand on the right-hand side is just $\P[X]$. For $\lambda\neq \mu$ we have $b(\lambda)\neq b(\mu)$, so that $\lim_{R\rightarrow\infty}\frac1R\int_0^R e^{-it(b(\lambda)-b(\mu))}dt=0$. Thus we obtain \eqref{2.24}. 
For a coarsely resolved $\L$ we write $e^{-it\sum_{k=1}^K A_k} = e^{-it (B_1+\cdots+B_{r+1})} = \otimes_{\ell=1}^{r+1} e^{-it B_\ell}$ (see (iii) in Definition \ref{def:loccolres}). Then for $X\in\b(\h_\s)$,
\begin{align} 
e^{-i\sum_{\ell=1}^{r+1} t_\ell B_\ell} &Xe^{i\sum_{\ell=1}^{r+1}t_\ell B_\ell} = \big(\bigotimes_{\ell=1}^{r+1} e^{-it_\ell B_\ell}\big)
X\big(\bigotimes_{\ell=1}^{r+1} e^{it_\ell B_{\ell}}\big)\nonumber\\
&= \sum_{\substack{\lambda=(\lambda_1,\ldots,\lambda_{r+1})\in \Lambda \\ \mu=(\mu_1,\ldots,\mu_{r+1})\in \Lambda}} \big( \bigotimes_{\ell=1}^{r+1} P^{[\ell]}_{\lambda_\ell} e^{-it_\ell B_\ell}\big)  X\big( \bigotimes_{\ell=1}^{r+1} P^{[\ell]}_{\mu_{\ell}} e^{it_\ell B_{\ell}}\big)\nonumber\\
&= \sum_{\substack{\lambda=(\lambda_1,\ldots,\lambda_{r+1})\in \Lambda \\ \mu=(\mu_1,\ldots,\mu_{r+1})\in \Lambda}}  e^{-i \sum_{\ell=1}^{r+1}t_\ell (b_\ell(\lambda_\ell)-b_\ell(\mu_\ell))}
 \big( \bigotimes_{\ell=1}^{r+1} P^{[\ell]}_{\lambda_\ell} \big)  X\big( \bigotimes_{\ell=1}^{r+1}P^{[\ell]}_{\mu_{\ell}}\big)
 \label{2.22}
\end{align}
where we inserted $\mathbf 1 = \sum_{\lambda\in \Lambda}P_\lambda$ and took into account \eqref{coares}, and where $b_\ell(j)$ is the eigenvalue of $B_\ell$ associated to the eigenprojection $P^{[\ell]}_j$. For $\lambda\neq\mu$ there is an $\ell$ such that $\lambda_\ell\neq \mu_\ell$ and so
\begin{align}
\lim_{R\rightarrow\infty} \frac{1}{R^{r+1}} \int_0^R\cdots\int_0^R &e^{-i \sum_{\ell=1}^{r+1}t_\ell (b_\ell(\lambda_\ell)-b_\ell(\mu_\ell))} dt_1\cdots dt_{r+1}\nonumber\\
&=\lim_{R\rightarrow\infty} \prod_{\ell=1}^{r+1} \frac 1R \int_0^Re^{-i t (b_\ell(\lambda_\ell)-b_\ell(\mu_\ell))} dt=0.
\end{align}
Therefore only the terms $\lambda=\mu$ in \eqref{2.22} survive the integration and limit,
\begin{align*} 
\lim_{R\rightarrow\infty} \frac{1}{R^{r+1}} &\int_0^R\cdots\int_0^R  e^{-it_1 B_1}\cdots e^{-it_{r+1}B_{r+1}} Xe^{it_1 B_1}\cdots e^{it_{r+1}B_{r+1}} dt_1\cdots dt_{r+1}\nonumber\\
&= \sum_{\lambda=(\lambda_1,\ldots,\lambda_{r+1})\in \Lambda}
 \big( \bigotimes_{\ell=1}^{r+1} P^{[\ell]}_{\lambda_\ell} \big)  X\big( \bigotimes_{\ell=1}^{r+1}P^{[\ell]}_{\lambda_{\ell}}\big)=\sum_{\lambda\in\Lambda} P_\lambda X P_\lambda.
\end{align*}
This together with \eqref{defP} shows \eqref{2.24.1}. \hfill\qed

\medskip

Combining \eqref{locresreas} with Corollary \ref{cor:deco} gives the following result. 

\begin{theorem}[Product state for locally resolved $\L$]
\label{thm3} 
Suppose the genericity condition \eqref{gencond} holds, that the system is well-connected \eqref{wellconn}, and that the boundary operators are locally resolved, \eqref{locres}. Then
\begin{equation}
\label{DeFinettifinalstate}
\lim_{N\rightarrow\infty}\rho_N^{(K)} = \tau^{\otimes K},\qquad \tau = \sum_{s=1}^\sss \Pi_s\rho\Pi_s.
\end{equation}
\end{theorem}

For a collectively resolved $\L$ we can represent the time ergodic average \eqref{2.24} as the space average over a torus. 
We say that the eigenvalues $a_1,\ldots,a_\sss$ of $A$ (see \eqref{m4}) are rationally independent modulo constants if for $m_1, \ldots, m_\sss \in \Z$, the relations 
\[
    m_1a_1 + \cdots + m_\sss a_\sss = 0
\text{\qquad and\qquad}  m_1 + \cdots + m_\sss = 0
\]  
imply that $m_1 = \cdots = m_\sss = 0$. The question of rational (in)dependence emerges naturally when expressing \eqref{2.24} as an integral over a torus via the ergodic theorem. We then obtain the following result, in which appear the single-vertex unitary operators
\begin{equation}
\label{Utheta}
 U(\btheta) = \sum_{s=1}^\sss e^{-i\theta_s}\Pi_s,\qquad \btheta = (\theta_1,\ldots,\theta_\sss)\in [0,2\pi]^\sss.
\end{equation}

\begin{theorem}[De Finetti representation for coarsely or collectively resolved $\L$]
\label{thm:definetti}
Suppose the generic case \eqref{gencond}, that the system is well-connected, \eqref{wellconn}, and that the eigenvalues of $A$ are rationally independent modulo constants. \ \\[-1.8ex]
\begin{itemize}
\item[{\rm (1)}] If the boundary operators are coarsely  resolved into $r+1$ subsystems as in \eqref{2.16}, then  
\begin{equation}
\label{m25.2}
\lim_{N\rightarrow\infty}  \rho_N^{(K)}=  e^{-iTH_\s} \big( \Gamma^{(1)}\otimes \cdots \otimes \Gamma^{(r+1)}\big) e^{iTH_\s},
\end{equation}
with
\begin{align}
\label{2.28-1}
\Gamma^{(\ell)} = \frac{1}{(2\pi)^{\sss}} \int_{[0,2\pi]^\sss}   \rho(\btheta)^{\otimes (K_\ell-K_{\ell-1})}\, d\btheta.
\end{align}
The integrand is the $(K_\ell-K_{\ell-1})$-fold tensor product of the single-vertex density matrix $\rho(\btheta) = U(\btheta)\rho U(\btheta)^*$.\ \\[-1ex]

\item[{\rm (2)}]   If the boundary operators are coarsely resolved into $r+1$ subsystems as in \eqref{2.16} and if within each subsystem either all possible edges are present or none are, while between each pair of distinct subsystems either all possible edges are present or none are, then
\begin{align}
\label{m25.3}
\lim_{N\rightarrow\infty}  \rho_N^{(K)} =\Gamma^{(1)}\otimes \cdots \otimes \Gamma^{(r+1)}.
\end{align}

\item[{\rm (3)}]  If the boundary operators are collectively resolved then \eqref{m25.2} holds with a single factor, 
\[
\lim_{N \to \infty}\rho_N^{(K)} = e^{-iTH_\s}\Gamma e^{iTH_\s},\qquad    \Gamma=\frac1{(2\pi)^\sss}
\int_{[0,2\pi]^\sss}\rho(\btheta)^{\otimes K}\,d\btheta.
\]
If further all pairs of vertices in $\s$ are linked by an edge, or none are, then $e^{-iTH_\s}\Gamma e^{iTH_\s} = \Gamma$. 
\end{itemize}
\end{theorem}

\begin{discussion}
\label{dis:312}
The local, coarse and collective resolutions introduced in Definition \ref{def:loccolres} lead to different structures in the  representation of $\P$. In the {\it purely} local case $\L = \{A_1, \ldots, A_K\}$, $\P$ acts independently on each vertex, so $\P[\rho^{\otimes K}]$ is a product state. In the coarse resolution, the $\P$ acts collectively within each subsystem but independently between subsystems, inducing a product of de Finetti mixtures. Collective resolution corresponds to a single action shared by all vertices.
\end{discussion}

\subsection{Proof of Theorem \ref{thm:definetti}} (1) Let $\bm=(m_1,\ldots,m_\sss) \in\mathbb Z^\sss$ and $\boldsymbol\omega=(\omega_1,\ldots,\omega_\sss)\in\mathbb R^\sss$ be given. By explicit integration,
\begin{equation}
\label{u1}
\lim_{R\rightarrow\infty} \frac1R\int_0^R e^{it \bm\cdot \bomega}dt = \left\{
\begin{array}{ll}
1 & \mbox{if $\bm\cdot \bomega =0$}\\
0 & \mbox{if $\bm\cdot \bomega\neq 0$}.
\end{array}
\right.
\end{equation}
On the other hand, with $\btheta = (\theta_1,\ldots,\theta_\sss)$ and $d\btheta=d\theta_1\cdots d\theta_\sss$ we have by periodicity,
\begin{equation}
\label{u2}
(2\pi)^{-\sss} \int_{[0,2\pi]^\sss}  \ e^{i \bm\cdot\btheta} d\btheta=
\left\{
\begin{array}{ll}
1 & \mbox{if $\bm=\boldsymbol 0$}\\
0 & \mbox{if $\bm\neq \boldsymbol 0$}.
\end{array}
\right.
\end{equation}
Suppose that the $\omega_j$ satisfy $\bm\cdot\bomega=0\Rightarrow \bm = \boldsymbol 0$. Then \eqref{u1}, \eqref{u2} give, for every $\bm \in \Z^\sss$,
\begin{equation}
\label{u3}
\lim_{R\rightarrow\infty} \frac1R\int_0^R e^{it \bm\cdot\bomega}dt = (2\pi)^{-\sss} \int_{[0,2\pi]^\sss} e^{i \bm\cdot\btheta} d\btheta.
\end{equation}
As the operators $\P^{(\ell)}$, \eqref{2.24.1} act nontrivially on different tensor factors for varying $\ell$ we only need to consider one of them. With $B_\ell=\sum_{K_{\ell-1}+1\le i\le K_\ell} A_i$ and setting $\kappa =K_\ell-K_{\ell-1}$ we have,
\begin{equation}
\label{u7}
e^{-itB_\ell} X e^{itB_\ell} = \sum_{\bs,\br\in \llbracket 1,\sss\rrbracket^{\kappa}} \frac{e^{-it(a_{r_1}+\cdots +a_{r_\kappa})}}{e^{-it(a_{s_1}+\cdots +a_{s_\kappa})}} \Pi_{\br} X\Pi_{\bs}
\end{equation}
where $\br=(r_1,\ldots,r_\kappa)$ and $\Pi_{\br}=\Pi_{r_1}\otimes\cdots\otimes \Pi_{r_\kappa}$ acts nontrivially on the tensor factors $K_{\ell-1}+1,
\ldots, K_\ell$. The notation for $\bs$ is analogous. 
For $\bs,\br$ fixed, the ratio in \eqref{u7} is a function of the eigenvalues $\boldsymbol a=(a_1,\ldots,a_\sss)$, of the form 
\begin{equation}
\label{u8}
\frac{e^{-it(a_{r_1}+\cdots +a_{r_\kappa})}}{e^{-it(a_{s_1}+\cdots +a_{s_\kappa})}} =e^{-it\bm(\br,\bs) \cdot \boldsymbol a},
\end{equation}
% for some $\bm(\br,\bs)\in \mathbb Z^\sss$. By \eqref{u3} we have for rationally independent $a_1,\ldots,a_\sss\in\mathbb R$, 
where $\bm(\br,\bs)=(m_1,\ldots,m_\sss)\in\mathbb Z^\sss$
is given by
\[
m_j=
|\{i:r_i=j\}|-|\{i:s_i=j\}|,
\qquad j=1,\ldots,\sss.
\]
Since $\br$ and $\bs$ both have $\kappa$ components, $\sum_{j=1}^\sss m_j=\kappa-\kappa=0$. Consequently, rational independence modulo constants
implies that
\[
\bm(\br,\bs)\cdot\boldsymbol a=0
\quad\Longleftrightarrow\quad
\bm(\br,\bs)=\boldsymbol 0.
\]
We may therefore apply \eqref{u3} to obtain
\begin{equation}
\label{u9}
\lim_{R\rightarrow\infty}\frac1R\int_0^R e^{-it\bm(\br,\bs) \cdot\boldsymbol a} dt =  (2\pi)^{-\sss} \int_{[0,2\pi]^\sss}   e^{-i \bm(\br,\bs) \cdot\btheta} d\btheta.
\end{equation}
Combining \eqref{u7} and \eqref{u9} yields,
\begin{align}
\lim_{R\rightarrow\infty} \P_R^{(\ell)}[\rho^{\otimes K}]&=  \sum_{\bs,\br\in \llbracket 1,\sss\rrbracket^\kappa} (2\pi)^{-\sss} \int_{[0,2\pi]^\sss}   e^{-i \bm(\br,\bs) \cdot\btheta} \, \Pi_{\br} \rho^{\otimes K}\Pi_{\bs}\, d\btheta\nonumber\\
&=\rho^{\otimes K_{\ell - 1}}\otimes\Big[ (2\pi)^{-\sss} \int_{[0,2\pi]^\sss}  \bigotimes_{j=K_{\ell-1}+1}^{K_\ell}  [\rho(\btheta)]_j\,  d\btheta\Big] \otimes \rho^{\otimes (K-K_\ell)} 
\label{u10}
\end{align}
where $[\rho(\btheta)]_j$ is the single site density operator $\rho(\btheta)=U(\btheta)\rho U(\btheta)^*$ ({\em cf.} \eqref{Utheta}) acting nontrivially on the $j$th tensor factor of $\h_\s$. We used \eqref{u7} and \eqref{u8} in the last step of \eqref{u10}. Combining \eqref{u10} with Corollary \ref{cor:deco} proves point (1) of Theorem \ref{thm:definetti}.
\smallskip

(2) Let
$I_\ell=\llbracket K_{\ell-1}+1,K_\ell\rrbracket$
and $B_\ell=\sum_{i\in I_\ell}A_i$.
For each $\ell$, let $c_\ell\in\{0,1\}$
indicate whether all possible edges within
$I_\ell$ are present. For $\ell<m$, let
$d_{\ell, m}\in\{0,1\}$ indicate whether all
possible edges between $I_\ell$ and $I_m$
are present. Then
\begin{equation}
\label{2.37}
H_\s=
\frac12\sum_{\ell=1}^{r+1}c_\ell
\Big(
B_\ell^2-\sum_{i\in I_\ell}A_i^2
\Big)
+
\sum_{\ell<m}d_{\ell, m}B_\ell B_m.
\end{equation}
The term $\frac12(B_\ell^2-\sum_{i\in I_\ell}A_i^2)$ is equal to $\sum_{(i,j)\in I_\ell, i<j} A_iA_j$, which is the contribution to $H_\s$ coming from the edges within the block $\ell$ of partition into subsystems. 
Since $\P[X]$ commutes with every $B_\ell$,
and all the operators $A_i$ commute, we obtain
\begin{align}
\label{3.39}
e^{-iTH_\s}\P[X]e^{iTH_\s}
=
e^{\frac{iT}{2}\sum_{\ell=1}^{r+1}
c_\ell\sum_{i\in I_\ell}A_i^2}
\P[X]
e^{-\frac{iT}{2}\sum_{\ell=1}^{r+1}
c_\ell\sum_{i\in I_\ell}A_i^2}.
\end{align}
Next we note that by \eqref{2.28-1}, each
$\Gamma^{(\ell)}$ is invariant under
conjugation by
$e^{it\sum_{i\in I_\ell}A_i^2}$ (any $t\in\mathbb R$).
Indeed, since
$A^2=\sum_{s=1}^\sss a_s^2\Pi_s$,
this conjugation amounts to translating
the torus variables according to
$\theta_s\mapsto\theta_s-ta_s^2$,
which leaves the torus integral
unchanged. Therefore \eqref{3.39} gives $e^{-iTH_\s}\P[\rho^{\otimes K}]
e^{iTH_\s}
=\P[\rho^{\otimes K}]$ and  \eqref{m25.3} follows from
\eqref{m25.2}.

The assertions of (3) follows by the same argument as (1), (2), now with a single block. This completes the proof of Theorem \ref{thm:definetti}. \hfill\qed

\subsection{Proof of Theorem \ref{thm2}}

We first state a preliminary result.

\begin{theorem}
\label{thm1}
The reduced state \eqref{m5} is given by
\begin{equation}
\label{2.32}
\rho_N^{(K)} = e^{-iTH_\s} \ \Phi_N(\rho^{\otimes K})\ e^{iTH_\s},
\end{equation}
with $H_\s$ as in \eqref{eq:ch2-H} and where $\Phi_N$ is the CPTP map on $\b(\h_\s)$ defined as
\begin{equation}
\label{m17.1}
\Phi_N(X) = \sum_{\lambda,\mu\in\Lambda} w_N(\lambda,\mu) P_\lambda XP_\mu,
\end{equation}
with
\begin{equation}
\label{w}
w_N(\lambda,\mu) = \prod_{\ell=1}^{|\L|} \Big[\psi\big(b_\ell(\lambda)-b_\ell(\mu)\big)\Big]^{\nu(\ell,N)}.
\end{equation}
\end{theorem}
We prove Theorem \ref{thm1} below and for now use it to show Theorem \ref{thm2}. For the `diagonal' term $\lambda=\mu$ in the double sum of \eqref{m17.1} we have $w_N(\lambda,\lambda)=1$, so those terms give the contribution $\sum_{\lambda\in\Lambda} P_\lambda XP_\lambda$. We estimate the off-diagonal as 
\begin{equation}
\label{m22}
\Big\| \sum_{\lambda,\mu\in\Lambda:\,  \lambda\neq \mu} w_N(\lambda,\mu) P_\lambda XP_\mu\Big\|_1\le \max_{\lambda,\mu: \lambda\neq \mu} |w_N(\lambda,\mu)| \sum_{\lambda,\mu\in\Lambda:\,  \lambda\neq \mu} \|P_\lambda XP_\mu\|_1.
\end{equation}
The H\"older inequality for Schatten norms gives for $X\ge0$,
$$
\|P_\lambda XP_\mu\|_1 \le \|P_\lambda \sqrt X\|_2 \|\sqrt X P_\mu\|_2 = \sqrt{\Tr(P_\lambda X)\Tr(P_\mu X)}.
$$
Then by the Cauchy-Schwarz inequality for sums,
\begin{align}
\sum_{\lambda,\mu\in\Lambda:\,  \lambda\neq \mu} \|P_\lambda XP_\mu\|_1&\le |\Lambda|\Big(\sum_{\lambda,\mu\in\Lambda:\,  \lambda\neq \mu} \|P_\lambda XP_\mu\|_1^2\Big)^{1/2}\nonumber\\
&\le |\Lambda| \Big(\sum_{\lambda,\mu\in\Lambda:\,  \lambda\neq \mu} \Tr(P_\lambda X)\Tr(P_\mu X)\Big)^{1/2}\le |\Lambda| \Tr(X).
\label{m17'}
\end{align}
Combining \eqref{m17'} with \eqref{m22} and Theorem \ref{thm1} we obtain
\begin{equation}
\label{m23.1}
\big\| \rho_N^{(K)} - e^{-iTH_\s} \big(\sum_{\lambda\in\Lambda}  P_\lambda (\rho^{\otimes K})P_\lambda\big) e^{iTH_\s}\big\|_1 \le |\Lambda| \max_{\lambda,\mu: \lambda\neq \mu} |w_N(\lambda,\mu)|.
\end{equation}
As  $\psi(0)=1$ we have for $\lambda,\mu\in\Lambda$, 
\begin{align}
|w_N(\lambda,\mu)| &= \prod_{\ell=1}^{|\L|} \big|\psi(b_\ell(\lambda)-b_\ell(\mu))\big|^{\nu(\ell,N)}  = \prod_{\ell : b_\ell(\lambda)\neq b_\ell(\mu) } \big|\psi(b_\ell(\lambda)-b_\ell(\mu))\big|^{\nu(\ell,N)}\le q^{\ve_K(N)}.
\label{m24}
\end{align}
Combining \eqref{m23.1} and \eqref{m24} yields \eqref{m23}. 
\medskip

{\bf Proof of Theorem \ref{thm1}.} The operators $A_i$ and $A_j$ commute for all $i,j$ and so according to \eqref{m10} the propagator \eqref{m7} splits into the product of three commuting unitaries,
\begin{equation}
\label{m8}
U(N) =  e^{-iTH_\s} e^{-iT H_{\r,N}} e^{-iT H_{\cc,N}}. 
\end{equation}
In \eqref{m5} the conjugation with $e^{-i TH_\s}$ can be taken out of the partial trace and the conjugation with $e^{-iTH_{\r,N}}$  disappears due to the cyclicity of the partial trace for operators acting on $\h_\r$ only, and since $e^{-iTH_{\r,N}}$ is unitary. 
Thus we have 
\begin{equation}
\label{m13}
\rho^{(K)}_N = e^{-iTH_\s}\big[ \Tr_\r \big(e^{-iTH_{\cc,N}} (\rho^{\otimes N}) e^{iTH_{\cc,N}}\big)\big] e^{iTH_\s}.
\end{equation}
Let
\begin{equation}
\label{def:Phi}
\Phi_N(X) = \Tr_\r\big(e^{-iTH_{\cc,N}} (X\otimes \rho^{\otimes(N-K)}) e^{iTH_{\cc,N}}\big), \qquad X\in\b(\h_\s).
\end{equation}
We decompose,
\begin{equation}
\label{m14}
H_{\cc,N} = \sum_{(i,j)\in E_{\cc,N}}A_iA_j = \sum_{j=K+1}^N \Big(\sum_{i: (i,j)\in E_{\cc,N}}A_i\Big) A_j = \sum_{j=K+1}^N  L(j) A_j,
\end{equation}
where the boundary operator $L(j)$ is given in \eqref{eq:bndryop}. The representation \eqref{m14} is suitable for the partial trace over $\h_\r$ (factors $j=K+1,\ldots,N$ of \eqref{m1}) because $L(j)$ acts nontrivially only on $\h_\s$ (factors $1\le j\le K$) and $A_j$ acts on the factor $j\ge K+1$. As all $L(j)$ and all $A_j$ commute, it follows from \eqref{def:Phi} that 
\begin{equation}
\label{m15}
\Phi_N(X) = \Tr_\r \left[\left(\prod_{j
=K+1}^N e^{-iTL(j)A_j}\right) \big(X\otimes\rho^{\otimes(N-K)}\big) \left(\prod_{j'
=K+1}^Ne^{iTL(j')A_{j'}}\right)\right].
\end{equation}
To further reduce this expression we diagonalize the $A_j$ by using \eqref{m4}. 
Inserting the spectral resolution \eqref{eq:specres} into \eqref{m15} yields,
\begin{align}
\Phi_N(X) =& \sum_{s_{K+1},\ldots,s_N=1}^\sss\ \sum_{s'_{K+1},\ldots,s'_N=1}^\sss\nonumber\\
&\times \Tr_\r \left[\left(\prod_{j
=K+1}^N e^{-iTL(j)a_{s_j}}\right) \big(X\otimes\Pi_{s_{K+1}}\rho\Pi_{s'_{K+1}}\otimes \cdots \otimes \Pi_{s_N}\rho\Pi_{s'_N}\big) \left(\prod_{j'
=K+1}^N e^{iTL(j')a_{s'_{j'}}}\right)\right]\nonumber\\
=& \sum_{s_{K+1},\ldots,s_N=1}^\sss p_{s_{K+1}}\cdots p_{s_N}  e^{-iT\sum_{j=K+1}^N a_{s_j}L(j)}  X e^{iT\sum_{j'=K+1}^N a_{s_{j'}}L(j')},
\label{m16}
\end{align}
where we used that $\Tr\big(\Pi_s \rho\Pi_{s'}\big) = \delta_{s,s'} \Tr(\Pi_s\rho) \equiv p_s\, \delta_{s,s'}$. Let $\{\xi_j\}_{j=K+1}^N$ be $N-K$ i.i.d. random variables defined after \eqref{eq:ps}. Then \eqref{m16} is the expectation,
\begin{equation}
\label{m17}
\Phi_N(X) = \E\big[ e^{-iT \sum_{j=K+1}^N \xi_j L(j)}Xe^{iT \sum_{j=K+1}^N \xi_j L(j)}\big].
\end{equation}
Each $L(j)$ belongs to $\L$ and is one of the $B_\ell$ (see \eqref{Lstar}). Hence, we have
$$
\sum_{j=K+1}^N \xi_j L(j) = \sum_{\ell=1}^{|\L|} B_\ell \sum_{\substack{j=K+1,\ldots,N:\\ L(j)=B_\ell}} \xi_j
$$
and inserting the resolution of the identity \eqref{resid} into \eqref{m17} yields,
\begin{equation}
\label{m19}
\Phi_N(X) = \sum_{\lambda,\mu\in\Lambda} \E\Big[ \exp\big({-iT \sum_{\ell=1}^{|\L|} \big(b_\ell(\lambda)-b_\ell(\mu)\big) \sum_{\substack{j=K+1,\ldots,N:\\ L(j)=B_\ell}} \xi_j }\big)\Big] P_\lambda XP_\mu.
\end{equation}
All the $\xi_j$ are i.i.d. random variables and so the expectation splits into a product,
\begin{align}
\E\Big[ \exp\big({-iT \sum_{\ell=1}^{|\L|} \big(b_\ell(\lambda)-b_\ell(\mu)\big) \sum_{\substack{j=K+1,\ldots,N:\\ L(j)=B_\ell}} \xi_j }\big)\Big] =\prod_{\ell=1}^{|\L|} \Big(\E\big[e^{-iT(b_\ell(\lambda)-b_\ell(\mu))\xi}\big]\Big)^{\nu(\ell,N)},
\end{align}
where $\nu(\ell,N)$ is the number of indices $j=K+1,\ldots,N$ such that $L(j)=B_\ell$. This completes the proof of Theorem \ref{thm1}.\hfill \qed

\bigskip
\medskip

\section{Scaling regimes}
\label{ch:regimes}
\medskip

In many-body quantum theory one often considers scaled interactions where the Hamiltonian \eqref{m10} is replaced by a rescaled version, in order to reflect certain physical properties. The Hamiltonians we now consider are of the form 
\begin{equation}
\label{s1}
H_N= \epsilon_\s(N) H_\s+\epsilon_\r(N)H_{\r,N}+\epsilon(N) H_{\cc,N},
\end{equation}
where $H_\s$, $H_{\r,N}$ and $H_{\cc,N}$ are given in \eqref{eq:ch2-H} and where $\epsilon_\s(N),\epsilon_\r(N),\epsilon(N)\ge 0$ are generally $N$-dependent scaling factors. 

\begin{example} 
In mean-field models, the `interaction energy' $H_{\cc,N}$ is scaled so that it is of the same size as the system energy $H_\s$ for $N\rightarrow\infty$. This is achieved by choosing $\epsilon_\s(N)=1=\epsilon_\r(N)$ and $\epsilon(N)\propto |E_{\cc,N}|^{-1}$, provided $|E_{\cc, N}| > 0$. 
\end{example}

Taking $H_N$ as in \eqref{s1} yields a unitary propagator $U(N)$, \eqref{m7} a corresponding density matrix $\rho_N$, \eqref{m2} and a reduced $K$-vertex density matrix $\rho^{(K)}_N$, \eqref{m5}. Define the unitary maps on $\h_\s$,
\begin{align}
U_0(N) &= e^{-iT \epsilon_\s(N) H_\s}\label{U_0} \\
U_1(N) &= e^{-iT\,\E[\xi]\epsilon(N)\sum_{\ell=1}^{|\L|} \nu(\ell,N)B_\ell} 
\label{3.24-1}
\end{align}
and the CPTP map on $\b(\h_\s)$,
\begin{align}
\label{3.241}
\mathcal M_N[X] = \sum_{\lambda,\mu\in\Lambda} u_N(\lambda,\mu)P_\lambda X P_\mu,\qquad X\in\b(\h_\s),
\end{align}
where
\begin{equation}
u_N(\lambda,\mu) = \exp\Big[-\tfrac12 T^2 {\rm Var}[\xi]\epsilon(N)^2 \sum_{\ell=1}^{|\L|}  \nu(\ell,N)\big(b_\ell(\lambda)-b_\ell(\mu)\big)^2 \Big]
\label{s15.1}
\end{equation}
and ${\rm Var}[\xi] = \E[\xi^2] - \E[\xi]^2\ge0$ is the variance of $\xi$. The fact that $\mathcal M_N$ is CPTP follows from (is equivalent to) the positive-semidefiniteness of the matrix $u_N$ with elements $(u_N)_{\lambda,\mu} = u_N(\lambda,\mu)$, {\it cf.} \cite[Proposition~1]{MarcantoniMerkli}. We obtain the following analogue of Theorem \ref{thm2}. 

\begin{theorem}
\label{thm3.1}
Suppose that $\epsilon(N) K T \|A\|^2\le 3/20$. Then 
\begin{align}
\label{sss}
\Big\|&\rho_N^{(K)} -  U_0(N)U_1(N)\, \mathcal M_N \big[\rho^{\otimes K}\big] \, U_1(N)^*U_0(N)^*\Big\|_1 \le |\Lambda|\max_{\lambda,\mu\in\Lambda}\big|e^{\epsilon(N)^3R_N(\lambda,\mu)}-1\big|,
\end{align}
where the remainder satisfies the bound 
\begin{equation}
\label{Rem}
|R_N(\lambda,\mu)|\le 16\, \big(KT\|A\|^2\big)^3 (N - K).
\end{equation}
\end{theorem}
\begin{remark} Our proof shows a sharper bound than  \eqref{Rem}. For $\lambda,\mu\in\Lambda$, $\lambda\neq \mu$, we define
 \begin{equation}
 \label{s30}
 J(\lambda,\mu) = \big\{\ell : b_\ell(\lambda)\neq b_\ell(\mu)\big\},
 \end{equation}
 which is not the empty set (see \eqref{sep}). Our proof gives the bound 
 \begin{equation}\label{Rem-sharp}
 |R_N(\lambda,\mu)|\le 16\, \big(KT\|A\|^2\big)^3 \sum_{\ell\in J(\lambda,\mu)}\nu(\ell,N).
 \end{equation}
\end{remark}

\begin{discussion}
\leavevmode\par
\begin{itemize}
    \item $U_0(N)$ and $U_1(N)$ are unitary. Those operators give an oscillatory contribution, while the CPTP map $\mathcal M_N$ drives decoherence (relative to the blocks defined by the $P_\lambda$). The operations of conjugating by $U_0(N)$, $U_1(N)$ and the action of $\mathcal M_N$ commute.
    \item The phase of the oscillatory part induced by the environment interaction, $U_1(N)$, is proportional to the average $\E[\xi]\in\mathbb R$. The decoherence operator $\mathcal M_N$ has an exponent proportional to the variance ${\rm Var}[\xi]\ge 0$ (see \eqref{s15.1}). We have ({\it cf.}~\eqref{charfunct}, \eqref{gencond})
    $$
    {\rm Var}[\xi]=0 \Leftrightarrow \big\{a_s=a={\rm const.}\  \forall s \ {\rm s.t.}\  p_s\neq 0\big\} \Rightarrow \psi(\Delta) = e^{-iT\Delta a} \Rightarrow q=1,
    $$
    in which case we are not in the generic regime. In fact, ${\rm Var}[\xi]=0\Rightarrow \mathcal M_N[X]=X$ in which case there is no decoherence.
\end{itemize}
\end{discussion}

\begin{lemma}[Integral representation of $\mathcal M_N$]
\label{lem:intrep}
We have for all $X\in\mathcal B(\h_\s)$,
\begin{align}
\label{l0}
\mathcal M_N[X] = \pi^{-|\L|/2}\int_{{\mathbb R}^{|\L|}} U_N(x_1,\ldots,x_{|\L|}) \, X  \, U_N^*(x_1,\ldots,x_{|\L|})  e^{-(x^2_1+\cdots +x^2_{|\L|})}dx_1\cdots dx_{|\L|},
\end{align}
where 
\begin{align}
\label{10.01}
U_N(x_1,\ldots,x_{|\L|}) =e^{iT \sqrt{2{\rm Var}[\xi]}\, \sum_{\ell=1}^{|\L|} \epsilon(N) \sqrt{\nu(\ell,N)} B_\ell x_\ell},\qquad (x_1,\ldots,x_{|\L|})\in\mathbb R^{|\L|}.
\end{align}
\end{lemma}
\medskip
\begin{discussion}
\label{dis4.6}
The local, coarse and collective resolutions introduced in Definition \ref{def:loccolres} lead to different structures in the Gaussian representation of $\mathcal M_N$ given in Lemma \ref{lem:intrep}. In the local case $\L = \{A_1, \ldots, A_K\}$, the Gaussian averages act independently on each vertex (each $B_\ell$ in \eqref{10.01} is one of the $A_i$), so $\mathcal{M}_N[\rho^{\otimes K}]$ is a product state. In the coarse resolution, the averages act collectively within each subsystem but independently between subsystems, inducing a product of Gaussian de Finetti mixtures. Collective resolution corresponds to a single Gaussian average shared by all vertices. 
\end{discussion}

{\bf Proof of Lemma \ref{lem:intrep}.}  Set for short 
 $\theta_\ell(N) = \frac12 T^2{\rm Var}[\xi]\epsilon(N)^2\nu(\ell,N)$. Then \eqref{s15.1} gives, 
\begin{align}
\label{l1}
u_N(\lambda,\mu) = \exp\big[ -\sum_{\ell=1}^{|\L|} \theta_\ell(N) (b_\ell(\lambda)-b_\ell(\mu))^2\big].
\end{align}
For $\delta\in\mathbb R$, $\theta\ge 0$, we have $e^{-\theta \delta^2}
=
\frac{1}{\sqrt{\pi}}
\int_{\mathbb R}e^{2i\sqrt{\theta}\delta x}e^{-x^2}\,dx$. 
It follows from \eqref{l1} that 
\begin{align}
\label{l2}
u_N(\lambda,\mu) = \pi^{-|\L|/2}\int_{{\mathbb R}^{|\L|}} e^{2i\sum_{\ell=1}^{|\L|} \sqrt{\theta_\ell(N)}\{b_\ell(\lambda)-b_\ell(\mu)\} x_\ell} e^{-(x^2_1+\cdots +x^2_{|\L|})} dx_1\cdots dx_{|\L|}.
\end{align}
Next we use that $B_\ell P_\lambda=b_\ell(\lambda)P_\lambda$,
$$
e^{2i\sum_{\ell=1}^{|\L|} \sqrt{\theta_\ell(N)}\{b_\ell(\lambda)-b_\ell(\mu)\} x_\ell} P_\lambda X P_\mu = e^{2i\sum_{\ell=1}^{|\L|} \sqrt{\theta_\ell(N)} B_\ell x_\ell} \big( P_\lambda X P_\mu \big) e^{-2i\sum_{\ell=1}^{|\L|} \sqrt{\theta_\ell(N)} B_\ell x_\ell}.
$$
Summing the latter expression over $\lambda,\mu\in\Lambda$ gives $e^{2i\sum_{\ell=1}^{|\L|} \sqrt{\theta_\ell(N)} B_\ell x_\ell} \, X \,  e^{-2i\sum_{\ell=1}^{|\L|} \sqrt{\theta_\ell(N)} B_\ell x_\ell}$. 
Thus,
\begin{align}
\mathcal M_N[X] = \pi^{-|\L|/2}\int_{{\mathbb R}^{|\L|}} e^{2i\sum_{\ell=1}^{|\L|} \sqrt{\theta_\ell(N)} B_\ell x_\ell} \, X  \, e^{-2i\sum_{\ell=1}^{|\L|} \sqrt{\theta_\ell(N)} B_\ell x_\ell}  e^{-(x^2_1+\cdots +x^2_{|\L|})}dx_1\cdots dx_{|\L|},
\nonumber
\end{align}
which is \eqref{l0}. This completes the proof of Lemma \ref{lem:intrep}.\hfill \qed
\bigskip

{\bf Proof of Theorem \ref{thm3.1}.} We first state a result analogous to Theorem \ref{thm1}. 
\begin{theorem}
\label{thm1'}
For $H_N$ given in \eqref{s1} the reduced state \eqref{m5} is given by
\begin{equation}
\label{2.32'}
\rho_N^{(K)} = e^{-iT \epsilon_\s(N) H_\s} \ \widetilde\Phi_N(\rho^{\otimes K})\ e^{iT \epsilon_\s(N) H_\s},
\end{equation}
with $H_\s$ as in \eqref{eq:ch2-H} and where $\widetilde\Phi_N$ is the CPTP map on $\b(\h_\s)$ defined as
\begin{equation}
\label{m17.1'}
\widetilde\Phi_N(X) = \sum_{\lambda,\mu\in\Lambda} \widetilde w_N(\lambda,\mu) P_\lambda XP_\mu,\qquad 
\widetilde w_N(\lambda,\mu) = \prod_{\ell=1}^{|\L|} \Big[\psi\Big(\epsilon(N)\big( b_\ell(\lambda)-b_\ell(\mu)\big)\Big)\Big]^{\nu(\ell,N)}
\end{equation}
and where $\psi$ is given in \eqref{charfunct}. 
\end{theorem}

The proof of Theorem \ref{thm1'} follows directly from the proof of Theorem \ref{thm1}. This is so because the unitary implementing the dynamics still splits into three parts (system, reservoir, interaction) as before, each one now scaled. Let us return to the proof proper of Theorem \ref{thm3.1}.

We have in mind that $\epsilon(N)\rightarrow 0$ for large $N$ and so we expand the characteristic function $\psi$, \eqref{charfunct} for small arguments $\delta$.  Iterating the equality $e^{-iT\delta \xi} = 1-iT\xi\int_0^\delta e^{-iTs\xi}ds$ three times, taking the expectation, and using the bound
\begin{equation}
\label{s17}
\big|\E[\xi^k]\big|\le \E\big[|\xi|^k\big]\le \|A\|^k,\qquad k\in\mathbb N,
\end{equation}
yields 
\begin{equation}
\label{charfunct'}
\Big| \psi(\delta)-\big(  1-iT\delta\,  \E[\xi]-\tfrac12 T^2\delta^2\E[\xi^2] \big)\Big|\le \frac{T^3 |\delta|^3}{6}  \E\big[|\xi|^3\big]\le  \frac{T^3 |\delta|^3}{6} \|A\|^3,\qquad \delta\in\R.
\end{equation}
For 
\begin{equation}
\label{bd1}
0\le|\delta| \,T \|A\| \le \frac{3}{10}
\end{equation}
both $\psi\equiv\psi(\delta)$ and 
\begin{equation}
\label{eq:tbd}
\alpha =  1-\beta,\qquad \text{where\ \  $\beta = iT\delta\,  \E[\xi]+\tfrac12 T^2\delta^2\E[\xi^2]$,}
\end{equation}
are in the closed disk of radius $\frac12$ centered at $z=1\in\C$. We expand the (complex) logarithm, $\log z=\log w
+(z-w)\int_0^1\frac{dt}{w+t(z-w)}$ to arrive at the bound (use $\min_{0\le t\le 1}|\psi+t(\alpha-\psi)|\ge \frac12$ and \eqref{charfunct'})
\begin{align}
\big| \log\psi -\log \alpha \big| \le 2|\psi-\alpha| \le\frac{T^3 |\delta|^3}{3} \|A\|^3.
\end{align}
Then for $\nu\in\mathbb N$,
\begin{align}
\label{3.7}
\psi^\nu = e^{\nu\log\psi}= e^{\nu(\log\alpha +|\delta|^3R_1)},\qquad |R_1|\le \frac13 (T \|A\|)^3.
\end{align}
Set 
\begin{equation}
\label{eps}
\delta_\ell\equiv\epsilon(N)\big(b_\ell(\lambda)-b_\ell(\mu)\big),\qquad \alpha_\ell=  1-\beta_\ell,\qquad \beta_\ell=iT\delta_\ell\,  \E[\xi]+\tfrac12 T^2\delta_\ell^2\E[\xi^2]
\end{equation}
We have for all $\lambda,\mu\in\Lambda$ and $\ell=1,\ldots,|\L|$,
\begin{align}
\label{bnd'}
|b_\ell(\lambda)-b_\ell(\mu)|\le 2K\|A\|,
\end{align}
because each $b_\ell(\lambda)$ is an eigenvalue of an operator $B=\sum_{j\in F}A_j\in\L$, with $
F\subseteq\llbracket1,K\rrbracket$, so that we have $|b_\ell(\lambda)|\le K\|A\|$.
We assume that 
\begin{align}
\label{assume}
\epsilon(N)KT\|A\|^2\le \frac{3}{20}.
\end{align}
This implies that \eqref{bd1} holds for all $\delta_\ell$. Recall the definition \eqref{s30}. We obtain from \eqref{m17.1'}, \eqref{3.7} and \eqref{eps},
\begin{equation}
\label{s14} 
\widetilde w_N(\lambda,\mu) = \prod_{\ell\in J(\lambda,\mu) } e^{\nu(\ell,N)\log(\psi(\delta_\ell)) }= \exp\Big[ \sum_{\ell\in J(\lambda,\mu)} \nu(\ell,N)\big(\log\big(\alpha_\ell)+\epsilon(N)^3R_2(\ell)\big) \Big]
\end{equation}
where (use \eqref{bnd'})
\begin{equation}
\label{R2}
|R_2(\ell)|\le\frac{T^3\|A\|^3}{3}\max_\ell |b_\ell(\lambda)-b_\ell(\mu)|^3\le  \frac83 \big(K T \|A\|^2\big)^3.
\end{equation}
Next we expand
\begin{align}
\log(\alpha_\ell)  = \log(1-\beta_\ell) &=-\beta_\ell-\frac12 \beta_\ell^2 -\beta_\ell^3\int_0^1\frac{(1-t)^2}{(1-t\beta_\ell)^3}\,dt\nonumber\\
&= -\beta_\ell-\frac12 \beta_\ell^2 +\epsilon(N)^3 R_3(\ell),\qquad |R_3(\ell)|\le \frac{1}{2}\epsilon(N)^{-3}|\beta_\ell|^3.
\label{s8}
\end{align}
Indeed, setting $x_\ell=|\delta_\ell|T\|A\|$,
we have by \eqref{assume}, \eqref{eps} and \eqref{bnd'}
\[ 
x_\ell\le 2\epsilon(N)KT\|A\|^2\le \frac3{10}.
\]
Since
\[
0\le\Re\beta_\ell
=\frac12T^2\delta_\ell^2\E[\xi^2]
\le\frac12x_\ell^2
\le\frac9{200},
\]
we obtain $|1-t\beta_\ell|\ge 1-t\Re\beta_\ell
\ge 191/200>9/10$ for $0\le t\le1$.
Consequently,
\[
\int_0^1
\frac{(1-t)^2}{|1-t\beta_\ell|^3}\,dt
\le\frac1{3(9/10)^3}<\frac12,
\]
which gives the remainder estimate in \eqref{s8}.
Furthermore, $|\beta_\ell|
\le x_\ell+\frac12x_\ell^2
\le\frac65x_\ell$, so
\begin{align}
\label{R3}
|R_3(\ell)|
&\le\frac12\epsilon(N)^{-3}
\left(\frac65|\delta_\ell|T\|A\|\right)^3 \le4\left(\frac65\right)^3
(KT\|A\|^2)^3
<7(KT\|A\|^2)^3.
\end{align}
Next we expand the main term on the right-hand side of \eqref{s8},
\begin{align}
\beta_\ell +\tfrac12 \beta_\ell^2 = iT\delta_\ell\,  \E[\xi] +\tfrac12 T^2\delta^2_\ell{\rm Var}[\xi] +\epsilon(N)^3 R_4(\ell),
\label{s16}
\end{align}
where $R_4(\ell) = \frac{i}{2}T^3(b_\ell(\lambda) - b_\ell(\mu))^3\E[\xi]\E[\xi^2] + \frac{1}{8}\epsilon(N)T^4(b_\ell(\lambda) - b_\ell(\mu))^4 (\E[\xi^2])^2$. As this is split into real and imaginary parts we have,
\[
|R_4(\ell)|^2 = \big(\tfrac{1}{2}T^3(b_\ell(\lambda) - b_\ell(\mu))^3\E[\xi]\E[\xi^2]\big)^2 + \big(\tfrac{1}{8}\epsilon(N)T^4(b_\ell(\lambda) - b_\ell(\mu))^4 (\E[\xi^2])^2\big)^2.
\] 
Using $|b_\ell(\lambda) - b_\ell(\mu)| \le 2K\|A\|$ and \eqref{s17}, the cubic term is at most $4(KT\|A\|^2)^3$. For the quartic term, $\frac{1}{8}\epsilon(N)T^4(b_\ell(\lambda) - b_\ell(\mu))^4 (\E[\xi^2])^2 \le 2\epsilon(N) (KT\|A\|^2)^4\le\frac{3}{10}(KT\|A\|^2)^3$, where we used \eqref{assume}. Combining these two terms we get 
\begin{align}
\label{R4}
|R_4(\ell)| \le \sqrt{16+\tfrac{9}{100}}\ (KT\|A\|^2)^3 .
\end{align}
Combining \eqref{s16}, \eqref{s8} and \eqref{s14} gives,
\begin{align}
\label{s15} 
\widetilde w_N(\lambda,\mu)
& = \exp\Big[ -iT \,\E[\xi]\epsilon(N)\sum_{\ell=1}^{|\L|}\nu(\ell,N) (b_\ell(\lambda)-b_\ell(\mu))\Big]\, u_N(\lambda,\mu) \, e^{\epsilon(N)^3 R_N(\lambda,\mu)},
\end{align}
where $u_N(\lambda,\mu)$ is given in \eqref{s15.1} 
and the remainder $R_N(\lambda,\mu) = \sum_{\ell\in J(\lambda,\mu)}\nu(\ell,N)[R_2+R_3-R_4]$ satisfies the bound \eqref{Rem}.  Since $b_\ell(\mu)P_\mu = B_\ell P_\mu$ (see \eqref{Beval}) we use the expression \eqref{s15} in the definition of $\widetilde\Phi_N$ \eqref{m17.1'}, to obtain
\begin{align}
\widetilde\Phi_N(X) &= U_1(N) \Big( \sum_{\lambda,\mu\in\Lambda}e^{\epsilon(N)^3 R_N(\lambda,\mu)} \, u_N(\lambda,\mu) P_\lambda XP_\mu\Big) U_1(N)^*\nonumber\\
&=U_1(N) \Big( \sum_{\lambda,\mu\in\Lambda} \, u_N(\lambda,\mu) P_\lambda XP_\mu\Big)U_1(N)^* + \eta_N[X],
\label{s18}
\end{align}
where $U_1(N)$ is defined in \eqref{3.24-1} 
and where $\eta_N$ is a linear operator on $\b(\h_\s)$ satisfying ($X\in \b(\h_\s)$)
\begin{align*}
\|\eta_N[X]\|_1 &=\|\sum_{\lambda,\mu\in\Lambda}(e^{\epsilon(N)^3R_N(\lambda,\mu)}-1) u_N(\lambda,\mu)P_\lambda XP_\mu\|_1\\
&\le \big[\max_{\lambda,\mu\in\Lambda}\big|e^{\epsilon(N)^3R_N(\lambda,\mu)}-1\big| \big]\sum_{\lambda,\mu\in\Lambda} \|P_\lambda XP_\mu\|_1\le |\Lambda|\max_{\lambda,\mu\in\Lambda}\big|e^{\epsilon(N)^3R_N(\lambda,\mu)}-1\big|,
\end{align*}
by \eqref{m17'} and for $X\ge 0$,\  $\Tr(X)=1$. Using Theorem \ref{thm1'} and collecting the bounds \eqref{R2}, \eqref{R3} and \eqref{R4} shows \eqref{sss} and \eqref{Rem}.  This concludes the proof of Theorem \ref{thm3.1}.\hfill\qed

\subsection{Scaling phases} There appear two scalings, $\epsilon(N)\nu(\ell,N)$ and $\epsilon(N)^2\nu(\ell,N)$ in the oscillatory phase term $U_1(N)$, \eqref{3.24-1} and the decoherence term $\mathcal M_N$, \eqref{3.241}, \eqref{s15.1}. The asymptotic form of $\rho_N^{(K)}$ depends on their behaviour as $N\rightarrow\infty$. We consider scalings such that
\begin{equation}
\label{3.28}
\lim_{N\rightarrow\infty}\epsilon(N)^3 \nu(\ell,N) = 0, \qquad \forall \ell =1,\ldots,|\L|.
\end{equation}
This implies that the remainder contribution in Theorem \ref{thm3.1} disappears for large $N$, since by \eqref{Rem-sharp} we have $\lim_{N\rightarrow\infty}\epsilon(N)^3R_N(\lambda,\mu)=0$, and consequently the right-hand side of \eqref{sss} tends to zero. The condition \eqref{3.28} implies that $\epsilon(N)\rightarrow 0$ and since we have $\nu(\ell,N)\le N - K$, \eqref{3.28} holds automatically if  $\epsilon(N)^3N\rightarrow 0$. Recall the definition \eqref{s30} of $J(\lambda,\mu)$.

\begin{theorem}
\label{thm3.4}
Assume that \eqref{3.28} holds and that $\rm Var[\xi]\neq 0$.

\begin{itemize}
\item[1.] Suppose that for all $\lambda,\mu\in \Lambda$ with  $\lambda\neq\mu$, $\forall \ell\in J(\lambda,\mu)$, the limit $\lim_{N\rightarrow\infty} \epsilon(N)^2\nu(\ell,N) =\nu_\ell\in[0,\infty)$ exists. Then 
\begin{align}
\label{aa}
\lim_{N\rightarrow\infty} U_1(N)^*U_0(N)^* \big(\rho_N^{(K)}\big) U_0(N)U_1(N)=  \P\big[\rho^{\otimes K}\big] +  \P_f\big[\rho^{\otimes K}\big],
\end{align}
where $\P$ is given in \eqref{defP} and where
\begin{align}
\label{3.31}
\P_f[X]=\!\!\sum_{\substack{\lambda,\mu\in\Lambda\\ \lambda\neq \mu}} \exp\Big[-\tfrac12 T^2 {\rm Var}[\xi]\,  \displaystyle \sum_{\ell=1}^{|\L|}  \nu_\ell\big(b_\ell(\lambda)-b_\ell(\mu)\big)^2 \Big] P_\lambda XP_\mu.
\end{align} 
\item[2.] Suppose for all $\lambda,\mu\in \Lambda$ with  $\lambda\neq\mu$, $\exists\ell\in J(\lambda,\mu)$ s.t.
$\lim_{N\rightarrow\infty}\epsilon(N)^2\nu(\ell,N)=\infty$, then \eqref{aa} holds with $\P_f=0$.
\end{itemize}
\end{theorem}
\medskip

{\bf Proof of Theorem \ref{thm3.4}.} 
\eqref{3.28} and \eqref{Rem-sharp} imply that $\lim_{N\rightarrow\infty}\epsilon(N)^3R_N(\lambda,\mu)=0$ and we obtain from \eqref{sss} that 
\begin{align}
\label{3.271}
\lim_{N\rightarrow\infty} \Big\|&\rho_N^{(K)} -  U_0(N)U_1(N)\, \mathcal M_N \big[\rho^{\otimes K}\big] \, U_1(N)^*U_0(N)^*\Big\|_1=0.
\end{align}
The terms with $\lambda=\mu$ in the sum \eqref{3.241} give $\P[X]$, as $u_N(\lambda,\lambda)=1$. Next let $\lambda\neq \mu$. Under condition 1.~of Theorem \ref{thm3.4} we have  $\lim_{N\rightarrow\infty} u_N(\lambda,\mu) = \exp\big[-\tfrac12 T^2 {\rm Var}[\xi]\,   \sum_{\ell=1}^{|\L|}  \nu_\ell\big(b_\ell(\lambda)-b_\ell(\mu)\big)^2 \big]$, which gives the term $\P_f[X]$ in \eqref{aa}. Under condition 2.~of Theorem \ref{thm3.4} there is an $\ell$ such that $(b_\ell(\lambda)-b_\ell(\mu))^2>0$ and $\epsilon(N)^2\nu(\ell,N)\rightarrow\infty$. Then $\lim_{N\rightarrow\infty}u_N(\lambda, \mu)=0$. This means that only the diagonal terms $\lambda=\mu$ survive in $\mathcal M_N$ in the limit $N\rightarrow\infty$. \hfill \qed

\begin{discussion}
\leavevmode\par
\begin{itemize}
    \item The decoherence effect due to the interaction with the reservoir is encoded in $\P$ and $\P_f$, while $U_1(N)$ describes the phase oscillations induced by the interaction with the reservoir. The term $U_0(N)$ is the system dynamics without the interaction with the reservoir (see \eqref{U_0}). 
    \item If $\s$ is not well-connected then the situation 2. in Theorem \ref{thm3.4} does not occur (see \eqref{equivalence}) and so for generic $\rho$ the term $\P_f[\rho^{\otimes K}]$ does not vanish. Coherences generically survive the limit $N\rightarrow\infty$. If $\s$ is well-connected, then for all $\lambda\neq\mu$ $\exists \ell\in J(\lambda,\mu)$ such that $\nu(\ell,N)\rightarrow\infty$ as $N\rightarrow\infty$. Whether or not coherences survive depends on the growth rate of $\nu(\ell,N)$ relative to $\epsilon(N)^2$.
    \item In the centered case $\E[\xi]=0$ we have $U_1(N)=\mathbf 1$ (see \eqref{3.24-1}). For $\E[\xi]\neq 0$, $U_1(N)$ might converge even if the exponent does not,  as oscillations in $N$ by multiples of $2\pi i$ are compensated by the periodicity of the exponential. We disregard this non-generic parameter regime. A sufficient condition for $U_1(N)$ to converge as $N\rightarrow\infty$ is that $\epsilon(N)\nu(\ell,N)$ has a limit in $\nu'_\ell\in[0,\infty)$ for each $\ell$. Under this sufficient condition, $\epsilon(N)^2\nu(\ell,N)\rightarrow 0$ for all $\ell$.  This means that $\mathcal M_N[X]\to X$ and \eqref{aa} implies,
    \begin{align}
    \label{a}
    \lim_{N\rightarrow\infty} \Big\|& U_0(N)^*\rho_N^{(K)} U_0(N) -e^{-iT\,\E[\xi]\sum_{\ell=1}^{|\L|}\nu'_\ell B_\ell} \big(\rho^{\otimes K}\big) e^{iT\,\E[\xi]\sum_{\ell=1}^{|\L|}\nu'_\ell B_\ell} \Big\|_1=0.
    \end{align}
\end{itemize}
\end{discussion}

\begin{example}
\begin{itemize}
\item[{\bf 1.}] Suppose that 
$\lim_{N\rightarrow\infty} \nu(\ell,N) N^{-\alpha}=\nu_\ell\in(0,\infty)$, $\epsilon(N)=N^{-\beta}$ for some exponents $\alpha>0$, $\beta>0$. Then (recall \eqref{defP})
\begin{align}
\lim_{N\rightarrow\infty} &\mathcal M_N[X]\nonumber\\
&=\left\{
\begin{array}{cl}
\P[X]=\displaystyle\sum_{\lambda\in\Lambda} P_\lambda X P_\lambda & \beta<\alpha/2\\[4ex]
\displaystyle\sum_{\lambda,\mu\in\Lambda} \exp\Big[-\tfrac12 T^2 {\rm Var}[\xi]\,  \displaystyle \sum_{\ell=1}^{|\L|}  \nu_\ell\big(b_\ell(\lambda)-b_\ell(\mu)\big)^2 \Big] P_\lambda XP_\mu \qquad  & \beta=\alpha/2\\[4ex]
X & \beta>\alpha/2
\end{array}
\right.
\end{align}
The case $\beta<\alpha/2$ can be interpreted as the regime of strong interaction. We obtain the same decoherence result as before (no scaling; $\beta=0$) (see Theorem \ref{thm2}). For $\beta>\alpha/2$ there is no decoherence effect (the coupling is too weak). The intermediate, balanced scaling regime $\beta=\alpha/2$ gives an explicit exponential suppression factor for each block $P_\lambda \cdot P_\mu$. The behaviour of the phases \eqref{3.24-1} for $\mathbb E[\xi]\neq 0$ is as follows. For $\beta<\alpha$ the phases continue oscillating (generically, see above), while
\begin{align}
\lim_{N\rightarrow\infty}U_1(N) 
&=\left\{
\begin{array}{cl}
\exp\Big[-iT\E[\xi]\displaystyle\sum_{\ell=1}^{|\L|}\nu_\ell B_\ell\Big] \qquad  & \beta=\alpha\\[4ex]
1 & \beta>\alpha
\end{array}
\right.
\end{align}

\item[\bf 2.] {\it Mean-field scaling and Hartree equation.}\  Take the scaling $\epsilon(N)=\frac{1}{\nu(1, N)}$ and suppose that $B_1=\sum_{k=1}^K A_k\in \L$ has incidence number $\nu(1,N)\propto N$, while all other incidence numbers satisfy $\lim_{N\rightarrow\infty} \nu(\ell,N)/N=0$ ($\ell\neq 1$). This means that asymptotically, almost all reservoir vertices connected to $\s$ are connected to every vertex of $\s$, with at most a subleading number of different connections ($B_\ell$, $\ell\neq 1$). With this choice of scaling, $\eps(N)\nu(1,N)=1$ for all sufficiently large $N$. Then 
\begin{align}
\label{sss-}
\lim_{N\rightarrow\infty} \Big\|\rho_N^{(K)} -  U_0(N) \big(\sigma^{\otimes K}\big) U_0(N)^*\Big\|_1 =0,\quad \sigma = \sigma(T)=e^{-iT\E[\xi]A} \rho\, e^{iT\E[\xi]A}.
\end{align}
This follows from Theorem \ref{thm3.1} since $\lim_{N\rightarrow\infty}\mathcal M_N[X]=X$  and $\lim_{N\rightarrow\infty}U_1(N)= \big(e^{-iT\E[\xi]A}\big)^{\otimes K}$ (see \eqref{3.241} and \eqref{3.24-1}). Furthermore, $\sigma(T)$ satisfies the non-linear Hartree equation
\begin{align}
\label{Hartree}
i\partial_T\sigma(T) = {\rm Tr}_2\big[A\otimes A, \sigma(T)\otimes \sigma(T)\big],\qquad \sigma(0)=\rho,
\end{align}
where the trace is taken over the second factor and $[\cdot,\cdot]$ is the commutator. 
This is a manifestation of the well-known phenomenon that mean-field limits generate non-linear evolution equations \cite{Spohn,MerBer, MerRaf}. \eqref{Hartree} is readily seen from $i\partial_T\sigma(T)= \E[\xi]\, [A,\sigma(T)]$ on the one hand and ${\rm Tr}_2[A\otimes A, \sigma(T)\otimes \sigma(T)\big] = [A,\sigma(T)]\, {\rm Tr}(\sigma(T)A)$ with ${\rm Tr}(\sigma(T)A)={\rm Tr}(\rho A)=\E[\xi]$ on the other hand.

\item[\bf 3.] {\it Fluctuation scaling and Gaussian de Finetti representation.}\ Take the scaling $\epsilon(N)\propto\frac{1}{\sqrt N}$, which is linked to fluctuation theory \cite{FaMe}. Suppose as in the previous point that $B_1=\sum_{k=1}^K A_k\in \L$ has incidence number $\nu(1,N)\propto N$ with $\lim_{N\rightarrow\infty} \nu(1,N)/N=1$, while all other incidence numbers satisfy $\lim_{N\rightarrow\infty} \nu(\ell,N)/N=0$ ($\ell\neq 1$). Then
\begin{align}
\label{sss--}
\lim_{N\rightarrow\infty} \Big\|\rho_N^{(K)} -  U_0(N) U_1(N)\frac{1}{\sqrt\pi}\Big(\int_{\mathbb R}\big(\sigma'(x)\big)^{\otimes K}e^{-x^2}dx\Big) U_1(N)^*U_0(N)^*\Big\|_1 =0,
\end{align}
where
\begin{align}
\sigma'(x) = e^{iT\sqrt{2{\rm Var}[\xi]} Ax} \rho\,  e^{-iT\sqrt{2{\rm Var}[\xi]}Ax}.
\end{align}
This follows from Theorem \ref{thm3.1} and  since $\lim_{N\rightarrow\infty}U_N(x_1,\ldots,x_{|\L|}) = e^{iT\sqrt{2{\rm Var}[\xi]}(\sum_{k=1}^K A_k) x_1}$ (see \eqref{10.01}) so that by Lemma \ref{lem:intrep},
\begin{align}
\lim_{N\rightarrow\infty}\mathcal M_N[\rho^{\otimes K}]=\frac{1}{\sqrt \pi} \int_{\mathbb R}\big(\sigma'(x)\big)^{\otimes K}e^{-x^2}dx.
\end{align}
In this parameter regime, $\epsilon(N)\nu(1,N)\propto\sqrt N\rightarrow\infty$. Consequently, if $\E[\xi]\neq0$, the unitary $U_1(N)$ \eqref{3.24-1} generally does not have a limit (except accidentally by periodicity of the exponential). 
If $\E[\xi]=0$, then $U_1(N)=1$, and the reduced system state $\rho_N^{(K)}$ converges, up to the system evolution $U_0(N)$, to the Gaussian de Finetti mixture 
\begin{align*}
    \frac{1}{\sqrt\pi} \int_{\mathbb R} \big(\sigma'(x)\big)^{\otimes K}e^{-x^2}\,dx.
\end{align*}
\end{itemize}
\end{example}

\bigskip
\medskip

\section{Canonical examples}
\label{sec:cano}
\medskip

\subsection{All-to-all connected systems}
\label{sec:all to all}

Consider the general setup \eqref{m1}-\eqref{m7} corresponding to a completely connected graph, 
$$
E(N)=\{ (i,j) :1 \le i < j \le N\}.
$$
The set of boundary operators \eqref{Lstar} is the singleton
$$
\L = \{B_1\},\qquad B_1=\sum_{\ell=1}^KA_\ell
$$
and the system is collectively resolved (Definition \ref{def:loccolres}). As the only incidence number \eqref{incidence number} is $\nu(1,N)=N-K$, we have $\ve_K(N) =N-K$ \eqref{connection number} and $\s$ is well-connected \eqref{wellconn}. The eigenvalues of $B_1$ are of the form 
\begin{equation}
\label{3.3}
b(n_1,\ldots,n_\sss)= n_1a_1+\cdots+n_\sss a_\sss \quad \text{where} \quad  n_j\in \mathbb N, \quad \sum_{s=1}^\sss n_s = K,   
\end{equation}
Let us assume that the eigenvalues $a_s$ of $A$ are rationally independent modulo constants. Then $b(n_1,\ldots,n_\sss)= b(n'_1,\ldots,n'_\sss)$ if and only if $n_1=n_1',\ldots,n_\sss=n_\sss'$. In this case the eigenvalues are  labelled uniquely by $n_1,\ldots,n_\sss$ and the eigenprojections are
\begin{equation}
P_{n_1,\ldots,n_\sss} = \sum_{\substack{(s_1,\ldots,s_K)\in\llbracket 1,\sss \rrbracket^K
\\ |\{j:s_j = r\}| = n_r \text{ for all } r}} \Pi_{s_1}\otimes \cdots\otimes \Pi_{s_K}.
\end{equation}
We obtain
\begin{equation}
\rank\, P_{n_1,\ldots,n_\sss} = {K\choose {n_1,\ldots,n_\sss}} \prod_{s=1}^\sss \big(\rank\,\Pi_s\big)^{n_s} = \frac{K!}{(n_1)!\cdots (n_\sss)!} \prod_{s=1}^\sss \big(\rank\,\Pi_s\big)^{n_s}.
\end{equation}
The cardinality of the index set $\Lambda$ is given by the number of solutions to the equation $n_1+\cdots +n_\sss=K$, for $n_s\in\mathbb N$ including zero, which by the `stars and bars' formula is,
\begin{equation}
|\Lambda| = {K+\sss-1 \choose K}.
\end{equation}
The set of Bohr energies \eqref{Bohr} is $\mathfrak B = \big\{(n_1-n'_1)a_1+\cdots +(n_\sss-n'_\sss)a_\sss \big\}$, where the $n_s$ and $n'_s$ satisfy the constraint of \eqref{3.3}. The genericity condition $q<1$ is equivalent to the following: for every $\Delta\in\mathfrak B\setminus\{0\}$ there exist $s,s'$ with $p_sp_{s'}>0$ such that
\[
(a_s-a_{s'})\Delta\notin\frac{2\pi}{T}\mathbb Z.
\]
This follows from \eqref{psi1}: $|\psi(\Delta)|=1$ precisely when all phases associated with nonzero spectral weights coincide.

According to Theorem \ref{thm:definetti}(3) we have
$\lim_{N\rightarrow\infty}  \rho_N^{(K)} = \frac{1}{(2\pi)^\sss}  \int_{[0,2\pi]^\sss}   \rho(\btheta)^{\otimes K} \, d\btheta$.

\subsection{Qubits}
\label{sec:qubits}

For a network of qubits (spins $1/2$) the single site Hilbert space is $\C^2$, $d=2$ in \eqref{m1} and we take the interaction operator to be
\begin{equation}
\label{A}
A =\begin{pmatrix}
\omega & 0\\
0 & 0
\end{pmatrix},\qquad \omega>0,\qquad a_1=\omega,\ a_2=0, \quad \Pi_1=|1\rangle\langle1|, \ \Pi_2= |2\rangle\langle 2|.
\end{equation}
Any self-adjoint operator $A\in\b(\C^2)$ can be diagonalized by a unitary transformation. We consider the particular choice \eqref{A}, noting that additive scalar shifts of $A$ generally change the pair-interaction Hamiltonian. The eigenvalues are rationally independent modulo constants for all $\omega>0$ ({\it cf.}~after \eqref{Utheta}).

Let us take for $\s$ a pair of qubits, $K=2$. Up to permutation of the qubits, there are five possible sets of boundary operators $\L$, \eqref{Lstar}.

\begin{table}[ht] 
\centering
\begin{tabular}{|l |c|c |c |c|}
\hline 
\raisebox{2.5ex}{}$\L$ & $|\L|$ & $|\Lambda|$ & $P_\lambda$ & $\mathfrak B$ \\ [0.5ex] 
\hline 
\hline
\raisebox{2.5ex}{}$\L_1 = \{A_1\}$ & $1$ & $2$ & $P_1= \Pi_1\otimes\mathbf 1$, $P_2=\Pi_2\otimes\mathbf 1$ & $\{-\omega,0,\omega\}$ \\ [0.5ex] 
\hline
 &  &  & \raisebox{-.5ex}{$P_1=\Pi_1\otimes \Pi_1$, $P_2=\Pi_1\otimes \Pi_2$} & \\
\raisebox{2.5ex}{$\L_2=\{A_1,A_2\}$} & \raisebox{2.5ex}{$2$} &\raisebox{2.5ex}{$4$} & \raisebox{1ex}{$P_3=\Pi_2\otimes \Pi_1$, $P_4=\Pi_2\otimes \Pi_2$} & \raisebox{2.5ex}{same as $\L_1$}\\
\hline
&  &  & \raisebox{-.5ex}{$P_1=\Pi_1\otimes \Pi_1$, $P_2=\Pi_2\otimes \Pi_2$} & \\
\raisebox{2.5ex}{$\L_3=\{A_1+A_2\}$} & \raisebox{2.5ex}{$1$} &\raisebox{2.5ex}{$3$} & \raisebox{1ex}{$P_3=\Pi_1\otimes \Pi_2 + \Pi_2\otimes \Pi_1$} &\raisebox{2.5ex}{$\{-2\omega,-\omega,0,\omega,2\omega\}$} \\
\hline
\raisebox{2.5ex}{}$\L_4=\L_1\cup\L_3$ & $2$ & $4$ & \text{same as $\L_2$} & same as $\L_3$ \\ [0.5ex] 
\hline
\raisebox{2.5ex}{}$\L_5=\L_2\cup\L_3$ & $3$ & 4 & \text{same as $\L_2$} & same as $\L_3$ \\ [1ex] % [1ex] adds vertical space
\hline %inserts single line
\end{tabular} \\ [1ex]
\caption{{\it Key quantities where each vertex hosts a qubit and where $\s$ consists of a pair of qubits.} The set of boundary operators $\L$ \eqref{Lstar}, its cardinality $|\L|$, the cardinality $|\Lambda|$ and the elements of the set of minimal projections $P_\lambda$ \eqref{Beval} and the set of Bohr energies $\mathfrak B$ \eqref{Bohr}. The families $\L_2, \L_4,\L_5$ are locally resolved, $\L_2$ is coarsely resolved (\eqref{2.16}, $K_1=1$) and $\L_3$ is collectively resolved.}
\label{table:guiding example 1} % is used to refer this table in the text
\end{table}

To find the effective connection number $\ve_2(N)$, \eqref{connection number}, we need to evaluate
\begin{equation}
\label{need}
S_{\lambda,\mu}=\sum_{\ell:\,b_\ell(\lambda)\neq b_\ell(\mu)}\nu(\ell,N),\qquad \lambda\neq \mu.
\end{equation}
By symmetry we only have to consider $\lambda <\mu$. If $\L$ consists of a single element $B_1$ then we have $b_1(\lambda)\neq b_1(\mu)$ for all $\lambda\neq \mu$ by the separating property \eqref{sep}. This implies that $\ve_2(N)=\nu(1,N)$ for any singleton $\L$, namely $\L_1,\L_3$ in Table \ref{table:guiding example 1}.

For $\L_2$ we have $B_{\ell =1}=A_1$ and $B_{\ell =2}=A_2$. The eigenvalues of the elements in $\L_2$ are obtained by applying the projections $P_\lambda$, as given in Table \ref{table:guiding example 1}. We have $b_1(1)=b_1(2)=\omega$, $b_1(3)=b_1(4)=0$ and  $b_2(1)=b_2(3)=\omega$, $b_2(2)=b_2(4)=0$. We then obtain from \eqref{need},
$S_{1,2}=S_{3,4}=\nu(2,N)$, $S_{1,3}=S_{2,4}=\nu(1,N)$, $S_{1,4}=S_{2,3}=\nu(1,N)+\nu(2,N)$. It follows that $\ve_2(N)=\min\{\nu(1,N),\nu(2,N)\}$.

For $\L_4$ we have $B_1=A_1$, $B_2=A_1+A_2$. Then $S_{1,2}=S_{3,4}=\nu(2,N)$, $S_{1,3}=S_{1,4}=S_{2,4}=\nu(1,N)+\nu(2,N)$, $S_{2,3}=\nu(1,N)$ and $\ve_2(N)=\min\{\nu(1,N),\nu(2,N)\}$. 

For $\L_5$ with $B_1=A_1$, $B_2=A_2$, $B_3=A_1+A_2$ we obtain $S_{1,2}=S_{3,4}=\nu(2,N)+\nu(3,N)$, $S_{1,3}=S_{2,4}=\nu(1,N)+\nu(3,N)$, $S_{1,4}=\nu(1,N)+\nu(2,N)+\nu(3,N)$, $S_{2,3}=\nu(1,N)+\nu(2,N)$, resulting in $\ve_2(N)=\min_{1\le\ell\neq\ell'\le 3}\big(\nu(\ell,N)+\nu(\ell',N)\big)$. 

We sum up these values of $\ve_2(N)$ in Table \ref{table:guiding example 2}.
\medskip

To obtain the value of $q$, \eqref{def q} we use that $a_1=\omega$, $a_2=0$ and setting $p\equiv p_1$, $p_2=1-p$, we evaluate \eqref{charfunct},
$$
\psi(\Delta) = p e^{-iT\omega \Delta} +1-p,\qquad |\psi(\Delta)|^2 = 1-4p(1-p)\sin^2\left(\frac{\omega T\Delta}{2}\right).
$$

\begin{table}[ht] 
\centering
\begin{tabular}{|l |c|c |}
\hline 
\raisebox{2.5ex}{}$\L$ & $\ve_2(N)$ &  $q^2$  \\ [0.5ex] 
\hline 
\hline
\raisebox{2.5ex}{}$\L_1 = \{A_1\}$ & $\nu(1,N)$ & $ 1-4p(1-p)\sin^2(\omega^2 T/2)$\\ [.5ex] 
\hline
\raisebox{2.5ex}{}$\L_2=\{A_1,A_2\}$ & $\min\{\nu(1,N),\nu(2,N)\}$ & same as $\L_1$ \\[.5ex]
\hline
\raisebox{2.5ex}{}$\L_3=\{A_1+A_2\}$ & $\nu(1,N)$ & $ 1-4p(1-p)\min_{\alpha=-1, 0}\sin^2(2^\alpha\omega^2 T)$  \\[.5ex]
\hline
\raisebox{2.5ex}{}$\L_4=\L_1\cup\L_3$ & $\min\{\nu(1,N),\nu(2,N)\}$ & same as $\L_3$\\ [.5ex] 
\hline
\raisebox{2.5ex}{}$\L_5=\L_2\cup\L_3$ & $\min_{1\le\ell\neq\ell'\le 3}\big(\nu(\ell,N)+\nu(\ell',N)\big)$ & same as $\L_3$ \\ [1ex] % [1ex] adds vertical space
\hline %inserts single line
\end{tabular} \\ [1ex]
\caption{{\it Key quantities where each vertex hosts a qubit and where $\s$ consists of a pair of qubits.} The effective connection number $\ve_2(N)$, \eqref{connection number} and the square of the parameter $q$, \eqref{def q}, for the five possible $\L$.
} 
\label{table:guiding example 2} % is used to refer this table in the text
\end{table}

We are in the generic regime \eqref{gencond} for $\L_1,\L_2$ provided that $\rho$ is not $|1\rangle\langle 1|$ nor $|2\rangle\langle 2|$ (not an eigenstate of $A$; $p\neq 0,1$) and provided that $\omega^2 T\not\in2\pi\mathbb Z$. In the cases $\L_3,\L_4,\L_5$ we are in the generic regime if $p\neq 0,1$ and $\omega^2 T\not\in\pi\mathbb Z$. 
In the case $\L_{1,2,3,4}$, the system is well-connected if and only if the incidence numbers $\lim_{N\rightarrow\infty}\nu(\ell,N)=\infty$ for all $\ell$. For $\L_5$ the system is well-connected if and only if at least two of the three $\nu(\ell,N)$, $\ell=1,2,3$, tend to $\infty$ for $N\rightarrow\infty$. 

In the case $\L_1$ the boundary operators are neither locally nor collectively resolved. By the definition \eqref{defP} we have $\P[\rho\otimes\rho] = \Pi_1\rho\Pi_1\otimes\rho + \Pi_2\rho\Pi_2\otimes\rho = \big(p_1 |1\rangle\langle 1| +p_2|2\rangle\langle 2|\big)\otimes\rho=\tau\otimes\rho$, 
where $\tau$ is the state $\rho$ in which the coherences (in the basis $|1\rangle, |2\rangle$) are eliminated (see \eqref{DeFinettifinalstate}). In the generic regime and if $\s$ is well-connected, Theorems \ref{thm2}, \ref{thm3} and \ref{thm:definetti} give, 
\begin{align}
\L_1&: \qquad \lim_{N\rightarrow\infty}\rho_N^{(2)}=e^{-iTH_\s} (\tau\otimes \rho) e^{iTH_\s} \label{rh1}\\
\L_{2,4,5}&: \qquad \lim_{N\rightarrow\infty}\rho_N^{(2)}=\tau\otimes\tau\label{lh2}\\
\L_3&:\qquad  \lim_{N\rightarrow\infty}\rho_N^{(2)} = 
 \frac1{(2\pi)^2}\int_{[0,2\pi]^2}   \rho(\btheta)^{\otimes 2} \, d\btheta.
\label{lh3}
\end{align}
Denoting the density matrix elements of $\rho$ in the basis $\{|1\rangle, |2\rangle\}$ by $\rho_{\pm}$ (see \eqref{A}), we obtain from \eqref{Utheta},
\begin{align}
\rho(\btheta) = \begin{pmatrix}
   \rho_{++} & e^{-i(\theta_1-\theta_2)}\rho_{+-}\\[1.5ex]
   e^{i(\theta_1-\theta_2)}\rho_{-+} & \rho_{--}
\end{pmatrix}.
\end{align}

\subsection{Erd\H os-R\'enyi random graphs}
\label{sec:Erdos Renyi}

In this model a graph is constructed by connecting the vertices randomly. Each edge is included in the graph with probability $p\in[0,1]$ and all edges are placed independently. In the setup of Section \ref{sec:model} we consider a family of increasing graphs (networks) labelled by $N$. We thus consider the Erd\H os-R\'enyi graph on the vertex set $\mathbb N$. For each realization of the infinite graph and each $N\in\mathbb N$ we obtain a subgraph on the vertex set $\{1,\ldots,N\}\subset \mathbb N$ by deleting all vertices $>N$ and all edges in the infinite graph which link to any vertex $> N$. Then we are in the setup of Section \ref{sec:model}. Each realization of the infinite graph yields a family of increasing graphs indexed by $N$ and we obtain an increasing family of random graphs by varying over all realizations. For $p=0$ the graph is entirely disconnected which is a trivial (non-interacting) model. For $p=1$  the graph is totally connected (and not random), the situation we analyzed in Section \ref{sec:all to all}. We assume from this point onward that
\begin{equation}
\label{er1}
0<p<1.
\end{equation}
The set of unordered pairs of distinct vertices \[
\{(k,\ell):k,\ell\in\mathbb N,\ k<\ell\}
\] is countable; enumerate its elements as $\{e_j\}_{j\ge1}$. For each $j$, let $X_j$ be the indicator that the edge $e_j$ is present. The variables $X_j$ are independent Bernoulli variables with $\Pr(X_j=1)=p$. The sample space $\Omega=\{0,1\}^{\mathbb N}$ carries a $\sigma$-algebra generated by the cylinder sets of the form $\{\omega \in\Omega : \omega_{j_1}=v_1,\ldots,\omega_{j_n}=v_n\}$, $n\in\mathbb N$, $1\le j_1<\cdots<j_n$, $v_1,\ldots,v_n\in \{0,1\}$. The probability measure $\mu$ on this $\sigma$-algebra is uniquely determined (Kolmogorov extension) by
$$
\mu\{\omega \in\Omega : \omega_{j_1}=v_1,\ldots,\omega_{j_n}=v_n\}=\Pr\{\omega_{j_1}=v_1,\ldots,\omega_{j_n}=v_n\} = \prod_{k=1}^n p^{v_k}(1-p)^{1-v_k}.
$$
The boundary operators $L(j)$, \eqref{eq:bndryop} and their collection $\L_N$, \eqref{LN} now become random variables. So do the incidence numbers $\nu(\ell,N)$, the effective connection number $\ve_K(N)$ and the reduced system density matrix $\rho^{(K)}_N$. Let $\varnothing \not =J\subseteq\llbracket 1,K\rrbracket$ and denote 
$$
B(J)=\sum_{i\in J}A_i.
$$
The maximal set of boundary operators is defined as
\begin{equation}
\label{maxL}
\L_{\rm max} =  \big\{B(J) : \varnothing \ne J\subseteq \llbracket 1,K\rrbracket\big\}. 
\end{equation}
For the following result we define
\begin{align}
\label{ttmin}
t_J = p^{|J|}(1-p)^{K-|J|} 
\text{\quad and\quad }
t_{\rm min} = \min_{\varnothing\not = J\subset \llbracket1,K\rrbracket} t_J = p \big(\tfrac12-|p-\tfrac12|\big)^{K-1}.
\end{align}

\begin{theorem}
\label{thm:erdoes-renyi}  \ \\
\vspace*{-.4cm}
\begin{itemize}
\setlength\itemsep{1em}
\item[(1)] There exists a positive constant $\beta_K(p)$ independent of $N$, such that for every $N > K$, 
\[
    \Pr\left\{\mathcal{L}_N \ne \mathcal{L}_{\rm max}\right\} \le 2^K e^{-(N - K)\beta_K}.
\] In particular, 
\begin{equation}
\label{er3.3}
\Pr\big\{\L_N=\L_{\rm max}\text{\ \rm for some $N\in\mathbb N$}\big\} =\Pr\big\{\exists N_0 \text{\ \rm s.t.\ }\L_N=\L_{\rm max} \ \forall N\ge N_0 \big\}=1.
\end{equation}

\item[(2)] Fix $\varnothing \ne J\subseteq \llbracket1,K\rrbracket$, let $\nu(J,N)=\big|\{ j=K+1,\ldots,N : L(j) = B(J)\}\big|$ be the incidence number of $B(J)$ (see \eqref{incidence number}). Then we have for all $N>K$ and all $0<\delta<t_J$,
\begin{align}
\label{chb}
 \Pr\Big\{\Big|\frac{\nu(J,N)}{N-K}-t_J\Big|\ge \delta\Big\}\le 2 e^{-\frac{\delta^2 (N-K)}{3 t_J}}.
\end{align}
It follows that 
\begin{equation}
\label{er5}
\Pr\Big\{ \lim_{N\rightarrow\infty} \nu(J,N) = \infty \Big\}=1 \text{\quad and\quad } \Pr\Big\{ \lim_{N\rightarrow\infty} \ve_K(N) = \infty \Big\}=1.
\end{equation}

\item[(3)] Suppose we are in the generic regime, that is, \eqref{gencond} holds with $\mathfrak B=\mathfrak B_{\rm max}$ the Bohr energies \eqref{Bohr} of $\L_{\rm max}$ \eqref{maxL}. Then, for every $N > K$, 
\begin{align}\label{er:finite-bound}
    \Pr\Big\{\|\rho_N^{(K)} - \tau^{\otimes K}\|_1 > \zeta^K q^{t_{\rm min}(N - K)/2}\Big\} \le 2^K e^{-t_{\rm min}(N - K)/8}.
\end{align}
In other words, the upper bound on the trace norm holds eventually in $N$, almost surely. This implies in particular that 
\begin{equation}
\label{3.13}
\Pr\Big\{ \lim_{N\rightarrow\infty}\big\|\rho^{(K)}_N - \tau^{\otimes K}\big\|_1=0\Big\}=1.
\end{equation}
\end{itemize}
\end{theorem}

\begin{discussion}
The relation \eqref{er3.3} means that $\s$ is locally resolved eventually (in $N$), almost surely. The bound \eqref{chb} shows that $\nu(J,N)$ almost surely grows linearly in $N$; note that $t_{\rm min}>0$ as $p\not= 0,1$. The second equality in \eqref{er5} shows that the system is well-connected almost surely.  
\end{discussion}

{\bf Proof of Theorem \ref{thm:erdoes-renyi}.} (1) Fix a vertex $j\ge K+1$ and a set $J=\{i_1,\ldots,i_{|J|}\}$ of vertices in $\s$. The probability that $j$ is linked to the vertices in $J$ and to no other vertices in $\s$, is $p^{|J|}(1-p)^{K-|J|}$. Therefore the probability that the boundary operator $L(j)$, \eqref{eq:bndryop} equals $B(J)$, is
\begin{equation}
\label{er2}
\Pr\big\{L(j)=B(J)\big\} = p^{|J|}(1-p)^{K-|J|}.  
\end{equation}
By symmetry of the edge probabilities and as existing edges are not altered as the graph grows, this probability depends only on the cardinality of the set $J$, and it is independent of $N$. Let $J\subseteq\llbracket 1,K\rrbracket$ and denote $B(J)=\sum_{i\in J}A_i$. The probability that $B(J)$ belongs to $\L_N$ is
\begin{align}
\Pr\big\{B(J)\in\L_N\big\} &= \Pr\Big\{L(j)=B(J)\ \text{for some $j\in\llbracket K+1,N\rrbracket$}\Big\}\nonumber\\
&=1- \Pr\Big\{L(j)\neq B(J)\ \text{for all $j\in\llbracket K+1,N\rrbracket$}\Big\}\nonumber\\
&=1-(1-t_J)^{N-K},\qquad t_J \equiv p^{|J|}(1-p)^{K-|J|}.
\label{5.19}
\end{align}
We used \eqref{er2} in the last step, and that the $L(j)$ are independent for distinct $j$.  We have $\{\L_N\neq \L_{\rm max}\}= \cup_{\varnothing \ne J\subseteq\llbracket1,K\rrbracket}\{B(J)\not\in\L_N\}$ and 
\begin{align}
\Pr\big\{\L_N\neq \L_{\rm max} \big\} &\le \sum_{\substack{J\subseteq \llbracket 1,K\rrbracket\\ J \ne \varnothing}} \Pr\big\{B(J)\not\in\L_N\big\}=\sum_{\substack{J\subseteq \llbracket 1,K\rrbracket\\ J \ne \varnothing}} (1-t_{J})^{N-K}\nonumber\\
& =\sum_{|J|=1}^K {K\choose |J|} (1-t_J)^{N-K} = \sum_{|J|=1}^K {K\choose |J|} \big(1-p^{|J|}(1-p)^{K-|J|}\big)^{N-K}.
\label{5.23}
\end{align}
As $\min_{1\le j\le K} p^j(1-p)^{K-j} =p \big(\tfrac12-|p-\tfrac12|\big)^{K-1}$
we obtain from \eqref{5.23},
\begin{equation}
\label{3.8}
\Pr\big\{\L_N\neq \L_{\rm max} \big\}\le 2^K e^{-(N-K)\beta_K(p)},
\end{equation}
where 
\begin{equation}
\beta_K(p)=-\ln\Big[1-p \big(\tfrac12-|p-\tfrac12|\big)^{K-1}\Big]>0\qquad \text{($0<p<1$)}.
\end{equation}
This shows that $\Pr\{\L_N\neq \L_{\rm max} \}$ decays exponentially quickly in $N$. For small values of $N$ this probability may not be small: take for instance $K=1, N=2$, then  $\Pr\{\L_N\neq\L_{\rm max}\}=1-p$. 

Now we show \eqref{er3.3}. We have
$$
\Pr\big\{\L_N=\L_{\rm max}\text{\ for some $N\in\mathbb N$}\big\} = 1-\Pr\big\{\L_N\neq \L_{\rm max}\  \forall N\in\mathbb N \big\}
$$
and $\Pr\big\{\L_N\neq \L_{\rm max}\  \forall N\in\mathbb N \big\}=\Pr\big\{\cap_{N\in\mathbb N} \{\L_N\neq \L_{\rm max}\}\big\}\le \Pr\{\L_N\neq\L_{\rm  max}\}$, for any $N\in\mathbb N$. Thus by \eqref{3.8}, $\Pr\big\{\L_N\neq \L_{\rm max}\  \forall N\in\mathbb N \big\}=0$ and so $\Pr\big\{\L_N=\L_{\rm max}\text{\ for some $N\in\mathbb N$}\big\} = 1$. Hence with probability one $\exists N_*\in\mathbb N$ s.t. $\L_{N_*}=\L_{\rm max}$, that is, such that $\L_N=\L_{\rm max}$ for all $N\ge N_*$  (as $\L_N\subseteq\L_{N+1}$). In other words, $\L_N=\L_{\rm max}$ eventually with probability one. 
\medskip

(2)  For fixed $\varnothing\neq J\subseteq\llbracket1,K\rrbracket$ and
$j\ge K+1$, let $\chi_j(J)$ be the indicator of the event $\{L(j)=B(J)\}$. The $\{\chi_j(J)\}_{j=K+1}^N$ are independent Bernoulli random variables taking the value $1$ with probability $t_J$ and the value $0$ with probability $1-t_J$ (see \eqref{er2}, \eqref{5.19}). The incidence number is given by 
\begin{equation}
\label{er3}
\nu(J,N)=\sum_{j=K+1}^N\chi_j(J)
\end{equation}
and has the binomial distribution $\nu(J,N)\sim{\rm Bin}(N-K,t_J)$, $\E[\nu(J,N)]=t_J(N-K)$. The Chernoff bound (see for instance \cite{Tsun, Goemans}) gives \eqref{chb}. Next, 
Chernoff's lower tail bound for binomial random variables gives $\forall\  0<\delta<1$,
\begin{align}
\Pr\big\{\nu(J,N) \le  (1-\delta)t_J(N-K)\big\}\le e^{-\frac{\delta^2}{2}t_J(N-K)},\quad \forall N>K.
\label{xxy}
\end{align}
For every fixed realization, $\nu(J,N)$  is an increasing function of $N$, so that $\nu(J,N)$ converges to $\infty$ if and only if it is unbounded in $N$. Now 
\begin{align}
\Pr\big\{\nu(J,N) \text{\ is unbounded in $N$}\big\} 
&=  1-\Pr\Big\{\cup_{M\in\mathbb N}\{ \nu(J,N)< M \ \forall N>K\}\Big\}\nonumber\\
&=1-\lim_{M\rightarrow\infty}\Pr\big\{ \nu(J,N)< M \ \forall N>K \big\}\label{5.28}\\
&\ge 1-\lim_{M\rightarrow\infty}\inf_{N> K} \Pr\big\{ \nu(J,N)< M \big\}\nonumber%\label{5.29},
\end{align}
where we used the continuity of the probability in \eqref{5.28}. Given $M\in\mathbb N$ (and any $\delta$, $t_J$ and $K$ in their ranges) we have $\Pr\big\{ \nu(J,N)< M \big\}\le \Pr\big\{ \nu(J,N)< (1-\delta)t_J(N-K) \big\}$ provided $N$ is large enough  and so by \eqref{xxy}, $\inf_{N> K} \Pr\big\{ \nu(J,N)< M \big\}=0$. Consequently, $\Pr\big\{\nu(J,N) \text{\ is unbounded in $N$}\big\} =1$, which shows the first equality in \eqref{er5}. 

Next we show the second equality in \eqref{er5}. By \eqref{connection number}, $\ve_K(N)\ge \min_{\substack{J\subset \llbracket 1,K\rrbracket\\ J \ne \varnothing}}\nu(J,N)$. For fixed $K$ there are finitely many nonempty subsets $J\subseteq\llbracket1,K\rrbracket$ and for each one, there is a set $\Omega_J$ of full measure, on which $\nu(J,N)$ converges to $\infty$. Therefore $\bigcap_{\varnothing \ne J\in\llbracket 1,K\rrbracket}\Omega_J$ has full measure and $\ve_K(N)$ converges to $\infty$ on this set. 
\medskip

(3) We first prove \eqref{er:finite-bound}.
For each $N>K$, define the event
\[
\mathcal G_N
=
\bigcap_{\varnothing\neq J\subseteq\llbracket1,K\rrbracket}
\left\{
\nu(J,N)\ge\frac12t_{\rm min}(N-K)
\right\}.
\]
Since $t_J\ge t_{\rm min}$, the Chernoff bound
\eqref{xxy}, with $\delta=1/2$, gives
\begin{align}
\Pr(\mathcal G_N^c)
&\le
\sum_{\varnothing\neq J\subseteq\llbracket1,K\rrbracket}
\Pr\left\{
\nu(J,N)<\frac12t_J(N-K)
\right\}
\nonumber\\
&\le
\sum_{\varnothing\neq J\subseteq\llbracket1,K\rrbracket}
e^{-t_J(N-K)/8}
\nonumber\\
&\le
2^Ke^{-t_{\rm min}(N-K)/8}.
\label{er:good-event}
\end{align}

On $\mathcal G_N$, every nonempty boundary type
has positive incidence number, so that
$\L_N=\L_{\rm max}$.
In particular, $\L_N$ is locally resolved,
and \eqref{connection number} gives
\[
\ve_K(N)\ge\frac12t_{\rm min}(N-K).
\]
Theorem \ref{thm2} and Lemma \ref{lem:P}
therefore imply that, on $\mathcal G_N$,
\[
\|\rho_N^{(K)}-\tau^{\otimes K}\|_1
\le
\zeta^Kq^{\ve_K(N)}
\le
\zeta^Kq^{\frac12t_{\rm min}(N-K)}.
\]
Consequently, the event on the left-hand
side of \eqref{er:finite-bound} is contained
in $\mathcal G_N^c$, which proves
\eqref{er:finite-bound}.

Finally, the right-hand side of
\eqref{er:good-event} is summable in $N$.
By the Borel--Cantelli lemma,
$\mathcal G_N$ occurs for all sufficiently
large $N$, almost surely. Since $q<1$ and
$t_{\rm min}>0$, equation \eqref{3.13}
follows immediately.
\hfill\qed

\bigskip
\medskip

\section{Area law for entanglement entropy}
\label{sec:convinfo}
\medskip

For many ground states of quantum lattice and quantum
many-body systems, the entanglement entropy of a
subregion grows like the area of its boundary, rather
than its volume. This `area law' is observed in seemingly unrelated fields in quantum information science, quantum many-body physics and  black hole physics \cite{ECP,WVHC}. Here we show an area law
for quantum networks: the increase in the von Neumann
entropy of the system $\s$ due to its interaction
with the reservoir is bounded proportionally to the
size $|\partial\s|$ of its boundary.

The {\it boundary} $\partial\s$ of $\s$ is defined as the collection of vertices in $\s$ which are connected to any vertex in $\r$, \eqref{deltaS}.  As discussed in Section \ref{sec:stab}, for a fixed value of $K$ and sufficiently large $N$, it is independent of $N$. Our goal here is to estimate the von Neumann entropy of the state $\rho_N^{(K)}$ in terms of the size (cardinality) of the boundary, $|\partial\s|$. The von Neumann entropy of a density matrix $\rho$ (of any dimension) is defined by
\begin{equation}
\label{vne}
\S(\rho)=-\Tr(\rho\ln\rho)\ge 0.
\end{equation}

\begin{theorem}[Area law]
\label{thm:arealaw}
\ \ 
\begin{itemize}
\item[(1)]
For all single-vertex density matrices $\rho$ (see \eqref{m2}),
\begin{equation}
0\le\S\big(\rho_N^{(K)}\big) -\S(\rho^{\otimes K}) \le   |\partial \s|\big( \S(\q[\rho]) -\S(\rho)\big),
\label{2.27}
\end{equation}
where $\q$ is the CPTP map acting on single-vertex operators, given by $\q[\cdot] = \sum_{s=1}^\sss \Pi_s[\cdot]\Pi_s$.  In terms of $\sss$, the cardinality of the spectrum of $A$, \eqref{m4}, we have the bound
\begin{equation}
\label{res}
0\le \S\big(\rho_N^{(K)}\big) -\S\big(\rho^{\otimes K}\big)\le   |\partial \s|\ln(\sss).
\end{equation}

\item[(2)]
The bound \eqref{2.27} is optimal for $N\rightarrow\infty$: if $m$ of the $K$ vertices in $\s$ are connected, each to a number $\nu_i$, $i=1,\ldots,m$ of vertices $j\ge K+1$ with no vertex $j$ connecting to more than one $i\in\llbracket 1,K\rrbracket$, and if furthermore the system is well-connected and the genericity condition \eqref{gencond} holds, then we have for any $\rho$,
\begin{equation}
\lim_{N\rightarrow\infty}\S\big(\rho_N^{(K)}\big)  -\S(\rho^{\otimes K}) = |\partial \s|\big( \S(\q[\rho])-\S(\rho)\big) 
\end{equation}
\end{itemize}
\end{theorem}

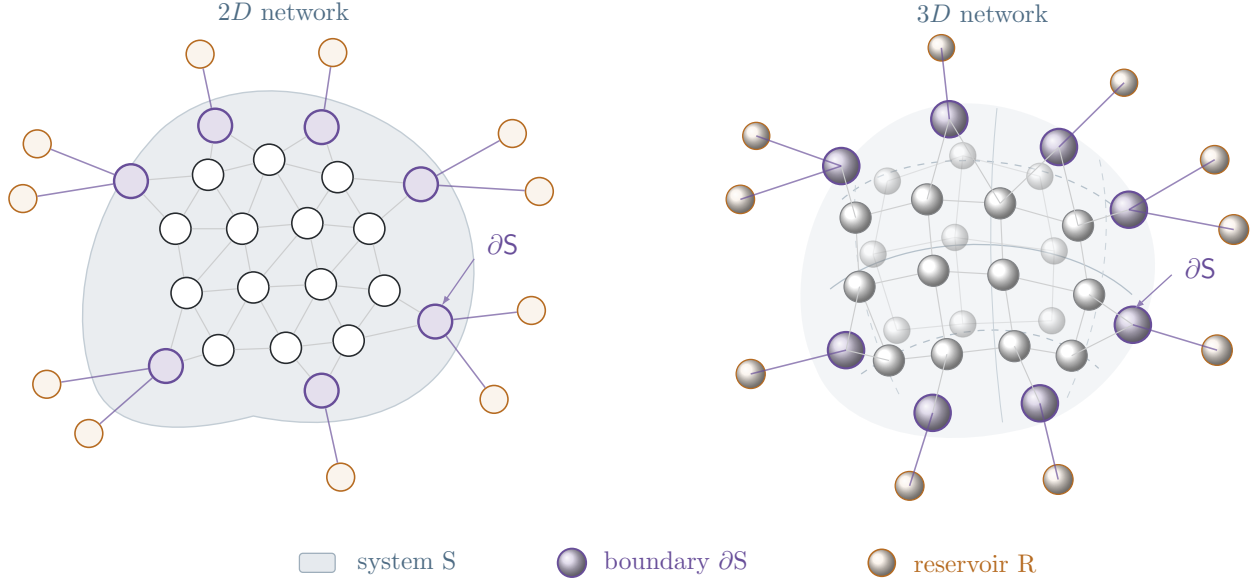
\begin{figure}[t]
    \centering
    \resizebox{\textwidth}{!}{%
        \begin{tikzpicture}[x=1.2cm,y=1.2cm,
         boundary/.style={circle,draw=Collective,line width=1.05pt,fill=Collective!18,minimum size=5.1mm,inner sep=0pt},
         reserv/.style={circle,draw=Coarse!92!black,line width=.68pt,fill=Coarse!8,minimum size=4.2mm,inner sep=0pt},
         biglabel/.style={font=\bfseries\small,fill=white,fill opacity=.88,text opacity=1,inner sep=1.6pt,rounded corners=1pt},
         intedge/.style={draw=Ink!23,line width=.52pt},
         crossedge/.style={draw=Collective,line width=.68pt,opacity=.68}]
        %\node[paneltitle] at (0,4.85) {Area law on quantum networks};
        %\node[panelsub,text=Ink!60] at (0,4.40) {the geometry of the ``area'' is induced by the graph};
        % explicit color key -- deliberately spaced so labels cannot collide
        \begin{scope}[shift={(0,3.92)}]
          \draw[fill=BlueGrey!14,draw=BlueGrey!55,line width=.55pt,rounded corners=1.5pt] (-4.15,-7.5-.12) rectangle (-3.70,-7.5+.12);
          \node[font=\normalsize,text=BlueGrey,anchor=west] at (-3.53,-7.5) {system $\s$};
        
          \shade[ball color=Collective!28,draw=Collective,line width=.8pt] (-.72,-7.5) circle (2.15mm);
          \node[font=\normalsize,text=Collective,anchor=west] at (-.43,-7.5) {boundary $\partial\s$};
        
          \shade[ball color=Coarse!12,draw=Coarse!92!black,line width=.65pt] (3.18,-7.5) circle (2.05mm);
          \node[font=\normalsize,text=Coarse!92!black,anchor=west] at (3.47,-7.5) {reservoir $\r$};
        \end{scope}
        
        \begin{scope}[shift={(-4.35,.15)}]
        \node[font=\normalsize,text=BlueGrey] at (0,3.22) {$2D$ network};
        \begin{scope}[on background layer]
        \path[fill=BlueGrey,fill opacity=.12,draw=BlueGrey,draw opacity=.34,line width=.6pt]
        (-2.35,-1.55)..controls(-2.75,-.65)and(-2.48,.78)..(-1.55,1.72)
        ..controls(-.82,2.38)and(.30,2.34)..(1.22,1.85)
        ..controls(2.20,1.34)and(2.62,.28)..(2.28,-.78)
        ..controls(1.96,-1.73)and(.92,-2.15)..(-.38,-1.88)
        ..controls(-1.28,-2.12)and(-2.08,-2.08)..(-2.35,-1.55)--cycle;
        \end{scope}
        \foreach \n/\x/\y in {a1/-.95/1.16,a2/-.18/1.35,a3/.68/1.13,a4/-1.36/.48,a5/-.52/.48,a6/.30/.55,a7/1.08/.48,a8/-1.22/-.33,a9/-.38/-.26,a10/.47/-.22,a11/1.27/-.30,a12/-.82/-1.05,a13/.03/-1.02,a14/.82/-.93}{\node[vertex] (\n) at (\x,\y) {};}
        \foreach \n/\x/\y in {b1/-1.92/1.08,b2/-.86/1.78,b3/.48/1.76,b4/1.73/1.04,b5/1.91/-.69,b6/.48/-1.56,b7/-1.48/-1.25}{\node[boundary] (\n) at (\x,\y) {};}
        \foreach \u/\v in {b1/a1,b1/a4,b2/a1,b2/a2,b3/a2,b3/a3,b4/a3,b4/a7,a1/a2,a1/a4,a1/a5,a2/a3,a2/a5,a2/a6,a3/a6,a3/a7,a4/a5,a4/a8,a5/a6,a5/a8,a5/a9,a6/a7,a6/a9,a6/a10,a7/a10,a7/a11,a8/a9,a8/a12,a9/a10,a9/a12,a9/a13,a10/a11,a10/a13,a10/a14,a11/a14,a12/a13,a13/a14,b5/a11,b5/a14,b6/a13,b6/a14,b7/a8,b7/a12}{\draw[intedge](\u)--(\v);}
        \foreach \n/\x/\y in {r1/-3.10/1.55,r2/-3.28/.86,r3/-2.98/-1.48,r4/-2.45/-2.10,r5/2.88/1.68,r6/3.22/.95,r7/3.12/-.55,r8/2.65/-1.67,r9/.72/-2.65,r10/-1.05/2.68,r11/.62/2.70}{\node[reserv] (\n) at (\x,\y) {};}
        \foreach \u/\v in {b1/r1,b1/r2,b2/r10,b3/r11,b7/r3,b7/r4,b4/r5,b4/r6,b5/r7,b5/r8,b6/r9}{\draw[crossedge](\u)--(\v);}
        
        \node[font=\bfseries\large,text=Collective,fill=white,fill opacity=.92,text opacity=1,inner sep=2pt,rounded corners=1pt,anchor=west] at (2.50,.28) {$\partial\mathsf S$};
        \draw[-{Latex[length=1.5mm,width=1.0mm]},Collective!72,line width=.58pt] (2.40,.10)--(1.99,-.49);
        
        \end{scope}
        
        \begin{scope}[shift={(4.45,.05)}]
        \node[font=\normalsize,text=BlueGrey] at (0,3.3) {$3D$ network};
        \begin{scope}[on background layer]
        \path[fill=BlueGrey,fill opacity=.065,draw=none]
        (-2.10,-1.18)..controls(-2.62,-.10)and(-2.05,1.55)..(-.93,2.02)
        ..controls(.10,2.45)and(1.50,1.85)..(2.02,.78)
        ..controls(2.50,-.20)and(1.77,-1.57)..(.61,-1.92)
        ..controls(-.72,-2.28)and(-1.83,-1.90)..(-2.10,-1.18)--cycle;
        \draw[BlueGrey!42,line width=.55pt,dashed] (-1.77,.95) .. controls (-.95,1.67) and (.84,1.58) .. (1.74,.79);
        \draw[BlueGrey!42,line width=.55pt] (-1.92,-.18) .. controls (-.92,.60) and (.98,.56) .. (1.88,-.25);
        \draw[BlueGrey!42,line width=.55pt,dashed] (-1.53,-1.25) .. controls (-.52,-.55) and (.72,-.50) .. (1.46,-1.18);
        \draw[BlueGrey!30,line width=.45pt,dashed] (-1.57,1.48)..controls(-1.72,.25)and(-1.45,-.75)..(-1.03,-1.55);
        \draw[BlueGrey!30,line width=.45pt] (.18,2.10)..controls(.04,.78)and(.10,-.73)..(.25,-1.86);
        \draw[BlueGrey!30,line width=.45pt,dashed] (1.47,1.45)..controls(1.69,.45)and(1.46,-.65)..(1.05,-1.58);
        \end{scope}
        \begin{scope}[opacity=.42]
        \foreach \xa/\ya/\xb/\yb in {-1.2/1.18/-.25/1.50,-.25/1.50/.72/1.24,-1.2/1.18/-1.38/.26,-.25/1.50/-.35/.47,.72/1.24/.90/.30,-1.38/.26/-.35/.47,-.35/.47/.90/.30,-1.38/.26/-1.08/-.70,-.35/.47/-.25/-.62,.90/.30/.87/-.58,-1.08/-.70/-.25/-.62,-.25/-.62/.87/-.58}{\draw[Ink!28,line width=.48pt](\xa,\ya)--(\xb,\yb);}
        \foreach \x/\y in {-1.20/1.18,-.25/1.50,.72/1.24,-1.38/.26,-.35/.47,.90/.30,-1.08/-.70,-.25/-.62,.87/-.58}{\shade[ball color=white,draw=Ink!42,line width=.45pt] (\x,\y) circle (2.0mm);}
        \end{scope}
        \foreach \xa/\ya/\xb/\yb in {-1.60/.72/-.70/.95,-.70/.95/.22/.90,.22/.90/1.20/.62,-1.60/.72/-1.55/-.15,-.70/.95/-.62/.05,.22/.90/.26/.00,1.20/.62/1.34/-.25,-1.55/-.15/-.62/.05,-.62/.05/.26/.00,.26/.00/1.34/-.25,-1.55/-.15/-1.18/-1.08,-.62/.05/-.45/-1.00,.26/.00/.40/-.90,1.34/-.25/1.12/-1.02,-1.18/-1.08/-.45/-1.00,-.45/-1.00/.40/-.90,.40/-.90/1.12/-1.02}{\draw[Ink!22,line width=.5pt](\xa,\ya)--(\xb,\yb);}
        \foreach \x/\y in {-1.60/.72,-.70/.95,.22/.90,1.20/.62,-1.55/-.15,-.62/.05,.26/.00,1.34/-.25,-1.18/-1.08,-.45/-1.00,.40/-.90,1.12/-1.02}{\shade[ball color=white,draw=Ink!58,line width=.50pt] (\x,\y) circle (2.35mm);}
        \foreach \x/\y in {-1.78/1.37,-.42/1.96,.96/1.61,1.84/.82,1.89/-.63,.72/-1.62,-.63/-1.74,-1.72/-.95}{\shade[ball color=Collective!28,draw=Collective,line width=.92pt] (\x,\y) circle (2.75mm);}
        \foreach \xa/\ya/\xb/\yb in {-1.78/1.37/-1.60/.72,-.42/1.96/-.70/.95,-.42/1.96/.22/.90,.96/1.61/.22/.90,.96/1.61/1.20/.62,1.84/.82/1.20/.62,1.89/-.63/1.34/-.25,1.89/-.63/1.12/-1.02,.72/-1.62/.40/-.90,-.63/-1.74/-.45/-1.00,-1.72/-.95/-1.18/-1.08,-1.72/-.95/-1.55/-.15}{\draw[Ink!24,line width=.48pt](\xa,\ya)--(\xb,\yb);}
        \foreach \x/\y/\s in {-2.85/1.75/2.05,-3.05/.95/2.15,-2.92/-1.25/2.2,2.92/1.45/2.15,3.16/.57/2.25,2.95/-.95/2.3,.95/-2.58/2.25,-.52/2.86/2.0,1.78/2.42/2.05,-.92/-2.66/2.18}{\shade[ball color=Coarse!10,draw=Coarse!92!black,line width=.55pt] (\x,\y) circle (\s mm);}
        \foreach \xa/\ya/\xb/\yb in {-1.78/1.37/-2.85/1.75,-1.78/1.37/-3.05/.95,-.42/1.96/-.52/2.86,.96/1.61/1.78/2.42,-1.72/-.95/-2.92/-1.25,1.84/.82/2.92/1.45,1.84/.82/3.16/.57,1.89/-.63/2.95/-.95,.72/-1.62/.95/-2.58,-.63/-1.74/-.92/-2.66}{\draw[crossedge](\xa,\ya)--(\xb,\yb);}
        
        \node[font=\bfseries\large,text=Collective,fill=white,fill opacity=.92,text opacity=1,inner sep=2pt,rounded corners=1pt,anchor=west] at (2.48,.08) {$\partial\mathsf S$};
        \draw[-{Latex[length=1.5mm,width=1.0mm]},Collective!72,line width=.58pt] (2.38,.0)--(1.91,-.43);
        
        \end{scope}
        
      %  \draw[Ink!13,line width=.55pt] (-7.75,-3.08)--(7.75,-3.08);
     %   \node[row] at (0,-3.46) {$\displaystyle
     %  \mathcal S\!\big(\rho_N^{(K)}\big)-\mathcal S\!\big(\rho^{\otimes K}\big)
     %   \le |\partial\mathsf S|\,\big(\mathcal S(\q[\rho])-\mathcal S(\rho)\big)
     %   \le |\partial\mathsf S|\ln(\sss)$};
        \end{tikzpicture}
    }
    \caption{{\it Area Law in quantum networks.} The boundary $\partial \s$ of $\s$, \eqref{deltaS} is induced by the graph: it is the collection of vertices in $\s$ which are connected to $\r$. Two representations in which the network is viewed as being $2$- and $3$-dimensional. The von Neumann entropy of the state of $\s$ is called the entanglement entropy when the total network state is pure. The entanglement entropy may grow with $\s$. We show that as $N\rightarrow\infty$, the system entanglement entropy is bounded above by $|\partial \s|$, the {\it area of the boundary} of $\s$ (not its volume).}
    \label{fig:area-law}
\end{figure}

\begin{discussion}
The increase in entropy due to the interactions between sites is  bounded above by a term $|\partial \s|\ln(\sss)$, uniformly in $\rho$ and $N\ge N_*$. This upper bound is of the size of the {\it boundary} of $\s$, not its volume. For pure states $\rho$ the {\it entanglement entropy} 
$$
{\rm EE}(K)=\S\left(\lim_{N\rightarrow\infty}\rho_N^{(K)}\right)
$$ 
is a measure for the entanglement between the subsystem (region) $\s$ and the reservoir $\r$. Its value ranges from zero (no entanglement) to $\ln(d^K)=K\ln(d)$ (maximally entangled). As $\S(\rho)=0$ for pure $\rho$ we have the area law
$$
{\rm EE}(K) \le |\partial \s| \ln(\sss). 
$$
This inequality confirms the intuition that  entanglement is created by the edges of the graph linking $\s$ to $\r$. 
\end{discussion}

\subsection{Proof of Theorem \ref{thm:arealaw}} 
(1)\  From \eqref{2.32}, \eqref{m17} and the invariance of the von Neumann entropy under conjugation by unitaries, we have 
\begin{equation}
\label{2.25}
\S\big(\rho_N^{(K)}\big) = \S\Big( \E\big[ e^{-iT \sum_{j=K+1}^N \xi_j L(j)} (\rho^{\otimes K}) e^{iT \sum_{j=K+1}^N \xi_j L(j)}\big]\Big).
\end{equation}
As the operation being applied to $\rho^{\otimes K}$ on the right-hand side (inside $\mathcal S$) is unital, and unital channels cannot decrease entropy, we obtain the lower bounds in \eqref{2.27} and \eqref{res}.

Consider the CPTP map acting on $\b(\h_\s) =\b((\C^d)^{\otimes K})$ given by
\begin{equation}
\label{prodstruc}
\q_{\partial\s} = \bigotimes_{i\in\partial \s} (\q)_i,\qquad \q[\cdot] = \sum_{s=1}^\sss \Pi_s[\cdot]\Pi_s
\end{equation}
and where $(\q)_i$ is the single-site super-operator $\q$ acting nontrivially on the $i$th tensor product factor. $\q_{\partial \s}$ is unital (leaves $\mathbf 1$ invariant) and so it does not decrease the entropy,
\begin{equation}
\label{2.24-1}
\S\big(\rho_N^{(K)}\big) \le  \S\Big( \q_{\partial\s} \Big[\, \E\big[ e^{-iT \sum_{j=K+1}^N \xi_j L(j)} (\rho^{\otimes K}) e^{iT \sum_{j=K+1}^N \xi_j L(j)}\big]\Big] \Big).
\end{equation}
Denoting $|\partial\s|=m$ we have,
$$
\q_{\partial\s}[\cdot ] = \sum_{s_1,\ldots,s_m=1}^\sss (\Pi_{s_m})_{i_m}\cdots (\Pi_{s_1})_{i_1} \ [\cdot] \ (\Pi_{s_1})_{i_1}\cdots (\Pi_{s_m})_{i_m}.
$$
For any $s_1,\ldots,s_m$ and $i_1,\ldots,i_m\in \partial\s$ the $(\Pi_{s_1})_{i_1}\cdots (\Pi_{s_m})_{i_m}$ are eigenprojections of every operator $L(j)$ (see \eqref{eq:bndryop}). It follows that 
$$
\q_{\partial\s} \Big[ \, \E\big[ e^{-iT \sum_{j=K+1}^N \xi_j L(j)}Xe^{iT \sum_{j=K+1}^N \xi_j L(j)}\big] \Big]= \q_{\partial\s}[X]
$$
for any $X\in\b(\h_\s)$, so that \eqref{2.24-1} yields \eqref{2.27}. 
In the last step we have used  the product structure \eqref{prodstruc} and the additivity of the von Neumann entropy.

It remains to establish the bound
\begin{equation}
\label{2.28}
\S(\q[\rho])\le \S(\rho)+\ln(\sss)
\end{equation}
which, together with \eqref{2.27}, implies \eqref{res}.
As $\mathcal \q$ is a CPTP map (acting on bounded operators on $\C^d$, the single-vertex Hilbert space) it is the reduction of an augmented system (dilation) of the form \cite{MerkliNeira2026, Holevo} 
\begin{equation}
\label{6.9}
\q[X] = \Tr_\e\Big(U\big(X\otimes|\Omega\rangle\langle\Omega|\big)U^*\Big),\qquad X\in\b\big(\C^d\big),
\end{equation}
where $U$ is a unitary acting on $\C^d\otimes\h_\e$ with $\h_\e$ an `environment' or `ancilla' Hilbert space and $\Omega\in \h_\e$ is a normalized vector (`pure environment state'). The map $\q$ has the Kraus representation $\q[\cdot]=\sum_{s=1}^\sss\Pi_s[\cdot]\Pi_s$. We may therefore choose a Stinespring dilation with 
\begin{equation}
\label{dimHe}
\dim \h_\e\le  \sss.
\end{equation}
The operator 
$$
\rho_{\s\e} = U\big( \rho\otimes|\Omega\rangle\langle\Omega|\big)U^*
$$
is a density matrix on $\C^d\otimes\h_\e$ and from \eqref{6.9}, $\q[\rho]=\Tr_\e(\rho_{\s\e})$. Denote the other marginal of $\rho_{\s\e}$ as $\rho_\e=\Tr_\s(\rho_{\s\e})$. The Araki-Lieb inequality \cite{AL} gives
$$
\big|\S\big(\q[\rho]\big) - \S(\rho_\e)\big| \le \S(\rho_{\s\e}).
$$
As $\S(\rho_{\s\e})=\S(\rho)$ we obtain
$$
\S\big(\q[\rho]\big) \le   \S(\rho)+ \S(\rho_\e).
$$
The maximum of $\S(\rho_\e)$ is attained for $\rho_\e\propto \mathbf 1$ and takes the value $\ln(\dim \h_\e)\le\ln(\sss)$. This gives the bound \eqref{2.28}. 
\smallskip

(2) Now we show the optimality. For a well-connected graph satisfying the genericity condition \eqref{gencond} we know from Corollary \ref{cor:deco} and the continuity of the von Neumann entropy relative to the trace norm \cite{Audenaert2007}, that 
\begin{equation}
\label{2.29}
\lim_{N\rightarrow\infty}\S\big(\rho_N^{(K)}\big) =\S\big(\P[\rho^{\otimes K}]\big). 
\end{equation}
Take a graph such that for some $1\le m\le K$, every vertex $i\in\llbracket 1,m\rrbracket$ is connected to a number $\nu_i$, $i=1,\ldots,m$, distinct vertices $j\ge K+1$ and no such $j$ connects to two distinct $i$. Then $\partial\s=\{1,\ldots,m\}$ and $\L=\{A_1,\ldots,A_m\}$. The projections $P_\lambda$ are $\Pi_{s_1}\otimes\cdots \otimes \Pi_{s_m} \otimes \mathbf{1}^{\otimes (K - m)}$ with $\lambda=(s_1,\ldots,s_m)\in\llbracket 1,\sss\rrbracket^m$. Then we have $\P=\q_{\partial \s}$ (see \eqref{locres} and \eqref{prodstruc}). The system is well-connected if $\nu_i\rightarrow\infty$  as $N\rightarrow\infty$ for every $i$ (see \eqref{wellconn}). It follows from \eqref{2.29} that  $\lim_{N\rightarrow\infty}\S\big(\rho_N^{(K)}\big) =\S(\q_{\partial\s}[\rho^{\otimes K}]) = |\partial \s|\S(\q[\rho])+(K-|\partial\s|)\S(\rho)$. This concludes the proof of Theorem \ref{thm:arealaw}.\hfill\qed

\bigskip
\medskip

\section{Convergence of information-theoretic quantities}
\label{sec:infotheo}

\medskip

In the generic and well-connected regimes, Theorem \ref{thm2} shows the convergence and speed of convergence of the reduced density matrix,
\begin{equation}
\label{rhoinf}
\lim_{N\rightarrow\infty} \big\| \rho_N^{(K)}- \rho^{(K)}_\infty\big\|_1=0, \qquad  \rho^{(K)}_\infty= e^{-iTH_\s}\P\big[\rho^{\otimes K}\big] e^{iTH_\s}.
\end{equation}
In this section we investigate the implications for information-theoretic quantities.

\subsection{(Conditional) von Neumann entropy, mutual information}

We consider the von Neumann entropy $\S(\rho_N^{(K)})$ introduced in \eqref{vne} as well as the conditional von Neumann entropy and the mutual information induced by a bipartition of the system,
\begin{equation}
\s=\s_1\sqcup \s_2,\qquad |\s_1|\ge 1,\quad  |\s_2|=K-|\s_1|\ge 1.
\end{equation}
The conditional entropy and the mutual information of the state $\rho_N^{(K)}$ associated to the bipartition are defined by
\begin{align}
\S\big(\s_1|\s_2\big)_{\rho_N^{(K)}} &= \S\big(\rho_N^{(K)}\big) - \S\big(\Tr_{\s_1}(\rho_N^{(K)})\big)\quad \text{and}\\
I\big(\s_1:\s_2\big)_{\rho_N^{(K)}} &= \S\big(\Tr_{\s_2}(\rho_N^{(K)})\big) - \S\big( \s_1|\s_2\big)_{\rho_N^{(K)}},
\label{mutinfo}
\end{align}
respectively \cite{Winter2016, WVHC}.

\begin{proposition}
\label{prop:info}

For every $0< b<1$,
\begin{align}
\big|\S\big(\rho^{(K)}_N\big) -\S\big(\rho^{(K)}_\infty\big) \big| &\le   \frac{1}{2^b}\left(\ln(d^K-1) + \frac{1}{e^b(1-b)}\right) |\Lambda|^b q^{b\,\ve_K(N)}, \label{5.7}\\ 
\big|\S\big(\s_1|\s_2\big)_{\rho^{(K)}_N} -\S\big(\s_1|\s_2\big)_{\rho^{(K)}_\infty} \big| &\le \,  2^{1-b}\left(|\s_1|\ln d + \frac{1}{(2e)^b (1-b)}\right) |\Lambda|^b q^{b\,\ve_K(N)} \label{5.8}\\
\Big| I\big(\s_1:\s_2\big)_{\rho_N^{(K)}} - I\big(\s_1:\s_2\big)_{\rho_\infty^{(K)}}\Big| &\le\, \frac{1}{2^b} \left(3|\s_1|\ln(d)+ \frac{1+2^{1-b}}{e^b(1-b)} \right) |\Lambda|^b q^{b\,\ve_K(N)}.
\label{5.8-1}
\end{align}
\end{proposition}

In the generic regime where $q<1$ and for well-connected systems where $\ve_K(N)\rightarrow\infty$ as $N\rightarrow\infty$, the right-hand sides of \eqref{5.7}-\eqref{5.8-1} converge to zero exponentially in the effective connection number. 

\begin{example}
In the case of Theorem \ref{thm3} (locally resolved boundary operators) we have $\rho_\infty^{(K)}=\tau^{\otimes K}$, \eqref{DeFinettifinalstate}. Then $\S\big(\rho_\infty^{(K)}\big) = K \S(\tau)$. The entropy of $\tau$ vanishes if and only if $\tau$ is a pure state, which in turn is equivalent to $\rho=|\psi\rangle\langle\psi|$ with a normalized vector $\psi\in\C^d$ in the range of one of the eigenprojections $\Pi_s$. As discussed earlier, this is a trivial case in the sense that the interaction via the operators $A$ does not induce any correlation in the initial state $\rho^{\otimes N}$, \eqref{m2}. In the non-trivial case we have $\S(\tau)>0$ and the conditional entropy and mutual information in the final state are $\S\big(\s_1|\s_2\big)_{\rho^{(K)}_\infty}=|\s_1|\S(\tau)$ and $I\big(\s_1:\s_2\big)_{\rho_\infty^{(K)}}=0$. The correlations between disjoint parts in $\s$, measured by the mutual information, are thus washed out by the interaction with the large reservoir ($N\rightarrow\infty$).
\end{example}

{\bf Proof of Proposition \ref{prop:info}.} 
The Fannes-Audenaert inequality (\cite[Theorem~1]{Audenaert2007}; while they use $\log_2$ we use $\ln$) gives 
\begin{equation}
\label{5.4}
\big|\S\big(\rho^{(K)}_N\big) -\S\big(\rho^{(K)}_\infty\big) \big| \le \mathcal T \ln(d^K-1) + h(\mathcal T),
\end{equation}
where $d^K$ is the dimension of the Hilbert space the density matrices act on, where 
\begin{equation}
\label{T}
\mathcal T = \frac12 \|\rho^{(K)}_N-\rho^{(K)}_\infty\|_1 \in[0,1]
\end{equation}
and where $h$ is the binary entropy
\begin{equation}
\label{h}
h(x) = -x\ln(x)-(1-x)\ln(1-x), \qquad x\in[0,1].
\end{equation}
For $0\le x<1$ we have  $-\ln(1-x)=\int_0^x\frac{dt}{1-t}\le \frac{x}{1-x}$ and for $0< b<1$, $0<x<1$ we get $h(x) \le x(1-\ln(x)) = x^b f_b(x)$ with $f_b(x) = x^{1-b}(1-\ln(x))$. The maximum of $f_b$ is $e^{-b}/(1-b)$, attained at $x=e^{-b/(1-b)}$. Thus,
\begin{equation}
\label{5.5}
h(x) \le \frac{x^b}{e^b(1-b)},
\end{equation}
which extends to all of $x\in[0,1]$.  

As $\mathcal T\le 1$ and $0\le b<1$ we have $\mathcal T\le \mathcal T^b$, so that combining \eqref{5.4} and \eqref{5.5} yields,
\begin{equation}
\label{any}
\big|\S\big(\rho^{(K)}_N\big) -\S\big(\rho^{(K)}_\infty\big) \big| \le   2^{-b}\left(\ln(d^K-1) + \frac{1}{e^b(1-b)}\right) \left\|\rho^{(K)}_N-\rho^{(K)}_\infty\right\|_1^b.
\end{equation}
The bound \eqref{5.7} follows from Theorem \ref{thm2}.

To show \eqref{5.8} we use the Alicki-Fannes-Winter inequality (\cite[Lemma~2]{Winter2016}),
\begin{equation}
\left|\S\big(\s_1|\s_2\big)_{\rho^{(K)}_N} -\S\big(\s_1|\s_2\big)_{\rho^{(K)}_\infty} \right| \le  2\mathcal T \ln(\dim \h_{\s_1}) +(1+\mathcal T)h\left(\frac{\mathcal T}{1+\mathcal T}\right)
\end{equation}
with $\mathcal T$ and $h$ given in \eqref{T} and \eqref{h}. As $1+\mathcal T\le 2$ we obtain from \eqref{5.5},
$$
(1+\mathcal T) h\left(\frac{\mathcal T}{1+\mathcal T}\right)\le \frac{2^{1-b}}{e^b(1-b)} \mathcal T^b
$$
and as $\mathcal T\le \mathcal T^b$ and $\dim \h_{\s_1} = d^{|\s_1|}$,
\begin{equation}
\label{5.13}
\left|\S\big(\s_1|\s_2\big)_{\rho^{(K)}_N} -\S\big(\s_1|\s_2\big)_{\rho^{(K)}_\infty} \right| \le  2\left(|\s_1|\ln d + \frac{1}{(2e)^b (1-b)}\right) \mathcal T^b.
\end{equation}
Then \eqref{5.8} follows from Theorem \ref{thm2}.

From \eqref{mutinfo} and \eqref{5.13} we have 
\begin{align}
\Big| I\big(\s_1:\s_2\big)_{\rho_N^{(K)}} &- I\big(\s_1:\s_2\big)_{\rho_\infty^{(K)}}\Big|\nonumber\\
\le &\  \Big| \S\big(\Tr_{\s_2}(\rho_N^{(K)})\big)-\S\big(\Tr_{\s_2}(\rho_\infty^{(K)})\big)\Big|  + 2\left(|\s_1|\ln d + \frac{1}{(2e)^b (1-b)}\right) \mathcal T^b.  
\label{5.15}
\end{align}
Furthermore, as the bound \eqref{any} holds for any pair of density matrices, 
\begin{align}
 \Big| \S\big(\Tr_{\s_2}(\rho_N^{(K)})\big)&-\S\big(\Tr_{\s_2}(\rho_\infty^{(K)})\big)\Big| \nonumber\\
 &\le   2^{-b}\left(\ln(d^{|\s_1|}-1) + \frac{1}{e^b(1-b)}\right) \big\|\Tr_{\s_2}\big( \rho^{(K)}_N\big)-\Tr_{\s_2}\big(\rho^{(K)}_\infty\big)\big\|_1^b.
 \label{5.16}
\end{align}
As $\big\|\Tr_{\s_2}\big( \rho^{(K)}_N\big)-\Tr_{\s_2}\big(\rho^{(K)}_\infty\big)\big\|_1\le \big\| \rho^{(K)}_N-\rho^{(K)}_\infty\big\|_1$ we obtain the following bound from \eqref{5.15}, \eqref{5.16},
\begin{equation}
\Big| I\big(\s_1:\s_2\big)_{\rho_N^{(K)}} - I\big(\s_1:\s_2\big)_{\rho_\infty^{(K)}}\Big|\le \left(\ln(d^{|\s_1|} - 1) + 2|\s_1|\ln(d)+ \frac{1+2^{1-b}}{e^b(1-b)} \right) \mathcal T^b.  
\label{5.15'}
\end{equation}
Using $\ln(d^{|\s_1|} - 1) \le |\s_1|\ln(d)$, the bound \eqref{5.8-1} follows from \eqref{5.15'} and Theorem \ref{thm2}. This concludes the proof of Proposition \ref{prop:info}.\hfill \qed

\subsection{Entanglement between $\s$ and $\r$}

We consider the bipartition of the network into $\s\sqcup\r$ where $|\s|=K$ and $|\r|=N-K$ with $K$ fixed and  $N\rightarrow\infty$. If the state of the total system  $\rho_N$, \eqref{m2} is pure, then the 
(generalized) concurrence,
\begin{equation}
\label{genconc}
C\big(\rho^{(K)}_N\big) = \sqrt{2\left(1-\Tr\left\{\left(\rho^{(K)}_N\right)^2\right\}\right)},
\end{equation}
is a good measure of the entanglement between $\s$ and $\r$. The state $\rho_N$ is pure if and only if the factor state $\rho$ is pure,
\begin{equation}
\label{rhopure}
\rho=|\psi\rangle\langle\psi|.
\end{equation}

\begin{proposition}
\label{concurrence convergence}
For network states $\rho_N$ defined in \eqref{m2} with arbitrary initial single-site state $\rho$ (pure or mixed), we have 
\begin{equation}
\label{3.17}
\big| C\big(\rho^{(K)}_N\big)-C\big(\rho^{(K)}_\infty\big)\big|\le \sqrt{2|\Lambda|}\  q^{\ve_K(N)/2}.
\end{equation}
\end{proposition}
\bigskip

{\bf Proof of Proposition \ref{concurrence convergence}.} We have 
\begin{align}
\label{5.21}
\Big| C\big(\rho^{(K)}_N\big)^2-C\big(\rho^{(K)}_\infty\big)^2\Big| = 2\Big| \|\rho_N^{(K)}\|_2^2 - \|\rho_\infty^{(K)}\|_2^2\Big|,
\end{align}
where $\|X\|_2 = \sqrt{\Tr (X^*X)}$ is the Hilbert-Schmidt norm. Now for density matrices  $\sigma$, $\tau$ we have
\[
\big|\Tr(\sigma^2) - \Tr(\tau^2)\big| = \big|\Tr((\sigma - \tau)(\sigma + \tau))\big| = \big|\Tr((\sigma - \tau)(\sigma + \tau - 1))\big|
\] 
as $\Tr(\sigma - \tau) = 0$. Then from H\"older's inequality,
\[
\big|\Tr(\sigma^2) - \Tr(\tau^2)\big| \le \|\sigma - \tau\|_1\|\sigma + \tau - 1\|_\infty \le \|\sigma - \tau\|_1
\] as $-1 \le \sigma + \tau - 1 \le 1$. Then \eqref{5.21} gives, $\big| C\big(\rho^{(K)}_N\big)^2-C\big(\rho^{(K)}_\infty\big)^2\big| \le 2\big\|\rho_N^{(K)}- \rho_\infty^{(K)}\big\|_1$. 
Next for $x,y\ge 0$ we have $|x^2-y^2|\ge (x-y)^2$, so $\big| C\big(\rho^{(K)}_N\big)-C\big(\rho^{(K)}_\infty\big)\big| \le \sqrt 2 \big\|\rho_N^{(K)}- \rho_\infty^{(K)}\big\|_1^{1/2}$. The bound \eqref{3.17} follows from Theorem \ref{thm2}. 
\hfill \qed

\medskip

%\color{red}
%Moreover, I'm pretty sure 2 is optimal. Consider \[
%    \sigma = \begin{pmatrix}
%        1 & 0\\
%        0 & 0
 %   \end{pmatrix}, \quad \tau_\ve = \begin{pmatrix}
%        1 - \ve & 0\\
%        0 & \ve
%    \end{pmatrix}.
%\]
%Then \[
%    \|\sigma - \tau_\ve\|_1 = 2\eps,
%\] while \[
%    |\Tr \sigma^2 - \Tr \tau_\ve^2| = 2\ve(1 - \ve).
%\] Thus \[
%    \frac{|C(\sigma)^2 - C(\tau_\ve)^2|}{\|\sigma - \tau_\ve\|_1} = 2(1 - \ve) \to 2
%\] as $\ve \to 0$.
%}

\begin{discussion}
In the generic and well-connected regimes \eqref{gencond} and  \eqref{wellconn}, the right-hand side of \eqref{3.17} converges to zero as $N\rightarrow\infty$. From \eqref{defP} and \eqref{rhoinf} we obtain for pure states \eqref{rhopure},
\begin{equation}
\label{3.19}
\Tr\, (\rho^{(K)}_\infty)^2 = \Tr\, (\P[\rho^{\otimes K}])^2 = \sum_{\lambda\in\Lambda} \big\| P_\lambda (\psi \otimes\cdots\otimes \psi)\big\|^4 \in [0,1].
\end{equation}
This quantity is equal to one and only if $\psi\otimes\cdots\otimes\psi$ is an eigenvector of exactly one of the $P_\lambda$, that is, if and only if for some (unique) $\lambda\in\Lambda$,
\begin{equation}
\label{3.20}
P_\lambda (\psi\otimes\cdots\otimes\psi)=\psi\otimes\cdots\otimes\psi.
\end{equation}
This in turn is equivalent to $\P[\rho^{\otimes K}]=\rho^{\otimes K}$ or, $\rho_\infty^{(K)}= e^{-iTH_\s} (\rho^{\otimes K}) e^{iTH_\s}$. In this case, and only in this case, the asymptotic state of the system ($N\rightarrow\infty$) is the same as the one obtained by deleting in $\rho_N$ all interactions of $\s$ with $\r$ (eliminating the $\s\r$ connections in the graph), meaning that the $\s\r$ coupling has no effect. 
Then the concurrence \eqref{genconc} vanishes. For all other states $\psi$ (the ones not satisfying \eqref{3.20}) we have $C(\rho_\infty^{(K)})>0$. We conclude that generically, unless the $\s\r$ interaction leaves the initial product state unaffected in the sense of \eqref{3.20}, the system is entangled with the reservoir.
\end{discussion}

\begin{example}  In the case of locally resolved boundary operators we have $\rho^{(K)}_\infty=\tau^{\otimes K}$ (see Theorem \ref{thm3}) and 
$$
C(\rho^{(K)}_\infty) = \sqrt{2\big(1-(\Tr\, \tau^2)^K\big)}, \qquad \Tr\,\tau^2 = \sum_{s=1}^\sss \|\Pi_s\psi\|^4\in [0,1],
$$
where the last equality holds for pure $\rho$ \eqref{rhopure}. $\Tr\, \tau^2=1$ if and only if $\Pi_s\psi=\psi$ for an $s$, that is, if and only if $\psi$ is an eigenvector of the operator $A$. In this case the coupling between $\s$ and $\r$ has no effect, $\rho_N=\rho^{\otimes N}$ (see \eqref{m2}) and $C(\rho^{(K)}_\infty) =0$. Whenever $\psi$ is not an eigenvector of $A$ the concurrence $C(\rho^{(K)}_\infty)$ is strictly positive and increases with the system size $K$ and approaches $\sqrt 2$ as $K \rightarrow \infty$. 
\end{example}

\subsection{Intra-system entanglement: Qubits}

We now consider $\s$ to consist of two qubits, $K=d=2$, and we examine the intra-system entanglement between the qubits. The entanglement for pure or mixed two-qubit states is conveniently measured by the Wootters concurrence which we define now.

Let $\rho$ be a two-qubit density matrix acting on the Hilbert space $\C^2\otimes\C^2$ and define the positive semidefinite operator 
\begin{equation}
\label{R}
R(\rho) = \sqrt{\sqrt\rho \tilde\rho \sqrt\rho},\qquad \tilde\rho = (\sigma^y\otimes\sigma^y)\bar\rho (\sigma^y\otimes\sigma^y) 
\end{equation}
where $\sigma^y=\begin{pmatrix}0 & -i\\ i & 0\end{pmatrix}$ is the Pauli $y$ matrix and $\bar\rho$ is the entry-wise complex conjugate of the density matrix $\rho$ when written in the standard product eigenbasis of the two local $\sigma^z$ operators, where $\sigma^z=\begin{pmatrix}1 & 0\\ 0 & -1\end{pmatrix}$ is the Pauli $z$ matrix.  The {\it Wootters concurrence} of the state $\rho$, introduced in \cite{Wootters_1998}, is given by 
\begin{equation}
\cw(\rho) = \max\big\{0, \lambda_1-\lambda_2-\lambda_3-\lambda_4\},
\end{equation}
where $\lambda_1\ge\lambda_2\ge\lambda_3\ge\lambda_4\ge 0$ are the ordered eigenvalues of $R(\rho)$. The concurrence takes values in $[0,1]$. It vanishes exactly for separable states $\rho$ and equals $1$ for so-called maximally entangled states,  such as Bell states. The Wootters concurrence is a monotonic function of the entanglement of formation \cite{Bennet} and it is thus a good measure of entanglement between two qubits. 

\begin{theorem}[Wootters concurrence between two qubits forming $\s$]
\label{thm:qubit conc}
We have 
\begin{equation}
\left|\cw\big(\rho_N^{(2)}\big) -\cw\big(\rho_\infty^{(2)}\big)\right| \le 2\sqrt{|\Lambda|} \, q^{\ve_2(N)/2}.
\end{equation}
\end{theorem}
Theorem \ref{thm:qubit conc} follows directly from Theorem \ref{thm2} and Lemma \ref{lem:concurrence} below, which is of independent interest.

\begin{example}
For the qubit model \eqref{A} the cardinalities of $\Lambda$ and the effective connection numbers $\ve_2(N)$ can be taken from Tables \ref{table:guiding example 1} and \ref{table:guiding example 2}. The final states $\rho_\infty^{(2)}$ are given by the right-hand sides of \eqref{rh1}-\eqref{lh3} and we have $\cw\big(\rho_\infty^{(2)}\big)=0$ as those states are separable. 
\end{example}

\begin{lemma}[$1/2$-H\"older continuity of Wootters concurrence]
\label{lem:concurrence}
We have for any density matrices $\rho,\sigma$ of two qubits, 
\begin{equation}
\label{ineq}
\big|\cw(\rho)-\cw(\sigma)\big| \le 2\|\rho-\sigma\|^{1/2}_1.
\end{equation}   
\end{lemma}

{\bf Proof of Lemma \ref{lem:concurrence}.} Denote by $\lambda_j(\rho)$, $\lambda_j(\sigma)$ the eigenvalues of $R(\rho)$, $R(\sigma)$. As $|\max\{0,a\}-\max\{0,b\}|\le |a-b|$ we have,
\begin{align}
\label{3.18}
\big|\cw(\rho) - \cw(\sigma)\big|  \le   \sum_{j = 1}^4|\lambda_j(\rho) - \lambda_j(\sigma)|.
\end{align}
Now $R(\rho)^2= B(\rho)B(\rho)^*$ for
$$
B(\rho) = \sqrt\rho \,Y\sqrt{\bar\rho},\qquad Y=\sigma^y\otimes\sigma^y.
$$
This means that the $\lambda_j(\rho)$ are the singular values of the matrix $B(\rho)$. We can then use Mirsky's inequality \cite[Theorem~3.4.5]{HoJo} and \cite{mirsky1960symmetric}, to conclude that  
\begin{equation}
\label{eq:mirskys'}
\sum_{j = 1}^4|\lambda_j(\rho) - \lambda_j(\sigma)|\le \sum_{j = 1}^4 \lambda_j(\rho,\sigma) = \|B(\rho) - B(\sigma)\|_1,
\end{equation}
where $\lambda_j(\rho,\sigma)$ are the singular values of $B(\rho)-B(\sigma)$. Then as 
\begin{align}
\|B(\rho) - B(\sigma)\|_1 = & \left\|\sqrt \rho\, Y \sqrt{\bar \rho} - \sqrt \sigma Y \sqrt{\bar \sigma}\right\|_1\nonumber\\
= & \left\|(\sqrt \rho - \sqrt \sigma)Y\sqrt{\bar \rho} + \sqrt \sigma Y\left(\sqrt{\bar \rho} - \sqrt{\bar \sigma}\right)\right\|_1\nonumber\\
\le & \left\|(\sqrt \rho - \sqrt \sigma)Y\sqrt{\bar \rho}\right\|_1 + \left\|\sqrt \sigma Y\left(\sqrt{\bar \rho} - \sqrt{\bar \sigma}\right)\right\|_1,
\label{3.22}
\end{align}
applying H\"older's inequality, 
\begin{align}
\|B(\rho) - B(\sigma)\|_1 \le & \left\|\sqrt \rho -  \sqrt \sigma\right\|_2\left\|Y\sqrt{\bar \rho}\right\|_2 + \left\|\sqrt \sigma Y\right\|_2\left\|\sqrt{\bar \rho} - \sqrt{\bar \sigma}\right\|_2\nonumber\\
= & \left\|\sqrt \rho -  \sqrt \sigma\right\|_2\left\|\sqrt{\bar \rho}\right\|_2 + \left\|\sqrt \sigma \right\|_2\left\|\sqrt{\bar \rho} - \sqrt{\bar \sigma}\right\|_2\nonumber\\
= & \left\|\sqrt \rho -  \sqrt \sigma\right\|_2\sqrt{\Tr(\rho)} + \sqrt{\Tr(\sigma)}\left\|\sqrt{\bar \rho} - \sqrt{\bar \sigma}\right\|_2\nonumber\\
= & \left\|\sqrt \rho -  \sqrt \sigma\right\|_2 + \left\|\sqrt{\bar \rho} - \sqrt{\bar \sigma}\right\|_2\nonumber\\
= & 2\left\|\sqrt \rho -  \sqrt \sigma\right\|_2.
\label{3.27}
\end{align}
Finally, by the Powers-St\o rmer inequality \cite[Lemma~4.1]{powers1970free} we have
\begin{equation}
\label{3.22-1}
\left\|\sqrt \rho -  \sqrt \sigma\right\|_2^2\le \|\rho-\sigma\|_1.
\end{equation}
Combining \eqref{3.18}-\eqref{3.22-1}
yields \eqref{ineq}.\hfill \qed

\bigskip
\smallskip

\printbibliography

\end{document}